\documentclass[a4paper]{iopart}        

\usepackage{amssymb}

\expandafter\let\csname equation*\endcsname\relax
\expandafter\let\csname endequation*\endcsname\relax 

\usepackage{amsmath}
\usepackage{amsthm} 
\usepackage{array} 
\usepackage{graphicx}
\graphicspath{{./pdf_figures/}}

\usepackage{multirow}  
\usepackage{float}  
\usepackage{longtable}

\usepackage[colorlinks,citecolor=blue,urlcolor=blue]{hyperref}
\def\bc{{\mathbb{C}}}

\def\bn{{\mathbb{N}}}
\def\br{{\mathbb{R}}} 

\def\bz{{\mathbb{Z}}}

\def\vs{\vskip.3cm}
\def\noi{\noindent}
\def\q{\quad}

\newcommand{\END}{\hfill\mbox{\raggedright$\Diamond$}}

 \newtheorem{thm}{Theorem}[section]
\newtheorem{lem}[thm]{Lemma}
\newtheorem{cor}[thm]{Corollary}
\newtheorem{Def}[thm]{Definition} 
\newtheorem{rmk}[thm]{Remark}

\newtheorem{ex}[thm]{Example}

\newtheorem{proposition}[thm]{Proposition}

\def\sig{\sigma}
\def\ome{\omega}
\def\lam{\lambda}

\def\a{\alpha}

\def\ka{\kappa}

\newcommand{\UU}{\mbox{$\mathcal U$}}
\newcommand{\VV}{\mbox{$\mathcal V$}}
\newcommand{\LL}{\text{L1}}
\newcommand{\CC}{\text{C2}}
\newcommand{\MM}{\mbox{$\mathcal M$}}
\newcommand{\NN}{\mbox{$\mathcal N$}}
\newcommand{\FF}{\mbox{$\mathcal F$}}

\def\iso{\text{\,Iso\,}}
\def\orb{\text{\,Orb\,}} 
\def\Ome{\Omega}
\def\gmdeg{\Gamma\text{\rm -Deg\,}}
\def\g1deg{\Gamma\times S^1\text{\rm -Deg\,}}
\def\deg{\text{\,deg\,}}
\def\fix{\text{\,Fix\,}}
\def\si{\text{\tiny $\Sigma$}}
\def\cir{\text{\,circ\,}}

\newcommand{\id}{\text{\rm Id}}
\def\ome{\omega}

\begin{document}
\vspace*{-1.5cm}
\title[Symmetry Analysis of 
Coupled Systems under Time Delay]{\centering Symmetry Analysis of 
Coupled Scalar Systems under Time Delay}

\bigskip\bigskip
\noindent\emph{Preprint.} Final version in 
\href{http://dx.doi.org/10.1088/0951-7715/28/3/795}{\emph{Nonlinearity} 28(3):795--824, 2015}, 
copyright 2015 IOP Publishing Ltd \& London Mathematical Society.
\bigskip

\hspace*{-2.5cm}
\begin{minipage}{15cm}

\author{Fatihcan M. Atay} 
\address{Max Planck Institute for Mathematics in the Sciences,
 D-04103 Leipzig, Germany\\
 \emph{Current address:} Department of Mathematics, Bilkent University, 06800 Ankara, Turkey.}
 \ead{f.atay@bilkent.edu.tr}
 \author{Haibo Ruan}
 \address{
 Universit\"at Hamburg,
 Fachbereich Mathematik,
 D-20146 Hamburg, Germany\\
 \emph{Current address:} Institute of Mathematics, Technical University of Hamburg, D-21073 Hamburg, Germany}
\ead{haibo.ruan@tuhh.de}

\begin{abstract}
We study systems of coupled units in a general network configuration with a coupling delay. We show that the destabilizing bifurcations from an equilibrium are governed by the extreme eigenvalues of the coupling matrix of the network. Based on the equivariant degree method and its computational packages, we perform a symmetry classification of destabilizing bifurcations in bidirectional rings of coupled units. Both stationary and oscillatory bifurcations are discussed. We also introduce the concept of secondary dominating orbit types to capture bifurcating solutions of submaximal nature. 
\vskip.5cm
\noi {\bf Keywords:} Symmetries, bifurcation,  coupled network, coupling delay, equivariant degree,  spatiotemporal pattern.
\end{abstract} 

\vspace{-2em}

\end{minipage}



\section{Introduction}

Networks of coupled systems are known to be capable of a wide range of interesting dynamics, especially in the presence of time delays \cite{ctds-book10}. 
One of the most well-studied types of behavior involves synchronization of oscillations in various forms \cite{pikovsky-book01},
and recent work has revealed more complicated activity patterns related to the synchronization. 
As an example, which has also been a motivation for the present paper, we mention the so-called {chimera states}: in a network of identical phase oscillators arranged on a ring, with each oscillator coupled to a fixed number of its spatial neighbors,
appropriate conditions can lead to oscillators splitting into two contiguous groups, 
one group oscillating synchronously while the other one incoherently \cite{Kuramoto02}, a behavior which has also been reported in the presence of time delays \cite{PRL08}. 
An important feature is that such states are observed for identical oscillators and under homogeneous coupling conditions, i.e., in highly symmetric situations.
Thus, a natural question arises as to 
{the} relation of system's symmetries to its possible dynamical states.
The aim of this paper is to present a systematic 
investigation
of the types of dynamics that can be deduced from the symmetries and bifurcations of coupled scalar systems under a time delay.  
 

We consider $n$ identical dynamical systems governed by the equation $\dot{x_i}=f(x_i)$ for $i=1,\dots,n$, coupled in a general network configuration: 
\begin{equation}
	\dot{x_i}(t) = f(x_i(t)) + \kappa g_i(x_1(t-\tau),x_2(t-\tau),\dots,x_n(t-\tau)), \quad i=1,2,\dots,n.
\label{eq:1}  
\end{equation}
Here $x_i \in \mathbb{R}$, the function $g_i$ describes the interaction among the coupled units, and $\tau\ge0$ is the time delay. 
The scalar $\kappa>0$ denotes the coupling strength; it can of course be subsumed into the definition of $g_i$, but it is sometimes used as a  bifurcation parameter when one studies the effects of coupling in comparison to the intrinsic dynamics $f$, or for distinguishing excitatory from inhibitory coupling by simply changing its sign. 
The functions $f:\mathbb{R}\to\mathbb{R}$ and $g_i:\mathbb{R}^n\to\mathbb{R}$ are assumed to be continuously differentiable; 
in addition, $g_i$ will be assumed to be  {equivariant} when we consider symmetry. We also assume that $f$ and the $g_i$ vanish at the origin; hence the zero solution is an equilibrium solution of \eqref{eq:1}. The local stability of the zero solution is given by the linear variational equation 
\begin{equation}
	\dot{y}(t) = f'(0) y(t)+ \kappa C y(t-\tau),\q y\in\br^n,
\label{lin-eq}
\end{equation}
{where the {\it coupling matrix} $C=[c_{ij}]:=[{\partial g_i(0)}/{\partial x_j}]$  is assumed to be a symmetric matrix}. 

Systems of the form \eqref{eq:1} include many well-known examples as special cases. For instance, the \emph{neural network} model
\begin{equation} 
	\dot{x_i}(t) = -x_i(t) + g \left( \sum_{j=1}^{n} a_{ij} x_j(t-\tau) \right),
\label{neural}
\end{equation}
where $g$ is a sigmoidal function and $a_{ij}\in\mathbb{R}$ are entries of the (weighted and directed) adjacency matrix $A$ that describes the coupling among the neurons, has the form of \eqref{eq:1}. 
The component $a_{ij}$ describes how strongly the $j$th neuron influences the $i$th one; the influence being excitatory if $a_{ij}>0$ and inhibitory if $a_{ij}<0$. Often one excludes self-coupling, taking $a_{ii}=0$ $\forall i$. 
Linearization of \eqref{neural} about the zero solution yields the form \eqref{lin-eq} with $\kappa=g'(0)$ and  $C=A$.
There are also some variant models which are not strictly in the form \eqref{eq:1}, for instance \emph{pulse-coupled} systems,
\begin{equation}
	\dot{x_i}(t) = f(x_i(t)) + h(x_i(t)) \cdot g \left( \sum_{j=1}^{n} a_{ij} x_j(t-\tau) \right),
\label{puls-coup}
\end{equation}
where the influence of the network on the $i$th unit may be different depending on the state of the $i$th unit at that particular time instant. 
Although \eqref{puls-coup} is not of form \eqref{eq:1} for nonzero $\tau$,
 its linearization is still given by \eqref{lin-eq} with $\kappa=h(0)g'(0)$ and $C=A$.
This is a crucial observation since many of our bifurcation results will depend only on the linear part \eqref{lin-eq} of the model at hand, and thus will also apply to \eqref{puls-coup} in particular.
Also, the term ``coupling matrix" is used in a general way that can take other familiar forms in applications. For example,
models of synchronization typically involve \emph{diffusive-type interactions}, e.g.,
\begin{equation}
	\dot{x_i}(t) = f(x_i(t)) +  \sum_{j=1}^{n} a_{ij} g (x_j(t-\tau)-x_i(t-\tau)).
\label{eq-b}
\end{equation}
Sometimes the order of summation and the function $g$ are interchanged; in this case, \eqref{eq-b} becomes a special case of \eqref{neural} obtained by setting $a_{ii}=-\sum_{j\neq i}a_{ij}$ $\forall i$ in \eqref{neural}. The linear variational equation corresponding to \eqref{eq-b} has the form \eqref{lin-eq} with $C$ given by the negative of the Laplacian matrix:  $C=-L=A-D$, where $D=\mathrm{diag}\{k_1,\dots,k_n\}$ is the diagonal matrix of vertex in-degrees $k_i=\sum_{j\neq i}a_{ij}$.
If the delay pertains only to the interaction between different units (so that there are no self-delays), then one obtains a slightly variant system 
\begin{equation}
	\dot{x_i}(t) = f(x_i(t)) + \sum_{j=1}^{n} a_{ij} g (x_j(t-\tau)-x_i(t)),
\label{eq-c}
\end{equation}
whose linearization can be put into the form \eqref{lin-eq} (with the identification $f(x_i) \to f(x_i)+g'(0)k_i x_i$) provided that all vertices have the same degree $k_i=k$.
{In all cases, the matrix $C$ being symmetric reflects the assumption of bi-directional interactions in the network.}

{We are interested in the effects of the time delay and the spectrum of $C$ in causing the zero equilibrium to lose its stability 
as the system bifurcates into other dynamical states.  We show that, among all the eigenvalues of $C$, only the extreme eigenvalues (i.e., the smallest and largest ones) play a role in destabilizing the zero equilibrium. 
Networks having the same extreme eigenvalues 
will exhibit the same destabilizing behavior, independently of the precise network configuration. 
In the presence of ring (dihedral) symmetry, we give a complete classification of bifurcating states using equivariant degree methods.} 
The ring configuration is motivated by the setting of chimera states mentioned in the first paragraph, although  we do not focus on chimeras in this paper but aim to capture all emergent dynamics that can bifurcate from the equilibrium.
To illustrate the theoretical results, we use throughout a ring of size $12$ with various coupling configurations, since on the one hand it gives rise to a large variety of dynamical patterns, and on the other hand, it can be presented in a manageable size.
    
{Equivariant bifurcations have been extensively studied using various techniques such as those based on singularity theory, as developed by Golubitsky, Stewart \etal (cf. \cite{GSS,GS_symm}), geometric techniques developed by Field, Richardson \etal (cf. \cite{FR1,FR2,F}), constructive methods using algebraic geometry developed by Bierstone, Milman \etal (cf. \cite{BM}) and, last but not least, equivariant degree methods developed by Ize, Krawcewicz \etal (cf. \cite{IV-B,AED}). While geometric methods are based on generic approximations, topological methods rely on deformations  subject to homotopy invariance. Since homotopy generally allows ``larger'' changes of  maps compared to approximations, topological methods tend to overlook finer properties of solutions  such as stability, but rather they catch the existence. This is also why topological methods are commonly believed to produce ``weaker'' results. On the other hand,  
they can deal with non-generic situations just as well as generic cases.  Results hold without generic assumptions since homotopy makes essentially no distinction between generic and non-generic maps. An example of using topological indices to predict non-generic global  equivariant bifurcations with least symmetry can be found in \cite{F1}. To compare, the index in \cite{F1} is defined for a subgroup pair $(K,H)$ of the symmetry group $\Gamma$, where $K$ is normal in $H$ and $H/K$ is finite cyclic, while the invariant we use later is defined for $\Gamma$ and captures every subgroup; the index in [F1] is used to predict global bifurcation with at least symmetry $K$, while the invariant we use predicts local bifurcation with precise symmetry for all adequate subgroups in $\Gamma$ (cf. Corollary \ref{cor:steady} and Proposition \ref{cor:Hopf}); last but not least, while computations of the index in \cite{F1} can be technically involved, the invariant we use can be computed instantly.

In this paper we use equivariant degree methods to give a complete classification of bifurcating branches of solutions according to their symmetry properties. This includes dealing with non-simple critical eigenvalues with non-simple representations in the kernel of linearization. An additional advantage of using equivariant degree is that it can be  effectively calculated. } 
  
The computational tool we use for exact calculations of equivariant degrees is the ``Equivariant Degree Maple$^\text{\copyright}$ Library Package''\footnote[1]{The Equivariant Degree Maple$^\text{\copyright}$ Library Package was created by Biglands and Krawcewicz at the University of Alberta in 2006 supported by NSERC research grant.
It is open source and can be freely downloaded, for example, from {\tt http://www.math.uni-hamburg.de/home/ruan/download}.} (EDML). More precisely, to a given bifurcating equilibrium, one associates a {\it bifurcation invariant} in form of an equivariant degree. {Here, the term ``invariant'' refers to the fact that the bifurcation invariant remains constant against all adequate homotopies. The precise value of the  bifurcation invariant, once calculated, carries the full topological information about the bifurcating solutions and} gives rise to a complete classification of bifurcating branches by their symmetry properties. The calculation task of bifurcation invariants for a given symmetry group $\Gamma$ is taken over by 
the EDML package.
Examples of this {computational approach using EDML} can be found in \cite{BFKR06,BKR07,BKR08,BKLN13}. 
In addition to the symmetry group $\Gamma$, the software package
takes several parameters as input, which are {solely} determined by the critical spectrum of the linearized operator. In other words, the exact value of the bifurcation invariant associated to the zero solution of (\ref{eq:1}) depends only on the characteristic operator of (\ref{lin-eq}). 
In fact, all results that follow from the bifurcation invariant of \eqref{eq:1} remain valid for any $\Gamma$-symmetric system whose linearization has the form (\ref{lin-eq}), in particular for systems of form \eqref{puls-coup}--\eqref{eq-c}.                 
           
{Another advantage of using equivariant degree is that its basic degrees can be easily programmed and calculated in other computer programming languages such as GAP\footnote[2]{{\it GAP} (``Groups, Algorithms, Programming'')  is a non-commercial system for computational discrete algebra. It provides a programming language and large data libraries of algebraic objects. The system is distributed freely at {\tt http://www.gap-system.org}}, MATLAB, C++,
Java, and so on. An existing extension of the EDML package is the ``Dihedral Calculator'', which is programmed using the non-commercial language GAP. It is currently available for dihedral symmetry $D_n$ for $n\le 200$\footnote[7]{See {\it Dihedral Calculator} from MuchLearning {\tt http://dihedral.muchlearning.org}}. Other symmetry groups that are supported by EDML are the quaternion group $Q_8$, the alternating groups $A_4$, $A_5$, and the symmetric group $S_4$.}

{Our main results are Theorem \ref{thm:steady}, Corollary \ref{cor:steady}, Theorem \ref{thm:Hopf} and Proposition \ref{cor:Hopf}. The classification results for dihedral symmetry $D_{12}$ are summarized in Tables \ref{t:d12_sb}--\ref{t:d12_hb_3}. Theorem \ref{thm:steady} gives an existence result of steady-state bifurcations with their least symmetry. Corollary \ref{cor:steady} using the implicit function theorem sharpens this to exact symmetry. For Hopf bifurcations, existence result is stated in Theorem \ref{thm:Hopf} with the least symmetry. To obtain the precise symmetry as well as for submaximal isotropies, we introduce the concept of {\it secondary dominating orbit types} (cf. Definition \ref{def:dom_sec_dom}) to complement {dominating orbit types}. In Proposition \ref{cor:Hopf} bifurcating branches of maximal or submaximal isotropies are predicted with their precise symmetry.}

{The paper is organized as follows.} In Section \ref{sec:prelim}, we provide the basic { definition of equivariant degrees} and introduce the {necessary} notation and preliminary calculations {for $D_{12}$}. In Section \ref{sec:stability} we give a brief account of the stability analysis and the derivation of the basic bifurcation diagram for the linear system \eqref{lin-eq}. 
We take two quantities $\a:=\tau f'(0)$ and $\beta:=\tau\kappa\xi$ as bifurcation parameters, where $\xi$ is an eigenvalue of $C$. 
As we shall see, bifurcations, either of stationary or oscillatory nature, that destabilize the zero equilibrium, are related only to the extreme eigenvalues of the coupling matrix. {The main equivariant bifurcation results are}  given in Section~\ref{sec:symm_bif}. In Section \ref{sec:ring} we {apply} our results {to bidirectional rings of $12$ and obtain classification results listed in Tables \ref{t:d12_sb}--\ref{t:d12_hb_3}.}  
For rings of larger size, {we refer to the ``Dihedral Calculator'' mentioned earlier for calculations of degrees and} the method can be applied {systematically}.
In Section \ref{sec:simul} we connect the extreme eigenvalues to coupling strengths by enumerating all possible first-and second-nearest-neighbor coupling configurations of the 12-cell ring. We conclude by giving simulation examples in a concrete nonlinear system, namely the neural network model \eqref{neural}.

In closing this Introduction, we note that since the bifurcation invariant remains invariant against all (admissible,  equivariant) continuous deformations on the system, the classification results we obtain using the bifurcation invariant remain valid under modeling variations {within} the framework of symmetry. 
They may also be useful for systems encountered in real-world applications that are only ``approximately symmetric". For modeling issues on systems with imperfect symmetry, we refer to \cite{GS_symm}.

\section{Preliminaries} \label{sec:prelim}

\subsection{Groups and Group Representations}
Throughout we consider groups that are either finite or of form $\Gamma\times S^1$, where $\Gamma$ is a finite group and $S^1$ is the group of complex numbers of unit length.
 
Let $G$ be a group and $H$ be a closed subgroup of $G$, written as $H\subset G$. Let $N(H)=\{g\in G: gHg^{-1}=H\}$ be the {\it normalizer} of $H$ and $W(H)=N(H)/H$ the {\it Weyl group} of $H$.
The set of all closed subgroups of $G$ can be partially ordered by set inclusion. 
For subgroups $H,K\subset G$, we write  $H\le K$  if $H \subseteq K$; $H<K$ if $H\subsetneq K$. The symbol $(H)$ stands for the conjugacy class of the subgroup $H$ in $G$; that is $(H)=\{gHg^{-1}\,:\, g\in G\}$. The set of all conjugacy classes of closed subgroups of $G$ affords a partial order given by: $(H)\le (K)$ if $H\subseteq gKg^{-1}$ for some $g\in G$; similarly, $(H)<(K)$ if $H\subsetneq gKg^{-1}$ for some $g\in G$.

 \begin{ex}\rm\label{ex:d12_subgrps} (cf. \cite{AED}) Let $\Gamma=D_{12}$ be the dihedral group of order $24$, which is represented as the group of $12$ rotations: $1$, $\eta$, $\eta^2$, $\dots$, $\eta^{11}$ and $12$ reflections: $\varsigma$, $\varsigma\eta$, $\varsigma\eta^2$, $\dots$, $\varsigma\eta^{11}$ of the complex plane $\bc$, where $\eta$ stands for the complex multiplication by $e^{\frac{i\pi}{6}}$ and $\varsigma$ denotes the complex conjugation. There are two kinds of subgroups in $D_{12}$: cyclic and dihedral. The cyclic subgroups are $\bz_1,\bz_2,\bz_3,\bz_4,\bz_6,\bz_{12}$, where $\bz_k$ denotes the cyclic subgroup generated by $\eta^l$ with $l=\frac{12}{k}$. The dihedral subgroups are 
$$D_{k,j}=\{1,\eta^l,\eta^{2l},\dots, \eta^{(k-1)l}, \varsigma\eta^j,\varsigma\eta^{j+l},\varsigma\eta^{j+2l},\dots, \varsigma\eta^{j+(k-1)l}\},\q \text{for}\, 0\le j <l=\frac {12}{k},$$
where $k\in\{1,2,3,4,6,12\}$. If $l$ is odd, then all subgroups $D_{k,j}$ for $0\le j<l$ are conjugate to $D_{k,0}:=D_k$. If $l$ is even, then all subgroups $D_{k,j}$ with $j$ being even are conjugate to $D_{k,0}=D_k$; all subgroups $D_{k,j}$ with $j$ being odd are conjugate to $D_{k,1}:=\tilde D_k$. Thus, up to conjugacy relation, we have the dihedral subgroups: $D_1$, $\tilde D_1$, $D_2$, $\tilde D_2$, $D_3$, $\tilde D_3$, $D_4$, $D_6$ $\tilde D_6$, $D_{12}$.
 \END
 \end{ex}

A  {\it real} (resp. {\it complex}) {\it representation} of $G$ is a finite-dimensional real (resp. complex) vector space $X$ with a continuous map, or {\it action}, $\psi:G\times X\to X$ such that the map $\psi(g,\cdot):X\to X$ is linear, for every $g\in G$. Banach representations are similarly defined for Banach spaces with an action for which  $\psi(g,\cdot)$ is linear and bounded. We abbreviate $\psi(g,x)$ with $gx$. 

A subset $\Omega\subset X$ is called {\it invariant} if $gx\in \Omega$ whenever $x\in \Omega$ for all $g\in G$. An action on an invariant subset $\Omega\subset X$ is called {\it free} if the existence of an $x\in \Omega$ with $gx=x$ implies $g=e$ is the neutral element.  
A representation $X$ of $G$ is called {\it irreducible} if $\{0\}$ and $X$ are the only invariant subspaces in $X$.

 \begin{ex}\rm\label{ex:d12_rep} (cf.~\cite{AED}) The dihedral group  $D_{n}$, for $n\in \bn$ even,  has the following real irreducible representations:
\begin{itemize} 
\item[(i)] The trivial representation $\VV_0\simeq \br$, where every element acts as the identity map.
\item[(ii)] For $1\le i\le \frac n2-1$, there is the representation $\VV_i\simeq \br^2\simeq \bc$ given by the following actions:
\begin{align*}
\eta z&=\eta^i\cdot z,\q \varsigma z=\bar z,
\end{align*}
where $``\cdot"$ is the complex multiplication and $``\bar{\q}"$ is the complex conjugation.
\item[(iii)] The representation $\VV_{\frac n2}\simeq \br$ given by: $\eta x=x$ and $\varsigma x=-x$.
\item[(iv)] The representation $\VV_{\frac n2+1}\simeq \br$ given by: $\eta x=-x$ and $\varsigma x=x$.
\item[(v)] The representation $\VV_{\frac n2+2}\simeq \br$ given by: $\eta x=-x$ and $\varsigma x=-x$.
\end{itemize} 
It has the following complex irreducible representations:
\begin{itemize}
\item[(i)] The trivial representation $\UU_0\simeq \bc$, where every element acts as the identity map.
\item[(ii)] For $1\le j\le \frac n2-1$, there is the representation $\UU_j\simeq \bc\times \bc$ given by the following actions:
\begin{align*}
\eta (z_1,z_2)&=(\eta^j\cdot z_1, \eta^{-j}\cdot z_2),\q \varsigma (z_1,z_2)=(z_2,z_1),
\end{align*}
where $``\cdot"$ is the complex multiplication.
\item[(iii)] The representation $\UU_{\frac n2}\simeq \bc$ given by: $\eta z=z$ and $\varsigma z=-z$.
\item[(iv)] The representation $\UU_{\frac n2+1}\simeq \bc$ given by: $\eta z=-z$ and $\varsigma z=z$.
\item[(v)] The representation $\UU_{\frac n2+2}\simeq \bc$ given by: $\eta z=-z$ and $\varsigma z=-z$.
\end{itemize} 
For $n\in \bn$ odd, the dihedral group  $D_{n}$  has the above listed irreducible representations  (i)--(iii), where $n$ is replaced with ($n+1$).
 \END
 \end{ex} 
 
Let $x\in X$. By the {\it symmetry} of $x$, we mean the {\it isotropy} subgroup of $x$ given by $\iso(x):=\{g\in G\,:\, g x=x\}$ with respect to the group action on $X$. 
The set $\orb(x):=\{g x\,:\, g\in G\}$ is called the {\it orbit} of $x$ and the {\it symmetry} of the orbit is defined by the {\it orbit type} of $x$, which is the conjugacy class $(\iso(x))$ of $\iso(x)$. 
Note that $\iso(gx)=g\iso(x)g^{-1}$ for $g\in G$; 
thus, the symmetry of the orbit is independent of the choice of $x$ from the orbit.
 
Let $\Omega\subset X$ be a subset and $H\subset G$ be a closed subgroup. Define $\Omega_H=\{x\in X:\iso(x)=H\}$. It can be verified that the Weyl group $W(H)$ acts freely on $\Omega_H$. Denote the {\it $H$-fixed point subspace} in $\Omega$  by $\Omega^H=\{x\in X: gx=x,\,\forall\, g\in H\}$. {In Sections \ref{sec:stability}-\ref{sec:ring} we use $\fix(H)$ to denote the $H$-fixed point subspace, since  the space on which $G$ acts will be  clear from context}. Note that $\Omega_H\subset  \Omega^H$. Moreover, $\Omega^H$ is the  disjoint union of $\Omega_{\tilde H}$ for all $\tilde H\supseteq H$.

\begin{ex}\rm\label{ex:d12_v1_fixed} Let $\Gamma=D_{12}$ and $X=\VV_1$ be the real irreducible representation of $D_{12}$ given in Example \ref{ex:d12_rep}. Then, orbit types that occur in $X$ are $(D_{12})$, $(D_1)$, $(\tilde D_1)$ and $(\bz_1)$ (refer to Example \ref{ex:d12_subgrps} for notations), with the corresponding fixed point subspaces:
\begin{align*}
X^{D_{12}}=\{(0,0)\},\q X^{D_1}=\{(x,0):x\in \br\},\q X^{\tilde D_1}=\{r e^{-\frac{i\pi}{12}}:r\in \br\}, \q X^{\bz_1}=X.
\end{align*}
Note that $X^{D_1}$ is the disjoint union of subsets $X_{D_1}=\{(x,0):x\in \br,\, x\ne 0\}$ and $X_{D_{12}}=\{(0,0)\}$. On the subset $X_{D_1}$, the Weyl group $W(D_1)=D_2/D_1\simeq \bz_2$ acts freely by the reflection. On the subset $X_{D_12}$, the Weyl group $W(D_{12})=D_{12}/D_{12}\simeq \bz_1$ acts freely by the neutral element.
\END
\end{ex}

Finally, we remark that there is a natural way of ``converting'' a complex $\Gamma$-representation into a real $\Gamma\times S^1$-representation. Let $U$ be a complex $\Gamma$-representation. Define a $\Gamma\times S^1$-action on $U$ by
 \begin{equation}\label{eq:gamma_ext_s1}
(\gamma,z) u=z\cdot (\gamma u),\q \text{for $(\gamma, z)\in \Gamma\times S^1$, $u\in U$,}
\end{equation}
 where $\cdot$ stands for the complex multiplication. The obtained representation is denoted by $\bar U$ and called the {\it $\Gamma\times S^1$-representation induced} from $U$. Note that $\bar U$ is irreducible as a real $\Gamma\times S^1$-representation if $U$ is irreducible as a complex $\Gamma$-representation.

\subsection{Equivariant Maps and Equivariant Degree}

Let $X,Y$ be two Banach representations of $G$. A continuous map $f:X \to Y$ is called {\it equivariant} if $f(g_\circ x)=g_* f(x)$, for all $x\in X$ and $g\in G$, where $_\circ$ and $_*$ stand for the $G$-actions on $X$ and $Y$, respectively. 
 In equivariant nonlinear analysis, one is interested in finding zeros of an equivariant map $f$ in an invariant domain $\Omega\subset X$. 
Note that by equivariance, the set of all zeros of $f$ in $\Omega$ is composed of disjoint group orbits; thus one speaks of {\it zero orbits}, instead of zeros, of $f$. 

A map $f$ is called {\it admissible} on $\Omega$ if $f(x)\ne 0$  for all $x\in \partial\Omega$. A homotopy $h:[0,1]\times X \to Y$ is called {\it admissible} if $h(t,\cdot)$ is admissible for all $t\in [0,1]$. An equivariant degree, intuitively speaking, is an algebraic count of zero orbits of an admissible $f$ in $\Omega$ with respect to orbit types, which remains unchanged against all admissible (equivariant) homotopies from $f$.

In the next two subsections we review from \cite{AED} two types of equivariant degrees that will be used in Section \ref{sec:symm_bif} for bifurcation analysis. In both cases, the equivariant degree is first defined in finite-dimensional representations for continuous maps and then extended to infinite-dimensional Banach representations for compact vector fields.

\subsubsection{Equivariant Degree without Parameters}
 
Let $G=\Gamma$ be a finite group acting on a finite-dimensional $\Gamma$-representation $X$. Let $\Phi$ be the set of all orbit types that appear in $X$. That is, every element of $\Phi$ is a conjugacy class of a finite subgroup of $\Gamma$. Consider a continuous equivariant map $f:X\to X$ on an open bounded invariant domain $\Ome\subset X$ such that $f$ is admissible on $\Omega$. Define an {\it equivariant degree (without parameter)} of $f$ in $\Omega$ by a finite sum of integer-indexed orbit types:
\begin{equation}\label{eq:deg_0}
\gmdeg(f,\Omega)=\underset{(K)\in \Phi}{\sum} n_K\cdot (K),
\end{equation}
where $n_K\in \bz$ is an integer counting zero orbits of orbit type $(K)$. {One can also think of $\gmdeg$ as associating to every such pair $(f,\Omega)$  an integer sequence indexed by the set $\Phi$ of conjugacy classes. Depending on the value of $f$ on $\overline \Omega$ (with $\overline \Omega=\Omega\cup \partial \Omega$), the degree associates different integer values to different conjugacy classes. }  The precise definition of $n_K$ can be given by the following {\it recurrence formula}:
\begin{equation}\label{eq:rec_0}
n_K=\frac{ \deg (f|_{\Omega^K},\Omega^K)-\underset{(\tilde K)>(K)}{\sum}n_{\tilde K}\cdot |W(\tilde K)|\cdot n(K,\tilde K)}{|W(K)|}.
\end{equation}
We explain the notations used in (\ref{eq:rec_0}) and their geometric meaning. Recall that $\Omega^K$ denotes the fixed point subspace of $K$ in $\Omega$. By restricting $f$ on $\Omega^K$, one obtains an (admissible) map $f|_{\Omega^K}:\Omega^K\to \Omega^K$. Using the classical Brouwer degree ``$\deg$'', the integer ``$\deg (f|_{\Omega^K},\Omega^K)$'' counts  the zeros of  $f$ in $\Omega^K$. Since not every element in $\Omega^K$ has the precise isotropy $K$, one needs to subtract those zeros of larger isotropies. This is done by subtracting the summands in (\ref{eq:rec_0}). Within each summand, $n_{\tilde K}$ is the integer counting zero orbits of orbit type $(\tilde K)$. Since the Weyl group $W(\tilde K)$ acts freely on $\Omega_{\tilde K}$, the integer  $n_{\tilde K}\cdot |W(\tilde K)|$ then counts the zeros of isotropy $\tilde K$. The number $n(K,\tilde K)$ is defined as the number of distinct conjugate copies of $\tilde K$ that contain $K$, formally by
\begin{equation}\label{eq:nLH}
n(K,\tilde K)=\big|\frac{\{g\in \Gamma: K\subset g\tilde K g^{-1}\}}{N(\tilde K)}\big|.
\end{equation}
Thus, the number  $n_{\tilde K}\cdot |W(\tilde K)|\cdot n(K,\tilde K)$ counts the zeros of isotropy $K'$ for all $K'$  with  $(K')=(\tilde K)$. It follows that the expression of the numerator in (\ref{eq:rec_0}) gives the count of zeros of $f$ having precise isotropy $K$. Again, since $W(K)$ acts freely on $\Omega_K$, we have then the total expression on the right hand side of (\ref{eq:rec_0}) giving the count of zero orbits of $f$ having orbit type $(K)$.

\begin{ex}\rm Let $\Gamma=D_{12}$ and $X=\VV_1$ be the real irreducible representation of $D_{12}$ given in Example \ref{ex:d12_rep}. Consider the antipodal map $f=-\id:X\to X$ on the unit disc $B\subset X$, which is $D_{12}$-equivariant and $B$-admissible. 
As mentioned in Example \ref{ex:d12_v1_fixed}, orbit types that occur in $\VV_1$ are $(D_{12})$, $(D_1)$, $(\tilde D_1)$, and $(\bz_1)$. Thus,
\[\gmdeg(-\id, B)=n_{D_{12}}\cdot (D_{12})+n_{D_{1}}\cdot (D_{1})+n_{\tilde D_{1}}\cdot (\tilde D_{1})+n_{\bz_1}\cdot (\bz_1).\]
We compute $n_{D_1}$ using (\ref{eq:rec_0}). To do so, we first need to compute $n_{D_{12}}$:
\[n_{D_{12}}=\frac{\deg(-\id,B^{D_{12}})}{|W(D_{12})|}=\frac{1}{1}=1,\]
where we used the fact $B^{D_{12}}=X^{D_{12}}\cap B=\{(0,0)\}$, $W(D_{12})=\bz_1$ from Example \ref{ex:d12_v1_fixed} and $\deg(-\id, \br^m)=(-1)^m$ for $m\in\{0\}\cup \bn$. Thus,  we have 
\[n_{D_1}=\frac{\deg(-\id, B^{D_1})-1\cdot 1\cdot 1}{|W(D_1)|}=\frac{-1-1}{2}=-1,\]
where we used the fact $n(D_1,D_{12})=\big|\frac{D_{12}}{D_{12}}\big|=1$ and $W(D_1)= \bz_2$. Following (\ref{eq:rec_0}) further, one shows that
\[\gmdeg(-\id, B)=(D_{12})- (D_{1})-(\tilde D_{1})+(\bz_1).\]
\END
\end{ex}

The definition of equivariant degree can be extended, in a standard way, to infinite-dimensional Banach representations for {\it compact equivariant fields}, namely, equivariant maps of the form $f=\id-F:D\subset X\to X$ that are admissible on a bounded domain $D$ such that $\overline {F(D)}$ is compact. It was shown in \cite{BKR06} that the equivariant degree defined by (\ref{eq:deg_0})--(\ref{eq:rec_0}), as well as its infinite-dimensional extension, satisfies the usual properties of degree theory, such as the {\it existence property} which states that
\[n_K\ne 0\,\,\text{in\,\, (\ref{eq:deg_0})}\q\Rightarrow \q f^{-1}(0)\cap \Omega^K\ne\emptyset,\]
which can be useful for predicting zero orbits of orbit type at least $(K)$.

\subsubsection{Equivariant Degree with One Parameter}

Let $G=\Gamma\times S^1$ be the product of a finite group $\Gamma$ and the circle group $S^1$. There are two types of closed subgroups in $G$: those subgroups that are of the form $K\times S^1$ for some subgroup $K\subset\Gamma$, or 
the {\it twisted subgroups} of $G$, defined as follows.
 
\begin{Def}\label{def:twisted_sub}\rm A subgroup $H\subset \Gamma\times S^1$ is called a {\it twisted $l$-folded subgroup}, if there exists a subgroup $K\subset \Gamma$, an integer $l\ge 0$ and a group homomorphism $\phi: K\to S^1$ such that
\begin{equation*}
H=K^{\phi,l}:=\{(\gamma,z)\,:\, \phi(\gamma)=z^l\}.
\end{equation*}
For $l=1$, $H$ is called a {\it twisted subgroup} for simplicity.   
Conjugacy classes of twisted subgroups are called {\it twisted orbit types}.
\END
\end{Def}   

\begin{rmk}\rm\label{rmk:Weyl}
{Note that the  subgroups of the form $K\times S^1$ (for some $K\subset\Gamma$) and the twisted ($l$-folded) subgroups can also be distinguished using the dimension of their Weyl groups. While the former have $0$-dimensional Weyl groups, the latter have $1$-dimensional Weyl groups in $\Gamma\times S^1$. Thus, the Weyl group of a twisted  ($l$-folded) subgroup is homeomorphic to a number of finitely many disjoint circles.
} 
\end{rmk}

 \begin{ex}\rm\label{ex:d12s1_subgrps} Let $G=D_{12}\times S^1$ be the product group of the dihedral $D_{12}$ and the unit circle $S^1\subset \bc$. We describe its twisted subgroups $H=K^{\phi}$. Clearly, all subgroups of $D_{12}$ are twisted subgroups with $\phi\equiv 1\in S^1$. 
 Besides that, there are twisted subgroups that are not contained in $D_{12}$. These can be classified into two categories: those for which $K=\bz_k$ and those for which $K=D_{k,j}$ (refer to Example \ref{ex:d12_subgrps} for notation).
 
Let $K=\bz_k$ for some $k\in\{1,2,3,4,6,12\}$ and $\phi:K\to S^1$ be given by $\phi(\eta^l)=\eta^{jl}$ for some $j$ with $1\le j<k$. Then,
\[K^\phi=\{(1,1),(\eta^l,\eta^{jl}),(\eta^{2l},\eta^{2jl}),\dots, (\eta^{(k-1)l},\eta^{j(k-1)l})\}:=\bz_k^{t_j},\q\text{for $1\le j<k$.}\]
Among these subgroups, $\bz_k^{t_j}$ and $\bz_k^{t_{k-j}}$ are conjugate to each other, for  $1\le j<k$. Thus, for $k$ even, up to conjugacy relation, we have the twisted subgroups $\bz_k^{t_1}$, $\bz_k^{t_2},\dots,$  $\bz_k^{t_{\frac k2}}:=\bz_k^d$, and for $k$ odd, $\bz_k^{t_1}$, $\bz_k^{t_2},\dots,$  $\bz_k^{t_{\frac {k-1}{2}}}$. That is, we have $\bz_2^d$, $\bz_3^{t_1}$, $\bz_4^{t_1}$, $\bz_4^d$, $\bz_6^{t_1}$, $\bz_6^{t_2}$, $\bz_6^{d}$, $\bz_{12}^{t_1}$, $\bz_{12}^{t_2}$, $\bz_{12}^{t_3}$, $\bz_{12}^{t_4}$, $\bz_{12}^{t_5}$, $\bz_{12}^{d}$. 

Let $K=D_{k,j}$ for some $k\in\{1,2,3,4,6,12\}$ and $0\le j <l=\frac{12}{k}$. Up to conjugacy, it is sufficient to consider $K=D_k$ in case $l$ is odd, and $K=D_k$, $K=\tilde D_k$ in case $l$ is even (cf. Example \ref{ex:d12_subgrps}). Let  $\phi: K\to S^1$ be the group homomorphism such that $\ker \phi=\bz_k$. Then,
\[D_k^\phi=\{(1,1),(\eta^l,1),\dots, (\eta^{(k-1)l},1),(\varsigma, -1),(\varsigma\eta^l, -1),\dots, (\varsigma\eta^{(k-1)l},-1)\}:=D_k^z,\]
and 
{ \small \[\tilde D_k^\phi=\{(1,1),(\eta^l,1),\dots, (\eta^{(k-1)l},1),(\varsigma\eta , -1),(\varsigma\eta^{1+l} , -1),\dots, (\varsigma\eta^{1+(k-1)l},-1)\}:=\tilde D_k^z,\q\text{if $l$ is even.}\]}
Thus, we have $D_1^z$, $\tilde D_1^z$, $D_2^z$, $\tilde D_2^z$, $D_3^z$, $\tilde D_3^z$, $D_4^z$, $D_6^z$, $\tilde D_6^z$, $D_{12}^z$.

In the case $k$ is even, there is a group homomorphism $\phi: K\to S^1$ for which $\ker \phi=D_{\frac k2}$. Then,
\[D_k^\phi=\{(1,1),(\eta^l,-1),(\eta^{2l},1),\dots, (\eta^{(k-1)l},-1),(\varsigma, 1),(\varsigma\eta^l, -1),\dots, (\varsigma\eta^{(k-1)l},-1)\}:=D_k^d,\]
and 
{\small\[\tilde D_k^\phi=\{(1,1),(\eta^l,-1),(\eta^{2l},1), \dots, (\eta^{(k-1)l},-1),(\varsigma\eta , 1),(\varsigma\eta^{1+l}, -1),\dots, (\varsigma\eta^{1+(k-1)l},-1)\}:=\tilde D_k^d,\q\text{if $l$ is even.}\]}
Also, there is a group homomorphism $\phi: K\to S^1$ for which $\ker \phi=\tilde D_{\frac k2}$. Then,
\[D_k^\phi=\{(1,1),(\eta^l,-1),(\eta^{2l},1),\dots, (\eta^{(k-1)l},-1),(\varsigma, -1),(\varsigma\eta^l, 1),\dots, (\varsigma\eta^{(k-1)l},1)\}:=D_k^{\hat d},\]
and 
{\small \[\tilde D_k^\phi=\{(1,1),(\eta^l,-1),(\eta^{2l},1), \dots, (\eta^{(k-1)l},-1),(\varsigma\eta, -1),(\varsigma\eta^{1+l}, 1),\dots, (\varsigma\eta^{1+(k-1)l},1)\}:=\tilde D_k^{\hat d},\q\text{if $l$ is even.}\]}

One shows that for $l$ even, $D_k^d$ and $D_k^{\hat d}$ are conjugate, and $\tilde D_k^{d}$ and $\tilde D_k^{\hat d}$ are conjugate. Thus, in the case $k$ is even, up to conjugacy relation,  we have the twisted subgroups $D_k^d$ and $D_k^{\hat d}$ if $l$ is odd, and $D_k^d$ and  $\tilde D_k^d$ if $l$ is even. That is, for $D_{12}$, we have $D_1^d$, $\tilde D_1^d$, $D_2^d$, $\tilde D_2^d$, $D_3^d$, $\tilde D_3^d$, $D_4^d$, $ D_4^{\hat d}$, $D_6^d$, $\tilde D_6^d$, $D_{12}^d$, $D_{12}^{\hat d}$. 
 \END
 \end{ex}

Let $X$  be a finite-dimensional representation of $G$ and $\br$ be the one-dimensional  {\it parameter space} on which $G$ acts trivially.
Let $\Phi_1$ be the set of all twisted orbit types that appear in $\br\times X$. Consider a continuous equivariant map $f:\br\times X\to X$ on an open bounded invariant domain $\Ome\subset \br \times X$ such that $f$ is admissible on $\Omega$. Define an {\it equivariant degree (with one parameter)} of $f$ in $\Omega$ by a finite sum of integer-indexed twisted orbit types:
\begin{equation}\label{eq:deg_1}
\g1deg(f,\Omega)=\underset{(H)\in \Phi_1}{\sum} n_H\cdot (H),
\end{equation}
where $n_H\in \bz$ is an integer counting zero orbits of the twisted orbit type $(H)$. {More precisely, $n_H$ can be computed by the following recurrence formula
\[n_H=\frac{\underset{k}{\sum}\deg_k(f|_{\Omega^H},\Omega^H)-\underset{(\tilde H)>(H)}{\sum} n_{\tilde H}\cdot n(H,\tilde H)\cdot|W(\tilde H)/S^1|}{|W(H)/S^1|}\]
in the same spirit of (\ref{eq:rec_0}). We explain the notation in detail. Again,  $n_H$ is supposed to count the zero orbits of orbit type $(H)$ in $\Omega$, or equivalently, the zero orbits of isotropy $H$ in $\Omega^H$. Restricting the map $f$ on $\Omega^H$, we  consider $f|_{\Omega^H}:\Omega^H\to \Omega^H$. Since $W(H)$ is homeomorphic to $|W(H)/S^1|$ copies  of finitely many  disjoint circles (cf. Remark \ref{rmk:Weyl}) and $|W(H)|$ acts freely on $\Omega_H$, the number $n_H  |W(H)/S^1|$ counts the copies of circles in the zeros having isotropy $H$ in $\Omega^H$. 
Using the classical $S^1$-degree (e.g.~see \cite{IV-B}), the integer $\deg_k(f|_{\Omega^H},\Omega^H)$ counts the number of circles in the zeros of $f$ having isotropy $\bz_k$ in $\Omega^H$. The first sum  then counts the total number of circles in the zeros of $f$ in $\Omega^H$. The summand in the second sum counts the copies of circles in the zeros of $f$ which have isotopy $\tilde H$, where the number $n(H,\tilde H)$ is given by (\ref{eq:nLH}). 
Thus, 
we obtain from the numerator the number of circles in the zeros of $f$ in $\Omega^H$ with precise isotropy $H$. Divided by the number $|W(H)/S^1|$ of copies in $W(H)$, this gives the count of zero orbits with precise isotropy $H$.   
} 
\vs
This degree can be extended to infinite-dimensional Banach representations for compact equivariant fields. The resulting degree satisfies all classical properties of an equivariant degree theory, among which the {\it existence property} plays an important role for our purpose:
\[n_H\ne 0\,\,\text{in\,\, (\ref{eq:deg_1})}\q\Rightarrow \q f^{-1}(0)\cap \Omega^H\ne\emptyset.\]

\section{Stability Analysis and the Bifurcation Diagram}
\label{sec:stability}
We now consider the coupled system \eqref{eq:1} and the corresponding linear variational equation \eqref{lin-eq} about the zero solution.
For $\tau>0$, it is convenient to rescale the time $t \to t/\tau$ so that the linearized equation takes the form  
\begin{equation}
	\dot{y}(t) = \tau f'(0)y(t) + \tau\kappa C y(t-1).	
\label{eq:2_lin} 
\end{equation}
The characteristic operator $\Delta(\lam):\bc^n\to \bc^n$ for (\ref{eq:2_lin}) is
\begin{equation} \label{eq:charop}
\Delta(\lam)=(\lam -\tau f'(0))I_n -\tau\kappa  e^{-\lam}C,
\end{equation}
and the corresponding characteristic equation is
\begin{equation}\label{eq:char}
\det\Delta(\lam)= \prod_{\xi\in \sig(C)} \big (\lam -\tau f'(0)-\tau\kappa  e^{-\lam}\xi\big)=0,
\end{equation}
where $\sig(C)$ denotes the spectrum of $C$.

%
%

Since $C$ is assumed to be a symmetric matrix, we have  $\sig(C)\subset \br$. 
In this case, each of the factors on the right side of (\ref{eq:char}) can be analyzed using well-known methods for scalar delay equations with real coefficients \cite{Hayes50,Diekmann95}. 
Thus, let $\xi \in \sigma(C) \subset \mathbb{R}$ and consider the corresponding factor in \eqref{eq:char}. 
If $\lam=u+iv$ is a characteristic root, then separating real and imaginary parts leads to
\begin{equation}\label{eq:uv_a_b}
\begin{cases}
u-\a -\beta e^{-u}\cos v=0\\
v+\beta e^{-u} \sin v=0,
\end{cases}
\end{equation}
where $\a=\tau f'(0)$ and $\beta=\tau\kappa \xi$.
For purely imaginary roots, we have $u=0$, giving
\begin{equation}\label{eq:v_a_b}
\begin{cases}
-\a -\beta \cos v=0\\
v+\beta \sin v=0.
\end{cases}
\end{equation}
For $v=0$ the solution is the line L1 defined by $\beta=-\alpha$, which corresponds to parameter values for which $\lambda=0$ is a characteristic root.
Over the intervals $v \in (k\pi,(k+1)\pi)$, $k\in \mathbb{Z}$, the solution can be expressed in the parametric form $(\alpha(v),\beta(v))=(v/ \tan(v), -v/ \sin(v))$, which gives parametric curves for which there exists a pair of purely imaginary characteristic roots of the form $\lambda=\pm i v$. 
These bifurcation curves are depicted in Figure~\ref{fig:bif}. 
Knowing that the zero solution is stable for $\beta=0$ and $\alpha<0$, 
and because characteristic roots can cross the imaginary axis only for parameter values belonging to the bifurcation curves,
one can then move vertically in the parameter plane, increasing the number of roots with positive real parts appropriately each time a bifurcation curve is crossed. 
Implicit differentiation on bifurcation curves shows that the characteristic roots on the imaginary axis move to the right as $|\beta|$ increases, yielding the picture shown in Figure~\ref{fig:bif}.

\begin{figure}[tb]
\begin{center}
\includegraphics[scale=0.8]{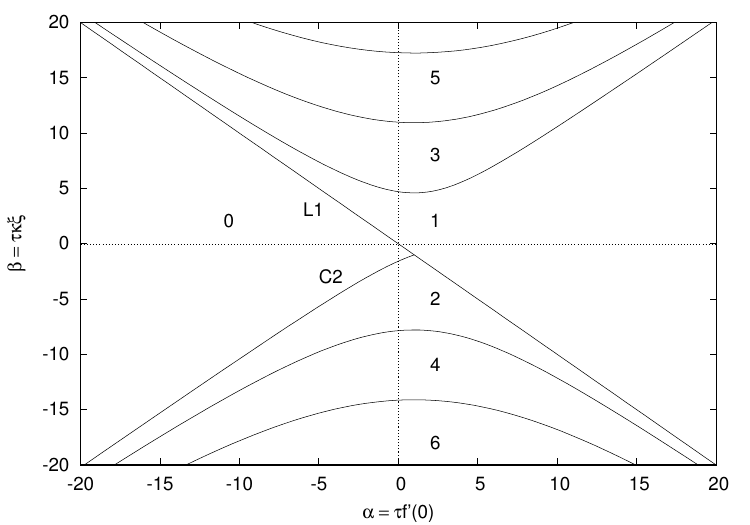} 
\caption{Bifurcation diagram of the characteristic equation {(\ref{eq:char})}. The curves indicate the parameter values for which the characteristic equation has a root on the imaginary axis. The curves separate the $\alpha$--$\beta$ parameter plane into regions in which the number of characteristic roots with positive real part is a constant, the value of which is indicated in the figure. Hence ``0" indicates the region where the origin is stable, which is bounded from above by the straight line L1 and from below by the curve C2.}
\label{fig:bif}
\end{center} 
\end{figure}

The region of stability is indicated in Figure~\ref{fig:bif} by the label ``0". It is bounded from above by the straight line L1 and from below by the curve C2. The latter is given by the parametric branch $(\alpha,\beta)=(v/ \tan(v), -v/ \sin(v))$, $v\in(0,\pi)$, and meets the line L1 at the point $(1,-1)$.
This is  for one particular spatial mode corresponding to the eigenvalue $\xi$.
One can then repeat the same argument for all eigenmodes $\xi\in\sigma(C)$.
If a parameter pair $(\a,\beta)$ is varied to leave the stable region by crossing the line L1, a bifurcation of steady states may occur, 
whereas crossing the curve C2 may lead to a bifurcation of periodic solutions. 
The codimension of these bifurcations is related to the multiplicity of the eigenvalue $\xi$ given by the critical value of $\beta=\tau\kappa \xi$.

\subsection{Effects of Network Structure} 

Suppose we start with stable systems $(f'(0)<0)$ without coupling, so we are initially on the negative $\alpha$-axis. 
As the coupling strength $\kappa$ is increased from zero, stability may be lost via a stationary or an oscillatory bifurcation through the first eigenmode $\xi$ to hit L1 or C2. The important observation is that this first bifurcation depends only on the extremal eigenvalues $\xi$ of the coupling matrix $C$.
Hence, the number of relevant parameters is greatly reduced and one needs to check only the two extremal eigenvalues of the coupling matrix regardless of the network size. In this way it is possible to classify networks by defining equivalence classes according to the extreme eigenvalues: networks having the same smallest and largest eigenvalues will have identical stability properties with regard to the class.

In special cases it is possible to give more precise statements. For diffusively coupled systems such as (\ref{eq-b}) or (\ref{eq-c}), the coupling matrix $C$ equals the negative of the Laplacian matrix. Therefore, in case the connection weights $a_{ij}$ are nonnegative, all the eigenvalues of $C$ are non-positive, the largest one always being zero. 
In fact, for connected networks, the eigenvalues are strictly negative, except for a single zero eigenvalue (see e.g.~\cite{Godsil-Royle}). 
In this case, it is the smallest eigenvalue of $C$ (i.e., the largest Laplacian eigenvalue) that determines the first bifurcation. As far as the network structure is concerned, this is the only relevant quantity. 

For systems of the form \eqref{neural} or \eqref{puls-coup}, $C$ is given by the adjacency matrix $A$, which can have both negative and positive eigenvalues even when all $a_{ij}$ have the same sign. Thus both $\xi_{\min}$ and $\xi_{\max}$ should be considered for the first bifurcation. For sufficiently small $\tau$, the bifurcation occurs in the vicinity of the origin of the $\alpha$-$\beta$ parameter plane of Figure \ref{fig:bif}. Since the line L1 intersects the origin where the curve C2 has a gap, the likely bifurcation is a stationary one and the  eigenvalue responsible for the bifurcation is the largest positive eigenvalue of $A$. This agrees with the observation of Section \ref{sec:delayeffect} below that oscillatory bifurcations arise from sufficiently large delays, for the class of scalar systems studied in this paper.

\subsection{Effects of Delay}
\label{sec:delayeffect} 

In the absence of delays ($\tau=0$), the characteristic equation for \eqref{lin-eq} is 
\begin{equation}
	\prod_{\xi\in \sig(C)} (\lambda - f'(0)-\kappa \xi)=0.
\end{equation}
from which the characteristic roots can be directly read off as $\lambda = f'(0)+\kappa \xi$, $\xi \in \sigma(C)$. The roots are real for real network eigenvalues $\xi$; hence the only critical root is $\lambda=0$, which occurs when $f'(0)=-\kappa \xi$. The corresponding critical curve is a straight line on the parameter plane of  $f'(0)$ versus $\kappa\xi$, which is identical with the line L1 of Figure~\ref{fig:bif}. Thus, one has stability below this line and one real positive characteristic root above, for a given spatial mode corresponding to $\xi$. In particular, Hopf bifurcations do not occur. 

To see the effects of the delay, we fix the other quantities 
$\kappa,\xi$ and $f'(0)$ and notice that the values of $\alpha,\beta$ in Figure~\ref{fig:bif} then change only along the ray emanating from the origin with slope $m=\kappa\xi/f'(0)$.
The delay $\tau$ parametrizes the distance of points along the ray to the origin. Hence, to use the delay as a bifurcation parameter, one goes along the ray starting from the origin and obtain bifurcations as the curves given in Figure~\ref{fig:bif} are intersected. Such rays only intersect with L1 at the origin or else completely coincide with L1; moreover, they intersect the other curves if the slope $m$ is sufficiently large. Since the latter curves correspond to pairs of purely imaginary characteristic roots, large values of the delay cause oscillatory bifurcations. 

To summarize, stationary bifurcations given by L1 of Figure~\ref{fig:bif} are independent of the delay, whereas the remaining set of curves correspond to oscillatory bifurcation resulting from the delay.
In the following sections we will consider both stationary and oscillatory bifurcations in our symmetry analysis. Stationary bifurcations will be relevant for both delayed and undelayed systems, whereas oscillatory bifurcations will be a feature of delayed systems only, 
in the context of scalar systems that we consider.



\section{Symmetries and Equivariant Bifurcations}
\label{sec:symm_bif}

By a {\it symmetry} of a dynamical system, we mean a group of elements acting on the phase space that keep the system invariant. {More precisely, given a system of form 
\begin{equation}\label{eq:dx=Fx}
\frac{dx}{dt}=F(x)
\end{equation}
with $x\in \br^n$ and $F:\br^n\to \br^n$, and an action of a group $\Gamma$ on the phase space $\br^n$, an element $\gamma\in \Gamma$ is called a {\it symmetry} of (\ref{eq:dx=Fx}) if (\ref{eq:dx=Fx}) remains unchanged after applying the action of $\gamma$ on both sides. Since a group action is linear, it commutes with the linear operator $\frac{d}{dt}$; thus $\gamma$ is a symmetry of (\ref{eq:dx=Fx}) if and only if $\gamma F(x)=F(\gamma x)$ for all $x\in \br^n$.}
 
Let $S_n$ be the group of all permutations of $n$ symbols. 
For $\si\in S_n$, consider its natural action on $\br^n$ by $(x_1,\dots,x_n)\mapsto (x_{\si(1)},\dots, x_{\si(n)})$. Consider a subgroup  $\Gamma\subset S_n$.

\begin{lem}\label{lem:c_Gamma}\rm Let $\kappa\ne 0$. Then $\Gamma$ is a symmetry of systems of form  (\ref{eq:1})  if and only if 
\begin{equation}\label{eq:sym}
g_{\si(i)}(x_1,x_2,\dots, x_n)=g_i(x_{\si(1)},x_{\si(2)},\dots, x_{\si(n)}),
\end{equation}
for all   $\si\in\Gamma$ and $(x_1,x_2,\dots, x_n)\in \br^n$.
\end{lem}

\begin{proof} 
Let $\si\in \Gamma$ and apply its action on (\ref{eq:1}). We obtain
\begin{equation}\label{eq:1_symm}
\dot x_{\si(i)}(t)=f(x_{\si(i)}(t))+\kappa g_i(x_{\si(1)}(t-\tau),x_{\si(2)}(t-\tau),\dots, x_{\si(n)}(t-\tau)).
\end{equation}
Comparing with (\ref{eq:1}), we see that (\ref{eq:1_symm}) is the same system as  (\ref{eq:1}) if and only if 
\[\kappa g_i(x_{\si(1)}(t-\tau),x_{\si(2)}(t-\tau),\dots, x_{\si(n)}(t-\tau))=
\kappa g_{\si(i)}(x_{1}(t-\tau),x_{2}(t-\tau),\dots, x_{n}(t-\tau)).\]
This leads to (\ref{eq:sym}), since $\kappa\ne 0$.
\end{proof}

\begin{rmk} \label{rmk:symm}\rm Note that a necessary condition for (\ref{eq:sym}) to hold is that the coupling matrix $C$ in the linearization (\ref{lin-eq}) satisfies
\begin{equation}\label{eq:cij}
c_{ij}=c_{\si(i)\si(j)},\q\forall\, \si\in\Gamma.
\end{equation}
For the specific systems \eqref{neural}--(\ref{eq-c})
it can be checked that \eqref{eq:cij} is also a sufficient condition, since (\ref{eq:sym}) reduces to $a_{ij}
=a_{\si(i)\si(j)}$  $\forall \si\in \Gamma$.
\END
\end{rmk}  


       
   

In what follows, we will study the bifurcations that destabilize the zero solution under a group of symmetries $\Gamma\subset S_n$ of the system (\ref{eq:1}) using the equivariant degree. Exact values of associated bifurcation invariants are calculated using the EDML (Equivariant Degree Maple Library) Package, by calling  
\begin{equation}\label{eq:showdeg}
\text{\tt showdegree}[\Gamma](n_0,n_1, \dots,n_r,m_0,m_1,\dots,m_s),\q \text{for $n_i,m_j\in \bz$,}
\end{equation} 
where the $n_i$ and $m_j$ are integers to be determined by the critical spectrum of the linearized system at the equilibrium.
The integers $r$ and $s$  in (\ref{eq:showdeg}) are predetermined by $\Gamma$ and are equal to the number of all distinct (nontrivial) irreducible representations of $\Gamma$ over reals and over complex numbers, respectively. In what follows, 
we use $\VV_0,\VV_1,\dots, \VV_r$ for the distinct real irreducible representations and $\UU_0,\UU_1,\dots, \UU_s$ for the complex ones, where $\VV_0$ and $\UU_0$ are reserved for the trivial representations.

\subsection{Steady-State Bifurcations}

In reference to Figure~\ref{fig:bif}, suppose that the parameters $(\a,\beta)$ are varied to leave the shaded stability region by crossing $\LL$ at some point $(\a_o,\beta_o)$. Then, 
\begin{equation}\label{eq:xio} 
\a_o=-\beta_o=\tau\kappa \xi_o,
\end{equation} 
for an eigenvalue $\xi_o\in\sig(C)$.  For $\tau,\kappa >0$, $\xi_o$ is the maximal eigenvalue of $C$. Let $E(\xi_o)$ be the generalized eigenspace of $\xi_o$.
Given the $\Gamma$-action on $\br^n$, we decompose $\br^n$ into pieces of $\VV_i$'s: 
\[\br^n=V_0\times V_1 \times\cdots\times V_r,\]
where every $V_i$
\begin{equation}\label{eq:Vi_ni} 
V_i=\underbrace{\VV_i\times\cdots\times \VV_i}_{n_i\,\text{times}}
\end{equation}
is a product of $n_i$ copies of $\VV_i$ for some integer $n_i\in \bn\cup\{0\}$. Also, since $E(\xi_o)$ is a $\Gamma$-invariant subspace of $\br^n$, we can decompose $E(\xi_o)$ as:
\[E(\xi_o)=E_0\times E_1 \times\cdots\times E_r,\]
where every $E_i$ is given by
\begin{equation}\label{eq:Ei_ei} 
E_i=\underbrace{\VV_i\times\cdots\times \VV_i}_{e_i\,\text{times}}
\end{equation}
i.e., as a product of $e_i$ copies of $\VV_i$ for some integer $e_i\in \bn\cup\{0\}$. Using (\ref{eq:Vi_ni})--(\ref{eq:Ei_ei}), define
\begin{equation}\label{eq:ui}
u_i:=n_i-e_i,
\end{equation}
for $i=0,1,\dots,r$. Then, the bifurcation invariant around $(\a_o,\beta_o)$ is given by 
\begin{align}\label{eq:ome0}
\ome_0&:=\text{\tt showdegree}[\Gamma](n_0, \dots,n_r,1,0,\dots,0)-\text{\tt showdegree}[\Gamma](u_0, \dots,u_r,1,0,\dots,0).
\end{align} 
Running the EDML package we obtain the value of $\ome_0$, which is of form
\begin{equation*}
c_1(K_1)+c_2(K_2)+\cdots +c_p(K_p), 
\end{equation*}
for integers $c_i\in \bz$ and conjugacy classes $(K_i)$ of subgroups $K_i$ in $\Gamma$. 

\begin{thm}\label{thm:steady}  Let $(\a_o,\beta_o)$ be such that $\a_o=-\beta_o$ and $\xi_o\in \sig(C)$ be  given by (\ref{eq:xio}). If $\ome_0$ given by (\ref{eq:ome0}) is of form 
\[\ome_0=c_1(K_1)+c_2(K_2)+\cdots +c_p(K_p),
\]
for some $c_i\ne 0$, then there exists a bifurcating branch of steady states of symmetry at least $(K_i)$.
\end{thm} 
\begin{proof} {By the existence property of equivariant degree, it is sufficient to prove that the value of the bifurcation invariant for the steady state bifurcations around $(\a_o,\beta_o)$ is indeed given by formula (\ref{eq:ome0}), since if this is the case, then any non-zero coefficient in the value indicates a bifurcating branch with the corresponding symmetry.} The formula (\ref{eq:ome0}) {follows from Theorem 8.5.2 in} \cite{KW}, {but we give} an alternative and more straightforward  proof in \ref{sec:proof} {for completeness}.  
  
\end{proof}

\begin{cor} \label{cor:steady}  Assume the hypotheses of Theorem \ref{thm:steady}, and suppose furthermore that the subgroup $K_i$ satisfies
\begin{equation}\label{eq:cor_cd}
\xi_o\not\in \sig(C|_{\fix(H)}),\q \forall\, H\supsetneq K_i.
\end{equation}
Then there exists a bifurcating branch of steady states of symmetry precisely $(K_i)$.
\end{cor}

\begin{proof} By Theorem \ref{thm:steady}, there exists a bifurcating branch of steady states of symmetry {\it at least} $(K_i)$. Let $(H)$ be the symmetry of this branch of solutions. Then, $(H)\ge (K_i)$. Up to the group conjugacy, we have $H \supseteq K_i$. 
We need to show $H=K_i$. Assume to the contrary that $H\supsetneq K_i$. 
Then by (\ref{eq:cor_cd}) we have that, when restricted to $\fix(H)$, the characteristic operator $\Delta(0)|_{\fix(H)}: \fix(H)\to \fix(H)$ is an isomorphism for $(\a, \beta)$ in a neighborhood of $(\a_o,\beta_o)$. By the implicit function theorem, there can be no additional solution in neighborhood of the trivial solution $x=0\in \fix(H)$, which is a contradiction.

\end{proof}

\subsection{Hopf Bifurcations} 
Assume that $(\a, \beta)$ leaves the shaded stability region of Figure \ref{fig:bif} by crossing C2 at some point $(\a_o,\beta_o)$. Since $\CC$ bounds the region from below and $\tau, \kappa >0$, the first  parameter pair that crosses $\CC$ must be related to the minimal eigenvalue $\xi_{\min}$ of $C$.
Let $\xi_o\in \sig(C)$ be the corresponding eigenvalue, i.e. 
\begin{equation}\label{eq:hopf_xi}
\beta_o=\tau \kappa \xi_o.
\end{equation} 
That is, $\xi_o=\xi_{\min}$ becomes critical. Consider the complexification $\bc^n=\bc\otimes \br^n$ of the phase space $\br^n$ and extend the $\Gamma$-action on $\bc^n$ by defining 
\begin{equation}\label{eq:gamma_c_n_action}
\gamma (z\otimes x)=z\otimes (\gamma x),\q\text{ for}\q \gamma\in \Gamma,\, x\in \br^n.
\end{equation}
The (generalized) eigenspace $E(\xi_o)$ remains $\Gamma$-invariant as a complex subspace of $\bc^n$. Thus, we decompose $E(\xi_o)$ into irreducible representations $\UU_0,\UU_1,\dots,\UU_s$ as:
\[E(\xi_o)=F_0\times F_1 \times\cdots\times F_s,\]
where every $F_j$ is given by
\begin{equation}\label{eq:Fj_mj}
F_j=\underbrace{\UU_j\times\cdots\times \UU_j}_{m_j\,\text{times}}
\end{equation}
that is, a product of $m_j$ copies of $\UU_j$ for some integer $m_j\in \bn\cup\{0\}$. Then, the bifurcation invariant around $(\a_o,\beta_o)$ for Hopf bifurcation is given by 
\begin{align}\label{eq:ome1}
\ome_1&:=\text{\tt showdegree}[\Gamma](0,0, \dots,0,-m_0,-m_1,\dots,-m_s).
\end{align}
Running the EDML package, we obtain the value of $\ome_1$ being of form
\[c_1(H_1)+c_2(H_2)+\cdots +c_q(H_q), 
\]
for integer coefficients $c_j\in \bz$ and conjugacy classes $(H_j)$ of subgroups $H_j\subset \Gamma\times S^1$. 
 
\begin{thm}\label{thm:Hopf}  Let $(\a_o,\beta_o)$ be such that $(\a_o,\beta_o)\in \CC$ in Figure \ref{fig:bif} and $\ome_1$ be given by (\ref{eq:ome1}). If 
\[\ome_1=c_1(H_1)+c_2(H_2)+\cdots +c_q(H_q), 
\]
contains a non-zero coefficient $c_j\ne 0$ for some $(H_j)$, then there exists a bifurcating branch of oscillating states of symmetry at least $(H_j)$. 
\end{thm}
\begin{proof} Using equivariant degree theory, the bifurcation invariant is computed by (cf. \cite{AED})
\[\ome_1=\text{\tt showdegree}[\Gamma](k_0,k_1, \dots,k_r,t_0,t_1,\dots,t_s),
\]
where $k_i$'s are related to the {\it positive spectrum} of the right hand side of (\ref{eq:2_lin}) in the constant function space, and the $t_j$'s are the {\it crossing numbers} which are equal to either $m_j$ or $-m_j$, depending on the direction of the crossing of the critical characteristic roots. 

Consider (\ref{eq:2_lin}) in the constant function space. Then, 
\[\big(\tau f'(0)\id+\tau \ka C\big) x=0,\q x\in \br^n.\] 
The positive spectrum  $\sig_+$ of the linear operator ($\tau f'(0)\id+\tau \ka C$) is 
\begin{align*}
\sig_+&=\{\tau f'(0)+\tau \ka \xi\,:\, \tau f'(0)+\tau \ka \xi>0,\q\xi\in \sig(C)\}=\{\a+\beta(\xi)\,:\, \a+\beta(\xi)>0,\q\xi\in \sig(C)\},
\end{align*}
which is an empty set, since the curve $\CC$ lies in the area $\a+\beta<0$.  Since the integer $k_i$ is the total number of copies of $\VV_i$ in $E(\mu)$ for $\mu\in\sig_+$, we have $k_i=0$ for all $i=0,1,\dots,r$. 

The crossing numbers are positive if the critical characteristic roots cross from the right to the left of the complex plane; and negative otherwise. As $(a,\beta)$ crosses $\CC$ at $(\a_o,\beta_o)$ from the shaded region in Figure \ref{fig:bif}, the count of characteristic roots with positive real part increases by $2$, thus all nonzero $t_j$'s are negative and equal to $-m_j$. 
\end{proof}
  
Theorem \ref{thm:Hopf} gives an existence result of bifurcating branches together with their {\it least} symmetry. To sharpen to the {\it precise} symmetry, one can work with orbit types that satisfy certain maximal condition. Here, we recall the concept of {\it dominating orbit types} from \cite{AED} and introduce a new complementing definition of {\it secondary dominating orbit types}.

\begin{Def}\rm \label{def:dom_sec_dom}
Let $\{\UU_1,\UU_2,\dots,\UU_m\}$ be the set of irreducible $\Gamma$-representations that occur in $\bc^n$, where $\bc^n$ is the complexification of the phase space $\br^n$ of the system (\ref{eq:1}). Let $\bar\UU_j$ be the $\Gamma\times S^1$-representation induced from $\UU_j$, for $j=1,2,\dots,m$ (cf. (\ref{eq:gamma_ext_s1})). 
Collect maximal orbit types from $\bar\UU_j$, for $j=1,2,\dots,m$, and denote this collection by  $\MM$. An orbit type $(H)\in\MM$ is called {\it dominating} if $(H)$ is maximal in $\MM$. 
A non-dominating orbit type $(L)\in \MM$ is  called {\it secondary dominating} if all orbit types $(H)\in \MM$ satisfying $(L)<(H)$ are dominating. 
\END
\end{Def}   

\begin{proposition}\label{cor:Hopf} Let $(\a_o,\beta_o)$ be such that $(\a_o,\beta_o)\in \CC$ in Figure \ref{fig:bif} and $\xi_o$ be the corresponding eigenvalue of $C$ given by (\ref{eq:hopf_xi}). Assume that $\ome_1$ defined by (\ref{eq:ome1}) contains $(H)$ with a nonzero coefficient. Then the following hold:
\begin{itemize}
\item[(i)] If $(H)$ is a dominating orbit type, then there exists a bifurcating branch of oscillating states of symmetry precisely equal to $(H)$. 
\item[(ii)] Suppose that $(H)$ is a secondary dominating orbit type, and for every dominating orbit type $(\tilde H)$ with $(H)<(\tilde H)$ there exists a flow-invariant subspace $S\subset\br^n$ such that
\begin{itemize}
\item[(a)] $S$ contains every state of symmetry $\tilde H$; and 
\item[(b)] $\xi_o\not\in\sig(C|_{S})$.
\end{itemize}
Then there exists a bifurcating branch of oscillating states of symmetry precisely being $(H)$.
\end{itemize}
\end{proposition}

\begin{proof} Statement (i) follows from \cite{AED}, and (ii) follows from the implicit function theorem, in the same spirit as Corollary \ref{cor:steady}. More precisely, let $(H)$ be a secondary dominating orbit type with a nonzero coefficient in $\ome_1$. 
By Theorem \ref{thm:Hopf}, there exists a a bifurcating branch of oscillating states of symmetry at least $(H)$. 
Let $(\tilde H)$ be the precise symmetry of this branch and suppose that $(H)<(\tilde H)$.  
By definition of secondary dominating orbit types, the only orbit types that are strictly larger than $(H)$ are dominating orbit types. Thus $(\tilde H)$ is dominating, and so there exists a flow-invariant subspace $S$ in $\br^n$ satisfying (a)--(b). 
Consider the restricted flow on $S$. The bifurcating branch of oscillating states, by condition (a),  is contained in $S$. However, by condition (b) and the implicit function theorem, there can be no bifurcation taking place in $S$. This leads to a contradiction. 

\end{proof}
%

\section{Bidirectional Ring Configuration}  \label{sec:ring} 
In this section, we study the bifurcations of the system (\ref{eq:1}) on a particular class of networks, namely bidirectional ring configurations. 
That is, we assume $g_i$'s satisfy (\ref{eq:sym}) for $\Gamma=D_n$. 
If the system has one of the specific forms \eqref{neural}--(\ref{eq-c}), this assumption can be weakened to (\ref{eq:cij}). In either case, the coupling matrix $C$ in (\ref{lin-eq}) satisfies (\ref{eq:cij}), which in case of dihedral configuration implies that $C$ is a {\it circulant matrix}\footnote{Recall that an $n\times n$ matrix is called {\it circulant} if every row is the right shift of the previous row (mod $n$). A circulant matrix $C=(c_{ij})$ is also denoted by $\cir[c_{11},c_{12},\dots, c_{1n}]$ using the entries of its first row.} with $c_{1j}=c_{1,(n+2-j)}$ for $1\le j\le n$. In particular, $C$ is a symmetric matrix.


A circulant matrix with first row entries $c_0,c_1,\dots, c_{n-1}$  has eigenvalues 
\begin{equation}  \label{circulant-eigenvalues}
	\xi_j=c_0+c_1 \varrho_j+c_2\varrho_j^2+\cdots +c_{n-1}\varrho_j^{n-1},
	\q j=0,1,2,\dots, n-1,
\end{equation}
with corresponding eigenvectors $v_j=(1,\varrho_j,\varrho_j^2,\dots, \varrho_j^{n-1})^T$, 
where $\varrho_j=\exp(2\pi ij/n)$ are the $n$-th roots of unity. Moreover, if the circulant matrix is $D_n$-symmetric, then  $\xi_{j}=\xi_{n-j}$ for $0< j<n$, which is essentially induced by the $D_n$-symmetry. In fact,  we have 
\begin{equation}\label{eq:E_xi}
\begin{cases}
E(\xi_0)=\VV_0,\\
E(\xi_j)=E(\xi_{n-j})=\VV_j\q\text{for $0< j<\frac n2$}\\
E(\xi_{\frac n2})=\VV_{(\frac n2+2)},\q \text{if $n$ is even}
\end{cases}
\end{equation}
(see Example \ref{ex:d12_rep} for notations $\VV_j$). An eigenvalue $\xi\in \sig(C)$ is called {\it simple} if $E(\xi)$ is irreducible. To a critical eigenvalue $\xi_o$, we associate an index set
\begin{equation}\label{eq:I}
I=\{i\,:\, \xi_i=\xi_o\}
\end{equation} 
(in case $n$ is even and $\xi_{\frac n2}=\xi_o$, we put $\frac n2+2$ into $I$ instead of $\frac n2$), which collects all indices of irreducible representations that have to do with the critical eigenvalue $\xi_o$.
%

\subsection{Steady-State Bifurcations for Bidirectional Rings}
Recall that $D_n$ acts on the phase space $\br^{n}$ by
\begin{gather}
\eta (x_1,x_2,\dots, x_{n})=(x_{n},x_1,x_2,\dots, x_{n-1})\label{eq:natl_eta_b}\\
\varsigma (x_1,x_2,\dots, x_{n}) =(x_{n}, x_{n-1},\dots, x_1),\label{eq:natl_ka_b}
\end{gather}
for $x=(x_1,x_2,\dots, x_{n})\in \br^{n}$. Using characters of representations, $\br^{n}$ can be decomposed into irreducible representations of $D_n$. In case of even $n$, we have
\begin{equation}\label{eq:r_even}
\br^{n}=\VV_0\times \VV_1\times \VV_2\times \cdots\times \VV_{\frac n2-1}\times \VV_{\frac n2+2}
\end{equation}
and in case of odd $n$, we have
\begin{equation}\label{eq:r_odd}
\br^{n}=\VV_0\times \VV_1\times \VV_2\times \cdots\times \VV_{\frac {n-1}2},
\end{equation}
(see Example \ref{ex:d12_rep} for notations $\VV_j$). It follows that the non-zero $n_i$'s in (\ref{eq:ome0}) are (cf. (\ref{eq:Vi_ni})) 
\begin{equation}\label{eq:d12_ni_b}
\begin{cases}
n_0=n_1=n_2=\cdots= n_{ \frac{n}{2}-1 }=n_{\frac n2+2}=1,\q \text{if $n$ is even,}\\
n_0=n_1=n_2=\cdots= n_{ \frac{n-1}{2} }=1,\q \text{if $n$ is odd.}
\end{cases}
\end{equation}
The integers $u_i$'s in (\ref{eq:ome0}) are determined by the critical eigenvalue $\xi_o$ and the corresponding $I$ (cf. (\ref{eq:I})). 
Based on (\ref{eq:E_xi}) and the definition (\ref{eq:ui}) of $u_i$, the non-zero $u_i$'s are
\begin{equation}\label{eq:d12_ui_b}
\begin{cases}
u_i=1,\q \text{for $i\in\{0,1,2,\dots, \frac n2-1,\frac n2+2\}\setminus I$,\q if $n$ is even,}\\
u_i=1,\q \text{for $i\in\{0,1,2,\dots, \frac {n-1}{2}\}\setminus I$,\q if $n$ is odd.}
\end{cases}
\end{equation}
Thus, the bifurcation invariant $\ome_0$ can be computed using (\ref{eq:ome0}), accompanied by (\ref{eq:d12_ni_b})--(\ref{eq:d12_ui_b}).

\begin{ex}\rm \label{ex:sb_dn}(Simple critical eigenvalues for bidirectional rings.) Let $C$ be a coupling matrix satisfying (\ref{eq:cij}) for $\Gamma=D_n$. 
Then $C$ is determined by ($\frac n2+1$) or ($\frac{n+1}{2}$) different entries depending on whether $n$ is even or odd, respectively. These entries decide which eigenvalue is maximal. Let $\xi_o\in \sig(C)$ be the maximal eigenvalue. 
Assume that $\xi_o$ is simple, i.e., $E(\xi_o)$ is irreducible. 
Then the index set $I$ is a singleton and there are only  possibly $\frac n2$ or $\frac {n-1}2$ different  values of  $\ome_0$, depending on whether $n$ is even or odd. As an example, for $n=12$, we have
{\small
\begin{equation}\label{eq:ome_0_12}
\ome_0=
\begin{cases}
-2(D_{12})+2(\tilde D_6)+4(D_4)-2(\tilde D_3)+2(D_3)-2(\tilde D_2)-2(D_2)-2( \bz_4)+2(\bz_2),\q\text{if $\xi_o=\xi_0$}\\
(D_1)-(\tilde D_1),\q\hskip8.4cm\text{if $\xi_o=\xi_1$}\\
-(D_2)+(\tilde D_2)+2(D_1)-2(\tilde D_1),\q\hskip6cm\text{if $\xi_o=\xi_2$}\\
-(\tilde D_3)+(D_3),\q\hskip8.2cm\text{if $\xi_o=\xi_3$}\\
2(D_4)-2(D_2)-(\bz_4)+(\bz_2)-2(\tilde D_1)+2(D_1),\q\hskip4cm\text{if $\xi_o=\xi_4$}\\
-(\tilde D_1)+(D_1),\q\hskip8.2cm\text{if $\xi_o=\xi_5$}\\
(\tilde D_6)-2(\tilde D_3)+(\bz_3), \q\hskip7.3cm\text{if $\xi_o=\xi_6$}
\end{cases} 
\end{equation} }
These values, combined with fixed point subspaces of subgroups of $D_{12}$ (cf. Table \ref{table:eig_fix}), lead to the classification result summarized in Table \ref{t:d12_sb}.
To illustrate, in case $\xi_o=\xi_1$, we have two orbit types $(D_1)$ and $(\tilde D_1)$ with non-zero coefficients in $\ome_0$. Using Table \ref{table:eig_fix}, we have that  $\xi_1\not\in \sig(C|_{\fix(H)})$ for all $H>D_1$, thus by Corollary \ref{cor:steady}, there exists at least one bifurcating branch of steady states of symmetry precisely $(D_1)$. Since $(D_1)$ consists of $6$ isotropy subgroups: $\eta^k D_1 \eta^{-k}$ for $k=0,1,\dots, 5$, we derive the form of the solution for each of these isotropies. The same can be applied to $(\tilde D_1)$. 

\renewcommand{\arraystretch}{1.15}
\begin{table}[!htb]
{\small
\hskip-.5cm
\begin{tabular}{|c|c|c|}
\hline
$K$ & $\fix(K)$ & $\sig(C|_{\fix(K)})$\\
\hline && \\[-1.2em]\hline
$D_{12}$& $\{x_1=x_2=\cdots=x_{12}\}$ & $\xi_0$\\
\hline
$D_{6}$& $\{x_1=x_2=\cdots=x_{12}\}$  & $\xi_0$\\
\hline
$\tilde D_{6}$ & $\{x_1=x_3=\cdots=x_{11}, x_2=x_4=\cdots=x_{12}\}$  & $\xi_0,\xi_6$\\
\hline
$\bz_{6}$& $\{x_1=x_3=\cdots=x_{11}, x_2=x_4=\cdots=x_{12}\}$ & $\xi_0,\xi_6$\\
\hline
$D_{4}$ & $\{x_1=x_3=x_4=x_6=x_7=x_9=x_{10}=x_{12}, x_2=x_5=x_8=x_{11}\}$ & $\xi_0,\xi_4$\\
\hline
$\bz_{4}$ & $\{x_1=x_4=x_7=x_{10}, x_2=x_5=x_8=x_{11},x_3=x_6=x_9=x_{12}\}$ & $\xi_0,\xi_4,\xi_4$\\
\hline
$D_{3}$ & $\{x_1=x_4=x_5=x_8=x_9=x_{12}, x_2=x_3=x_6=x_7=x_{10}=x_{11}\}$ & $\xi_0, \xi_3$\\
\hline
$\tilde D_{3}$& $\{x_1=x_3=x_5=x_7=x_9=x_{11}, x_2=x_6=x_{10}, x_4=x_{8}=x_{12}\}$ & $\xi_0,\xi_3,\xi_6$\\
\hline
$\bz_{3}$& $\{x_1=x_5=x_9, x_2=x_6=x_{10}, x_3=x_7=x_{11}, x_4=x_{8}=x_{12}\}$ & $\xi_0,\xi_3,\xi_3,\xi_6$\\
\hline
$D_{2}$ & $\{x_1=x_6=x_7=x_{12}, x_2=x_5=x_8=x_{11}, x_3= x_4=x_{9}=x_{10}\}$ & $\xi_0, \xi_2,\xi_4$\\
\hline
$\tilde D_{2}$& $\{x_1=x_5=x_7=x_{11}, x_2=x_4=x_8=x_{10}, x_3= x_{9},x_6=x_{12}\}$ & $\xi_0,\xi_2,\xi_4,\xi_6$\\
\hline
$\bz_{2}$ & $\{x_1=x_7, x_2=x_8,x_3= x_{9},x_4=x_{10},x_5=x_{11},x_6=x_{12}\}$ & $\xi_0,\xi_2, \xi_2,\xi_4,\xi_4,\xi_6$\\
\hline
$D_{1}$ & $\{x_1=x_{12}, x_2=x_{11},x_3= x_{10},x_4=x_{9},x_5=x_{8},x_6=x_{7}\}$  & $\xi_0, \xi_1,\xi_2,\xi_3,\xi_4,\xi_5$\\
\hline
$\tilde D_{1}$& $\{x_1=x_{11}, x_2=x_{10},x_3= x_{9},x_4=x_{8},x_5=x_{7}\}$ & $\xi_0,\xi_1,\xi_2,\xi_3,\xi_4,\xi_5,\xi_6$\\
\hline
$\bz_{1}$& $\br^{12}$ & $\xi_0,\xi_1,\xi_1,\xi_2, \xi_2,\xi_3,\xi_3,\xi_4,\xi_4,\xi_5,\xi_5,\xi_6$\\
\hline
\end{tabular}
}
\caption{Fixed point subspaces of $K\subset D_{12}$ and eigenvalues of the coupling matrix $C|_{\fix(K)}:\fix(K)\to  \fix(K)$ (up to conjugacy classes of subgroups).} \label{table:eig_fix}
\end{table}

\renewcommand\arraystretch{.72}  
\begin{table}
\hskip.5cm
\begin{tabular}{|c|c|l|l|}
\hline 
\multirow{4}{*}{\bf Critical Eigenvalue}& \multirow{4}{*}{\bf Symmetry }& \multirow{2}{*}{{\bf Form of Bifurcating Steady-States} } & \multirow{4}{*}{{\bf Figure} }  \\
&&&\\
& & \multirow{2}{*}{(for distinct $a,b,c,d,e,f,g\in \br$)}& \\
&&&\\
\hline \hline 
\multirow{6}{*}{$\xi_0$}
& \multirow{6}{*}{$D_{12}$} & \multirow{6}{*}{$(a,a,a,a,a,a,a,a,a,a,a,a)$}&  \multirow{6}{*}{ \includegraphics[width=1.8cm]{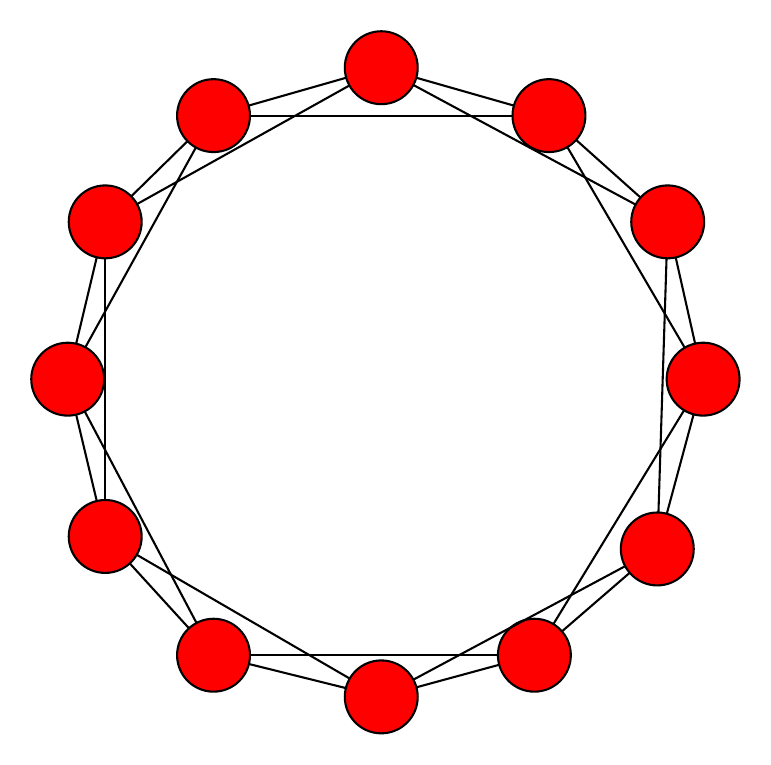}}\\
&&&\\
&&&\\ 
&&&\\
&&&\\ 
&&&\\ 
\hline
\multirow{24}{*}{$\xi_1$ or $\xi_5$}
& \multirow{2}{*}{$D_{1}$} & \multirow{2}{*}{$(a,b,c,d,e,f,f,e,d,c,b,a)$}&  \multirow{12}{*}{\includegraphics[width=2cm]{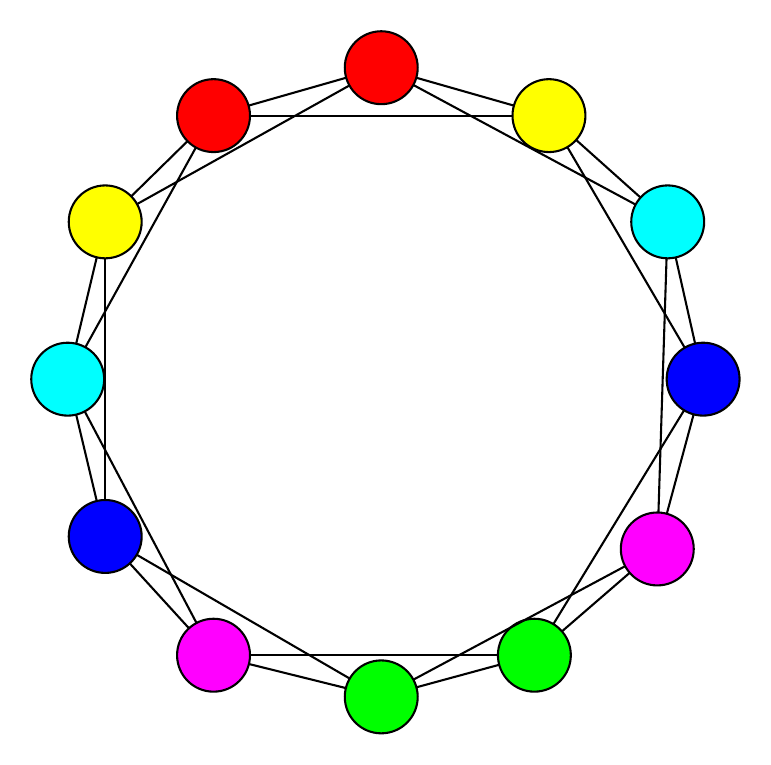}}  \\
&&&\\
\cline{2-3}  
& \multirow{2}{*}{$\eta D_{1} \eta^{-1}$} & \multirow{2}{*}{$(a,a,b,c,d,e,f,f,e,d,c,b)$}&  \\
&&&\\
\cline{2-3}
& \multirow{2}{*}{$\eta^2 D_{1} \eta^{-2}$} & \multirow{2}{*}{$(b,a,a,b,c,d,e,f,f,e,d,c)$}& \\
&&&\\
\cline{2-3}
& \multirow{2}{*}{$\eta^3 D_{1} \eta^{-3}$} & \multirow{2}{*}{$(c,b,a,a,b,c,d,e,f,f,e,d)$}& \\
&&&\\
\cline{2-3}
& \multirow{2}{*}{$\eta^4 D_{1} \eta^{-4}$} & \multirow{2}{*}{$(d,c,b,a,a,b,c,d,e,f,f,e)$}& \\
&&&\\
\cline{2-3}
& \multirow{2}{*}{$\eta^5 D_{1} \eta^{-5}$} & \multirow{2}{*}{$(e, d,c,b,a,a,b,c,d,e,f,f)$}& \\
&&&\\ 
\cline{2-4}\cline{2-4}
& \multirow{2}{*}{$\tilde D_{1}$} & \multirow{2}{*}{$(a,b,c,d,e,f,e,d,c,b,a,g)$}&  \multirow{12}{*}{\includegraphics[width=2cm]{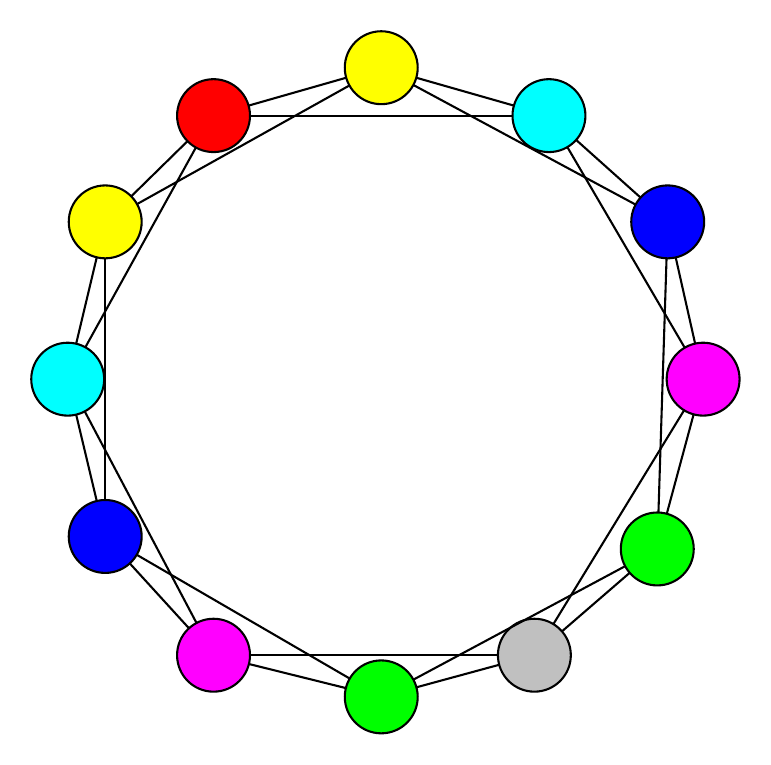}}  \\ 
&&&\\
\cline{2-3}
& \multirow{2}{*}{$\eta \tilde D_{1} \eta^{-1}$} & \multirow{2}{*}{$(g,a,b,c,d,e,f,e,d,c,b,a)$}&  \\
&&&\\
\cline{2-3}
& \multirow{2}{*}{$\eta^2 \tilde D_{1} \eta^{-2}$} & \multirow{2}{*}{$(a,g,a,b,c,d,e,f,e,d,c,b)$}& \\
&&&\\
\cline{2-3}
& \multirow{2}{*}{$\eta^3 \tilde D_{1} \eta^{-3}$} & \multirow{2}{*}{$(b,a,g,a,b,c,d,e,f,e,d,c)$}& \\
&&&\\
\cline{2-3}
& \multirow{2}{*}{$\eta^4 \tilde D_{1} \eta^{-4}$} & \multirow{2}{*}{$(c,b,a,g,a,b,c,d,e,f,e,d)$}& \\
&&&\\
\cline{2-3}
& \multirow{2}{*}{$\eta^5 \tilde D_{1} \eta^{-5}$} & \multirow{2}{*}{$(d,c,b,a,g,a,b,c,d,e,f,e)$}& \\
&&&\\ 
\hline
\multirow{12}{*}{$\xi_2$}
& \multirow{2}{*}{$D_{2}$} & \multirow{2}{*}{$(a,b,c,c,b,a,a,b,c,c,b,a)$}&  \multirow{6}{*}{\includegraphics[width=1.85cm]{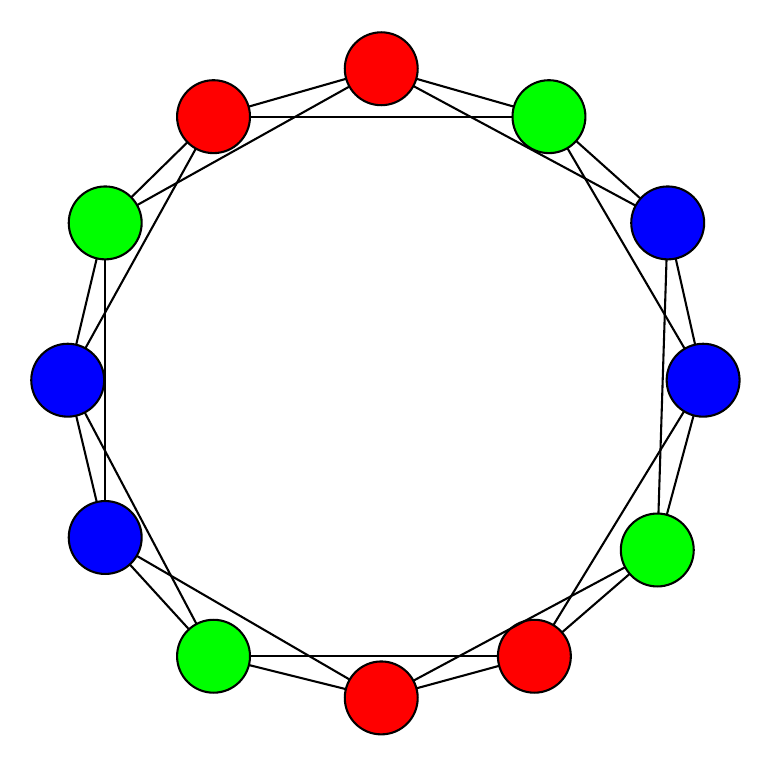}}  \\
&&&\\
\cline{2-3}
& \multirow{2}{*}{$\eta D_{2}\eta^{-1}$} & \multirow{2}{*}{$(a,a,b,c,c,b,a,a,b,c,c,b)$}&  \\
&&&\\
\cline{2-3}
& \multirow{2}{*}{$\eta^2 D_{2}\eta^{-2}$} & \multirow{2}{*}{$(b,a,a,b,c,c,b,a,a,b,c,c)$}&  \\
&&&\\
\cline{2-4}
& \multirow{2}{*}{$\tilde D_{2}$} & \multirow{2}{*}{$(a,b,c,b,a,d,a,b,c,b,a,d)$}&  \multirow{6}{*}{\includegraphics[width=1.85cm]{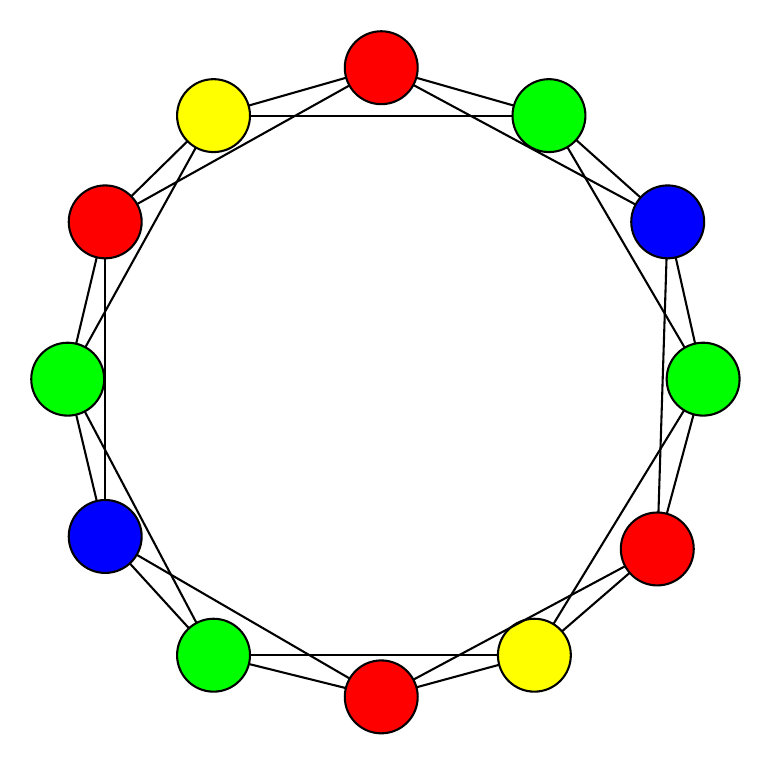}}  \\
&&&\\
\cline{2-3}
& \multirow{2}{*}{$\eta \tilde D_{2}\eta^{-1}$} & \multirow{2}{*}{$(d,a,b,c,b,a,d,a,b,c,b,a)$}&  \\
&&&\\
\cline{2-3}
& \multirow{2}{*}{$\eta^2 \tilde D_{2}\eta^{-2}$} & \multirow{2}{*}{$(a,d,a,b,c,b,a,d,a,b,c,b)$}&  \\
&&&\\
\hline
\multirow{12}{*}{$\xi_3$}
& \multirow{3}{*}{$D_{3}$} & \multirow{3}{*}{$(a,b,b,a,a,b,b,a,a,b,b,a)$}&  \multirow{6}{*}{\includegraphics[width=1.85cm]{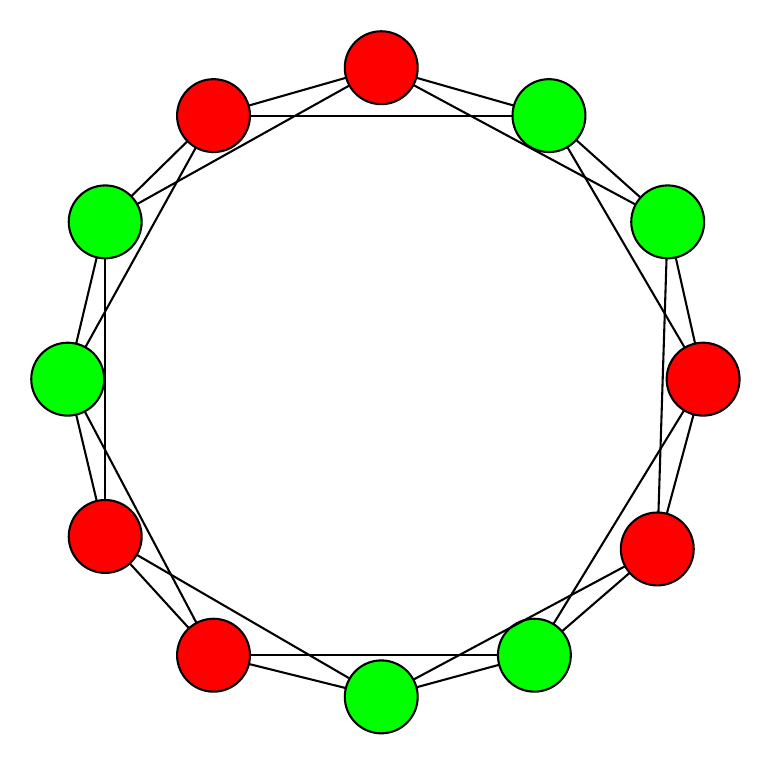}}  \\
&&&\\
&&&\\
\cline{2-3}
& \multirow{3}{*}{$\eta D_{3}\eta^{-1}$} & \multirow{3}{*}{$(a,a,b,b,a,a,b,b,a,a,b,b)$}&  \\
&&&\\
&&&\\
\cline{2-4}
& \multirow{3}{*}{$\tilde D_{3}$} & \multirow{3}{*}{$(a,b,a,c,a,b,a,c,a,b,a,c)$}&  \multirow{6}{*}{\includegraphics[width=1.85cm]{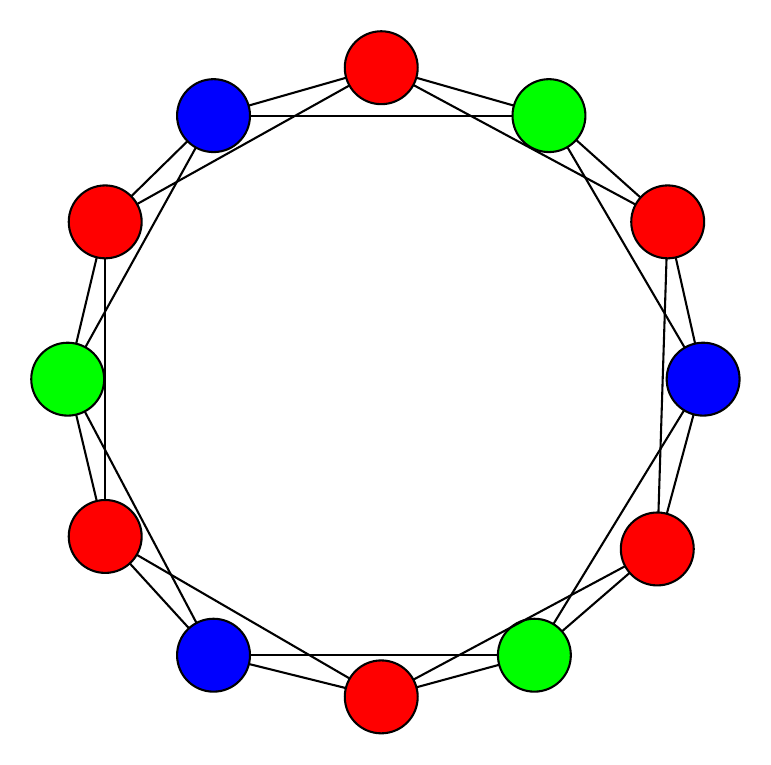}}  \\ 
&&&\\
&&&\\
\cline{2-3}
& \multirow{3}{*}{$\eta \tilde D_{3}\eta^{-1}$} & \multirow{3}{*}{$(c,a,b,a,c,a,b,a,c,a,b,a)$}&  \\  
&&&\\
&&&\\
\hline
\multirow{6}{*}{$\xi_4$}
& \multirow{2}{*}{$D_{4}$} & \multirow{2}{*}{$(a,b,a,a,b,a,a,b,a,a,b,a)$}&  \multirow{6}{*}{\includegraphics[width=1.85cm]{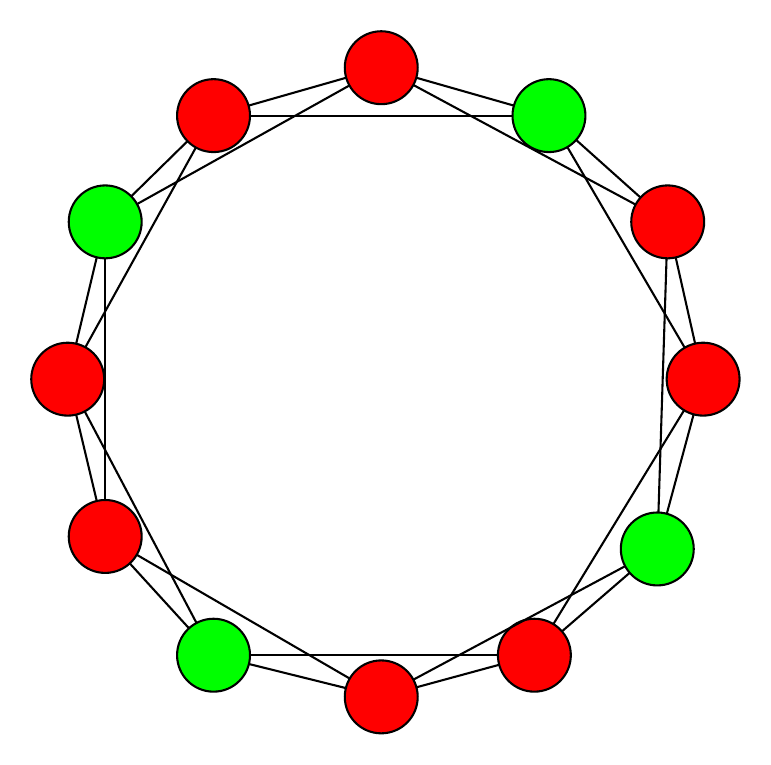}}   \\
&&&\\
\cline{2-3}
& \multirow{2}{*}{$\eta D_{4}\eta^{-1}$} & \multirow{2}{*}{$(a,a,b,a,a,b,a,a,b,a,a,b)$}& \\
&&&\\
\cline{2-3}
& \multirow{2}{*}{$\eta^2 D_{4}\eta^{-2}$} & \multirow{2}{*}{$(b,a,a,b,a,a,b,a,a,b,a,a)$}& \\
&&&\\
\hline
\multirow{6}{*}{$\xi_6$}
& \multirow{6}{*}{$\tilde D_{6}$} & \multirow{6}{*}{$(a,b,a,b,a,b,a,b,a,b,a,b)$}&  \multirow{6}{*}{\includegraphics[width=1.85cm]{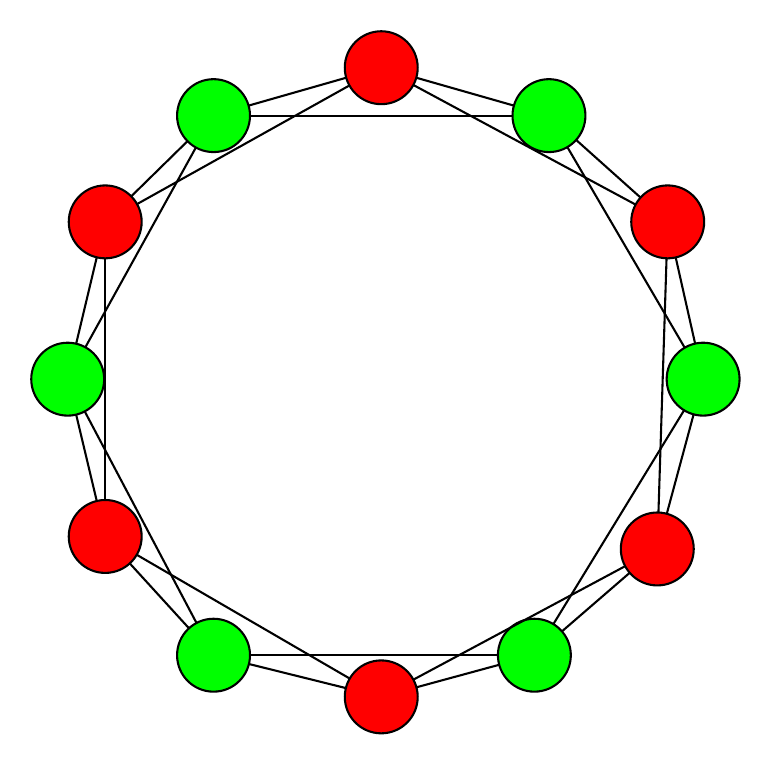}}  \\
&&&\\
&&&\\ 
&&&\\
&&&\\
&&&\\
\hline 
\end{tabular}
\caption{Summary of distinct forms of steady states bifurcating from the equilibrium $x=0$ of the system (\ref{eq:1}) for $n=12$.} \label{t:d12_sb}
\end{table} 

Note that the possible values of $\ome_0$ do not depend on the entries of $C$ directly, but rather on the maximal eigenvalue. For example, if every cell is connected only with its $2$ nearest neighbors, then $\xi_o=\xi_0$ if the coupling is excitatory; and $\xi_o=\xi_6$ if it is inhibitory. 
That is, this configuration does not allow $\xi_o$ to be $\xi_i$ for $i\in\{1,2,3,4,5\}$. However, if each cell $i$ is connected to its $4$ nearest neighbors, with coupling strength $d_1$ to cells $(i\pm 1)$ and with strength $d_2$ to $(i\pm 2)$, then every eigenvalue can be maximal for some choices of $d_1,d_2$. See Figure~\ref{F:max_eig_G12} for their precise relation.
 
\END  
\end{ex}
Besides those values listed in (\ref{eq:ome_0_12}), $\ome_0$ can take other values if $\xi_o$ is non-simple.  For example, the coupling configuration with four nearest neighbors allows double critical eigenvalues as shown in Figure \ref{F:max_eig_G12}, when the relation between $d_1,d_2$ follows one of the lines there. In this case, one can work out the index set $I$ and compute $\ome_0$ individually. The same result using Theorem \ref{thm:steady} and Corollary \ref{cor:steady} applies.

\subsection{Hopf Bifurcations for Bidirectional Rings}

The complexification of $E(\xi_j)$ for $\xi_j\in \sig(C)$ satisfies
\begin{equation}\label{eq:E_xi_c}
\begin{cases}
E^c(\xi_0)=\UU_0,\\
E^c(\xi_j)=E^c(\xi_{n-j})=\UU_j\q\text{for $0< j<\frac n2$}\\
E^c(\xi_{\frac n2})=\UU_{(\frac n2+2)},\q \text{if $n$ is even}
\end{cases}
\end{equation}
(see Example \ref{ex:d12_rep} for notations $\UU_j$). It follows that the non-zero integers $m_j$'s in (\ref{eq:ome1}) are 
\begin{equation}\label{eq:d12_mj_b}
m_j=1,\q \text{for $j\in I$.}
\end{equation}
where $I$ is given by (\ref{eq:I}). The bifurcation invariant $\ome_1$ can then be computed using (\ref{eq:ome1}) together with (\ref{eq:d12_mj_b}).

\begin{ex}\rm  (Simple critical eigenvalues for bidirectional rings.) Following Example \ref{ex:sb_dn}, we take $C$ that satisfies (\ref{eq:cij}) with $\Gamma=D_n$. 
The ($\frac n2+1$) or ($\frac{n+1}{2}$) different entries of $C$ decide which eigenvalue is minimal. Let $\xi_o\in \sig(C)$ be the minimal eigenvalue. Assume that $\xi_o$ is simple. 
Then the index set $I$ is a singleton and there are only $\frac n2$ or $\frac {n-1}2$ different  values of  $\ome_0$, depending on whether $n$ is even or odd, respectively. Again for $n=12$, we have
{\[\ome_1=
\begin{cases}
-(D_{12}),\q\hskip4cm\text{if $\xi_o=\xi_0$}\\
-(\bz_{12}^{t_1})-(D_2^d)-(\tilde D_2^d)+(\bz_2^d),\q\hskip.8cm\text{if $\xi_o=\xi_1$}\\
-(\bz_{12}^{t_2})-(D_4^d)-(D_4^{\hat d})+(\bz_4^d),\q\hskip.8cm\text{if $\xi_o=\xi_2$}\\
-(\bz_{12}^{t_3})-(D_6^d)-(\tilde D_6^d)+(\bz_6^d),\q\hskip.8cm\text{if $\xi_o=\xi_3$}\\
-(\bz_{12}^{t_4})-(D_4^z)-(D_4)+(\bz_4), \q\hskip.8cm\text{if $\xi_o=\xi_4$}\\
-(\bz_{12}^{t_5})-(D_2^d)-(\tilde D_2^d)+(\bz_2^d),\q\hskip.8cm\text{if $\xi_o=\xi_5$}\\
-(D_{12}^{\hat d}), \q\hskip4cm\text{if $\xi_o=\xi_6$}
\end{cases} 
\] }
To find dominating and secondary dominating orbit types, consider the maximal orbit types in $\UU_i$'s. They are \\
$(D_{12})$ in $\UU_0$;\\
$(\bz_{12}^{t_1})$, $(D_2^d)$  $(\tilde D_2^d)$ in $\UU_1$; \\
$(\bz_{12}^{t_2})$, $(D_4^d)$   $(D_4^{\hat d})$ in $\UU_2$; \\
$(\bz_{12}^{t_3})$, $(D_6^d)$,   $(\tilde D_6^d)$ in $\UU_3$; \\
$(\bz_{12}^{t_4})$, $(D_4^z)$   $(D_4)$ in $\UU_4$; \\
$(\bz_{12}^{t_5})$, $(D_2^d)$  $(\tilde D_2^d)$ in $\UU_5$; \\
$(D_{12}^{\hat d})$ in $\UU_6$. \\
Among these orbit types, we find the dominating orbit types: $(D_{12})$, $(D_{12}^{\hat d})$, $(\bz_{12}^{t_1})$, $(\bz_{12}^{t_2})$, $(\bz_{12}^{t_3})$, $(\bz_{12}^{t_4})$, $(\bz_{12}^{t_5})$, $(D_6^d)$, $(\tilde D_6^d)$, $(D_4^d)$, $(D_4^z)$ and the secondary dominating orbit types: $(D_2^d)$, $(\tilde D_2^d)$, $(D_4)$, $(D_4^{\hat d})$.
The values of $\ome_1$ together with the dominating and secondary dominating orbit types lead to the classification result summarized in Table \ref{t:d12_hb_1}--\ref{t:d12_hb_3} using Proposition \ref{cor:Hopf}.

\renewcommand\arraystretch{.7}   
\begin{table}
\hspace*{-1.2cm}
\hskip-1cm
\begin{tabular}{|c|c|l|c|}
\hline 
\multirow{2}{*}{\bf Critical }& \multirow{4}{*}{\bf Symmetry }& \multirow{2}{*}{{\bf Form of Oscillating-States} } & \multirow{4}{*}{{\bf Figure} }  \\
&&&\\
\multirow{2}{*}{\bf Eigenvalue} &&\multirow{2}{*}{(for some period $T$)}&\\
& & & \\
\hline \hline
\multirow{6}{*}{$\xi_0$}
& \multirow{6}{*}{$D_{12}$} & \multirow{6}{*}{$(x(t),x(t),x(t),\dots,x(t))$}&  \multirow{6}{*}{ \includegraphics[width=1.8cm]{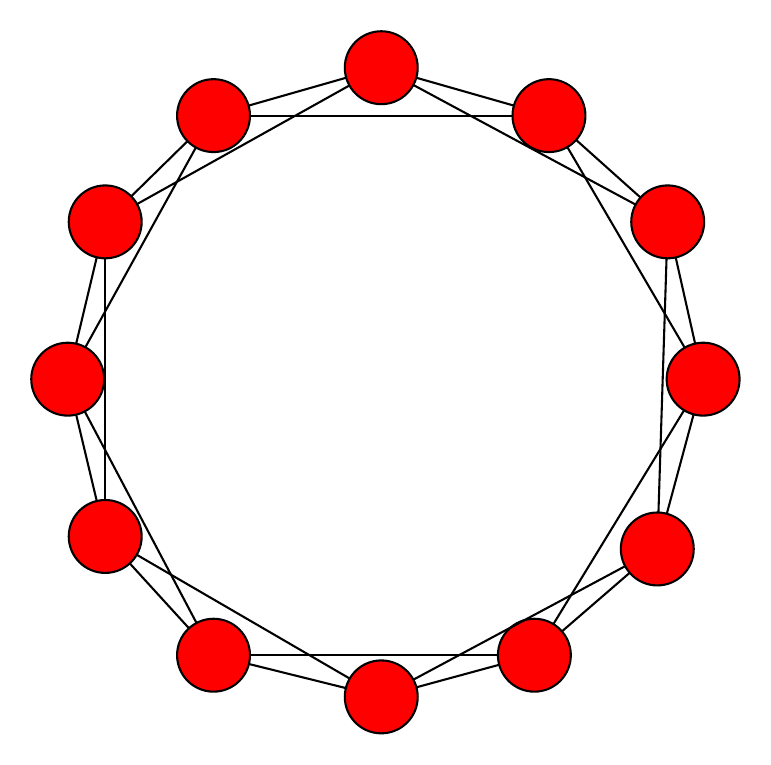}}\\
&&&\\
&&&\\
&&&\\
&&&\\
&&&\\
\hline

\multirow{12}{*}{$\xi_1$}
& \multirow{6}{*}{$\bz_{12}^{t_1}$} & \multirow{6}{*}{$(x_1(t),x_1(t+\frac{T}{12}),x_1(t+\frac{2T}{12})\dots, x_1(t+\frac{11T}{12}))$}&  \multirow{6}{*}{\includegraphics[width=4.6cm]{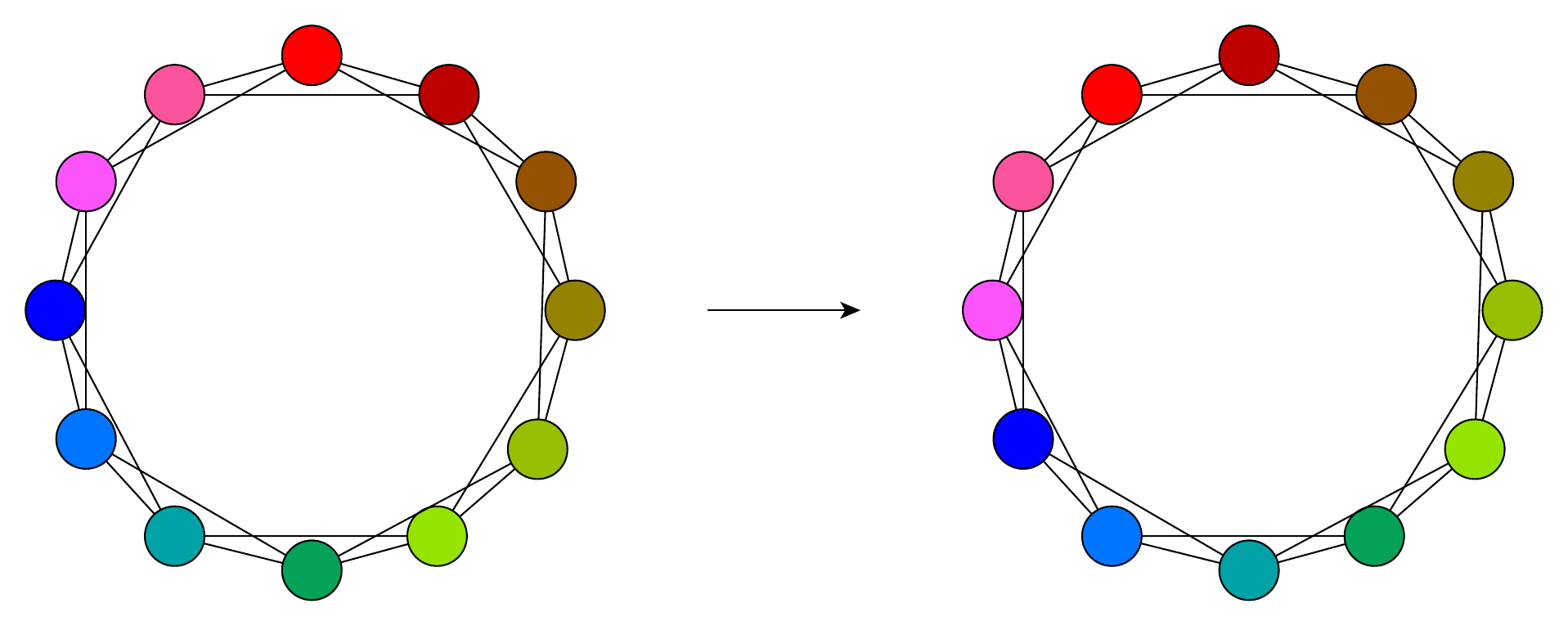}}  \\
&&&{\small $+\frac{T}{12}$}\\
&&&\\
&&&\\
&&&\\
&&&\\
\cline{2-4}
& \multirow{6}{*}{$\varsigma \bz_{12}^{t_1} \varsigma^{-1}$} & \multirow{6}{*}{$(x_1(t),x_1(t+\frac{11 T}{12}),x_1(t+\frac{10T}{12})\dots, x_1(t+\frac{T}{12}))$}&   \multirow{6}{*}{\includegraphics[width=4.6cm]{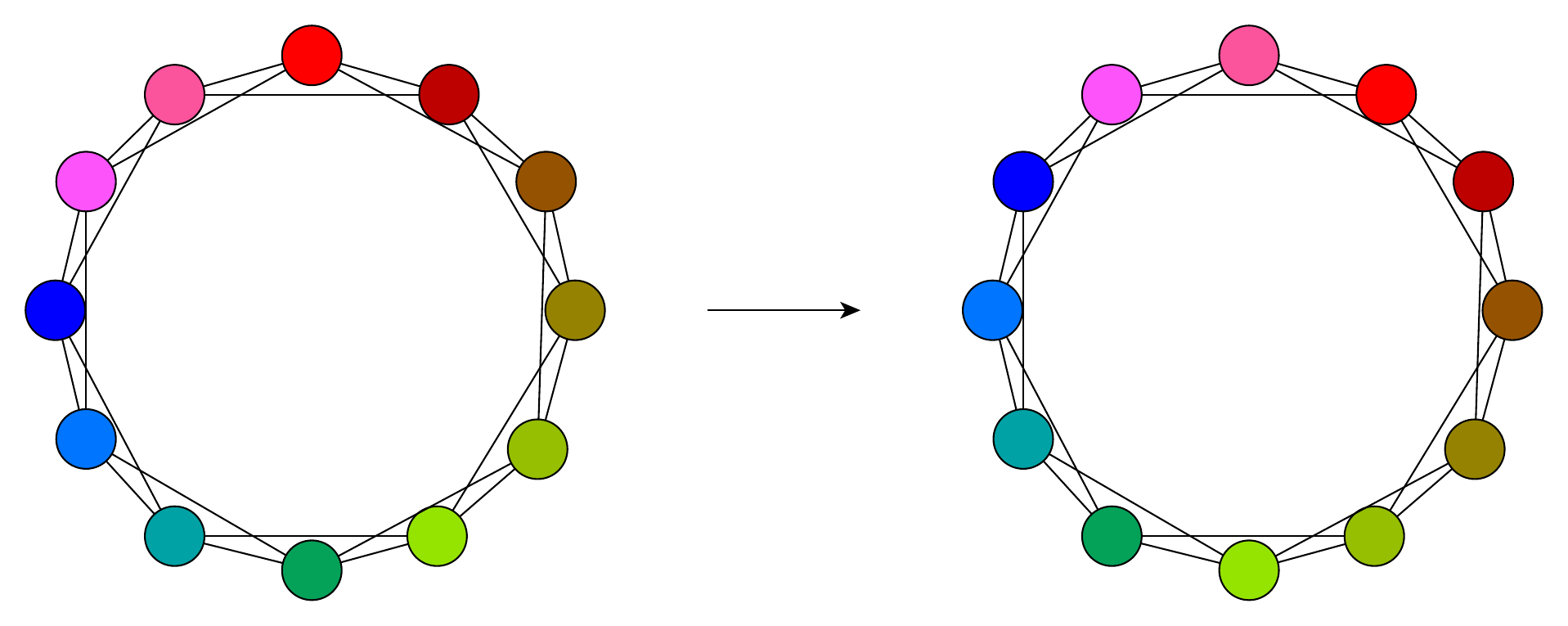}}  \\
&&&{\small $+\frac{T}{12}$}\\
&&&\\
&&&\\
&&&\\
&&&\\
\hline
\multirow{48}{*}{$\xi_1$ or $\xi_5$} & \multirow{4}{*}{$D_{2}^d$} & \multirow{2}{*}{$(x_1(t),x_2(t),x_3(t),x_3(t+\frac T2),x_2(t+\frac T2),x_1(t+\frac T2),$}&  \multirow{24}{*}{\includegraphics[width=5cm]{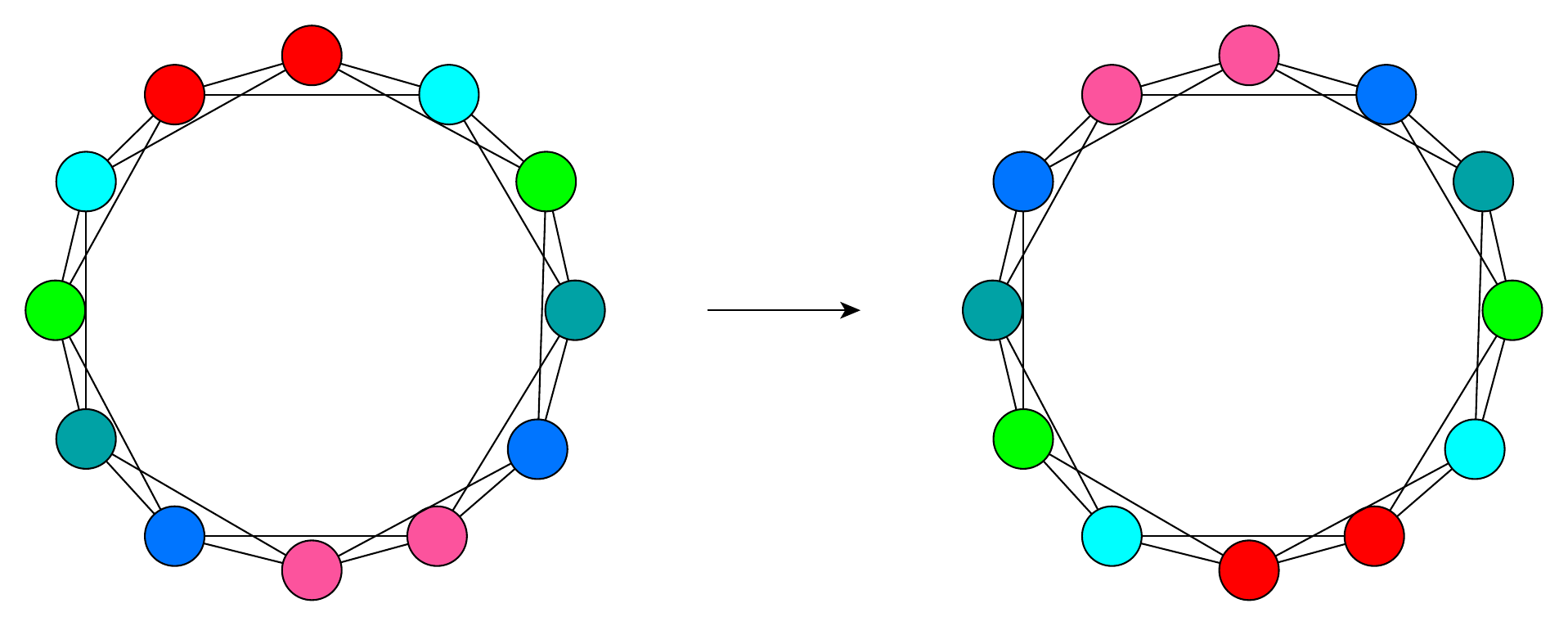}} \\ 
&&&\\
&&\multirow{2}{*}{$x_1(t+\frac T2),x_2(t+\frac T2),x_3(t+\frac T2),x_3(t),x_2(t),x_1(t)) $}&\\
&&&\\
\cline{2-3}
& \multirow{4}{*}{$\eta D_2^d \eta^{-1}$} & \multirow{2}{*}{$(x_1(t),x_1(t),x_2(t),x_3(t),x_3(t+\frac T2),x_2(t+\frac T2),$}& \\
&&&\\
&&\multirow{2}{*}{$x_1(t+\frac T2),x_1(t+\frac T2),x_2(t+\frac T2),x_3(t+\frac T2),x_3(t),x_2(t))$}&\\
&&&\\
\cline{2-3} 
& \multirow{4}{*}{$\eta^2 D_{2}^d \eta^{-2}$} & \multirow{2}{*}{$(x_2(t),x_1(t),x_1(t),x_2(t),x_3(t),x_3(t+\frac T2),x_2(t+\frac T2),$}& \\
&&&\\
&&\multirow{2}{*}{$x_1(t+\frac T2),x_1(t+\frac T2),x_2(t+\frac T2),x_3(t+\frac T2),x_3(t))$}& {\small $+\frac{T}{2}$}\\
&& &\\
\cline{2-3}
& \multirow{4}{*}{$\eta^3 D_{2}^d \eta^{-3}$} & \multirow{2}{*}{$(x_3(t),x_2(t),x_1(t),x_1(t),x_2(t),x_3(t),x_3(t+\frac T2),$}& \\
&&&\\ 
&&\multirow{2}{*}{$x_2(t+\frac T2),x_1(t+\frac T2),x_1(t+\frac T2),x_2(t+\frac T2),x_3(t+\frac T2))$}&\\
&&&\\
\cline{2-3}
& \multirow{4}{*}{$\eta^4 D_{2}^d \eta^{-4}$} & \multirow{2}{*}{$(x_3(t+\frac T2),x_3(t),x_2(t),x_1(t),x_1(t),x_2(t),x_3(t),$}& \\
&&&\\ 
&&\multirow{2}{*}{$x_3(t+\frac T2),x_2(t+\frac T2),x_1(t+\frac T2),x_1(t+\frac T2),x_2(t+\frac T2))$}&\\
&&&\\
\cline{2-3}
& \multirow{4}{*}{$\eta^5 D_{2}^d \eta^{-5}$} & \multirow{2}{*}{$(x_2(t+\frac T2), x_3(t+\frac T2),x_3(t),x_2(t),x_1(t),x_1(t),$}& \\
&&&\\ 
&&\multirow{2}{*}{$x_2(t),x_3(t),x_3(t+\frac T2),x_2(t+\frac T2),x_1(t+\frac T2),x_1(t+\frac T2))$}&\\
&&&\\

\cline{2-4}
& \multirow{4}{*}{$\tilde D_{2}^d$} & \multirow{2}{*}{$(x_1(t),x_2(t),x_3(t),x_2(t+\frac T2),x_1(t+\frac T2),x_4(t),$}&  \multirow{24}{*}{\includegraphics[width=5cm]{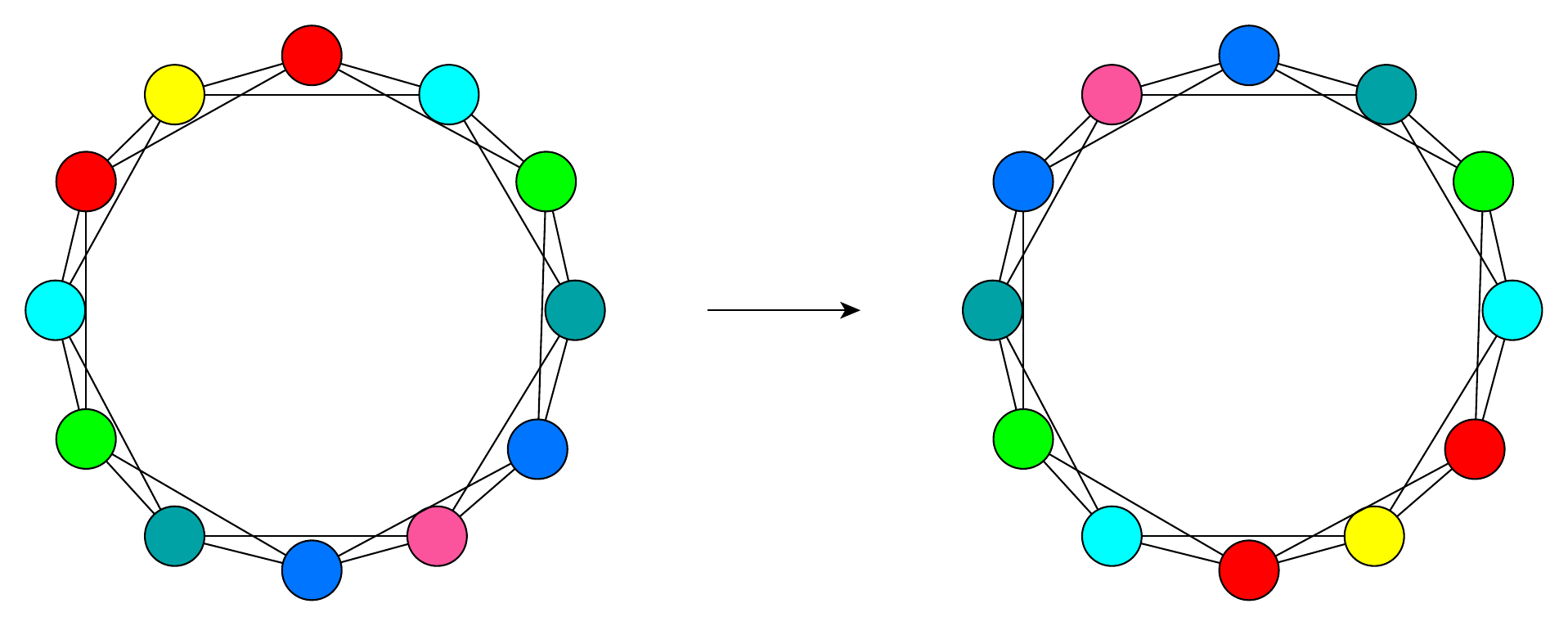}} \\
&&&\\
&&\multirow{2}{*}{$x_1(t+\frac T2),x_2(t+\frac T2),x_3(t),x_2(t),x_1(t),x_4(t+\frac T2)) $}&\\
&&&\\
\cline{2-3}

& \multirow{4}{*}{$\eta\tilde D_2^d \eta^{-1}$} & \multirow{2}{*}{$(x_4(t+\frac T2), x_1(t),x_2(t),x_3(t),x_2(t+\frac T2),x_1(t+\frac T2),$}& \\
&&&\\
&&\multirow{2}{*}{$x_4(t),x_1(t+\frac T2),x_2(t+\frac T2),x_3(t),x_2(t),x_1(t)) $}&\\
&&&\\
\cline{2-3} 
& \multirow{4}{*}{$\eta^2 \tilde D_{2}^d \eta^{-2}$} & \multirow{2}{*}{$(x_1(t),x_4(t+\frac T2), x_1(t),x_2(t),x_3(t),x_2(t+\frac T2),$}& \\
&&&\\
&&\multirow{2}{*}{$x_1(t+\frac T2),x_4(t),x_1(t+\frac T2),x_2(t+\frac T2),x_3(t),x_2(t))$}& {\small $+\frac{T}{2}$}\\
&& &\\
\cline{2-3}
& \multirow{4}{*}{$\eta^3 \tilde D_{2}^d \eta^{-3}$} & \multirow{2}{*}{$(x_2(t),x_1(t),x_4(t+\frac T2), x_1(t),x_2(t),x_3(t),x_2(t+\frac T2),$}& \\
&&&\\ 
&&\multirow{2}{*}{$x_1(t+\frac T2),x_4(t),x_1(t+\frac T2),x_2(t+\frac T2),x_3(t))$}&\\
&&&\\
\cline{2-3}
& \multirow{4}{*}{$\eta^4 \tilde D_{2}^d \eta^{-4}$} & \multirow{2}{*}{$(x_3(t),x_2(t),x_1(t),x_4(t+\frac T2), x_1(t),x_2(t),x_3(t),$}& \\
&&&\\ 
&&\multirow{2}{*}{$x_2(t+\frac T2),x_1(t+\frac T2),x_4(t),x_1(t+\frac T2),x_2(t+\frac T2))$}&\\
&&&\\
\cline{2-3}
& \multirow{4}{*}{$\eta^5 \tilde D_{2}^d \eta^{-5}$} & \multirow{2}{*}{$(x_2(t+\frac T2),x_3(t),x_2(t),x_1(t),x_4(t+\frac T2), x_1(t),$}& \\
&&&\\ 
&&\multirow{2}{*}{$x_2(t),x_3(t),x_2(t+\frac T2),x_1(t+\frac T2),x_4(t),x_1(t+\frac T2))$}&\\
&&&\\

\hline 
\end{tabular}
\caption{Summary of distinct forms of oscillating states bifurcating from the equilibrium $x=0$ of the system (\ref{eq:1}), where cells are coupled to their nearest and next nearest neighbors (Part I).} \label{t:d12_hb_1}
\end{table}

\begin{table}

\hspace*{-1.2cm}
\hskip-1cm
\begin{tabular}{|c|c|l|c|}
\hline 
\multirow{2}{*}{\bf Critical }& \multirow{4}{*}{\bf Symmetry }& \multirow{2}{*}{{\bf Form of Oscillating-States} } & \multirow{4}{*}{{\bf Figure} }  \\
&&&\\
\multirow{2}{*}{\bf Eigenvalue} &&\multirow{2}{*}{(for some period $T$)}&\\
& & & \\
\hline \hline
\multirow{36}{*}{$\xi_2$}
& \multirow{6}{*}{$\bz_{12}^{t_2}$} & \multirow{3}{*}{$(x_1(t),x_1(t+\frac{T}{6}),x_1(t+\frac{2T}{6}),\dots, x_1(t+\frac {5T}{6}),$}&  \multirow{6}{*}{\includegraphics[width=4.7cm]{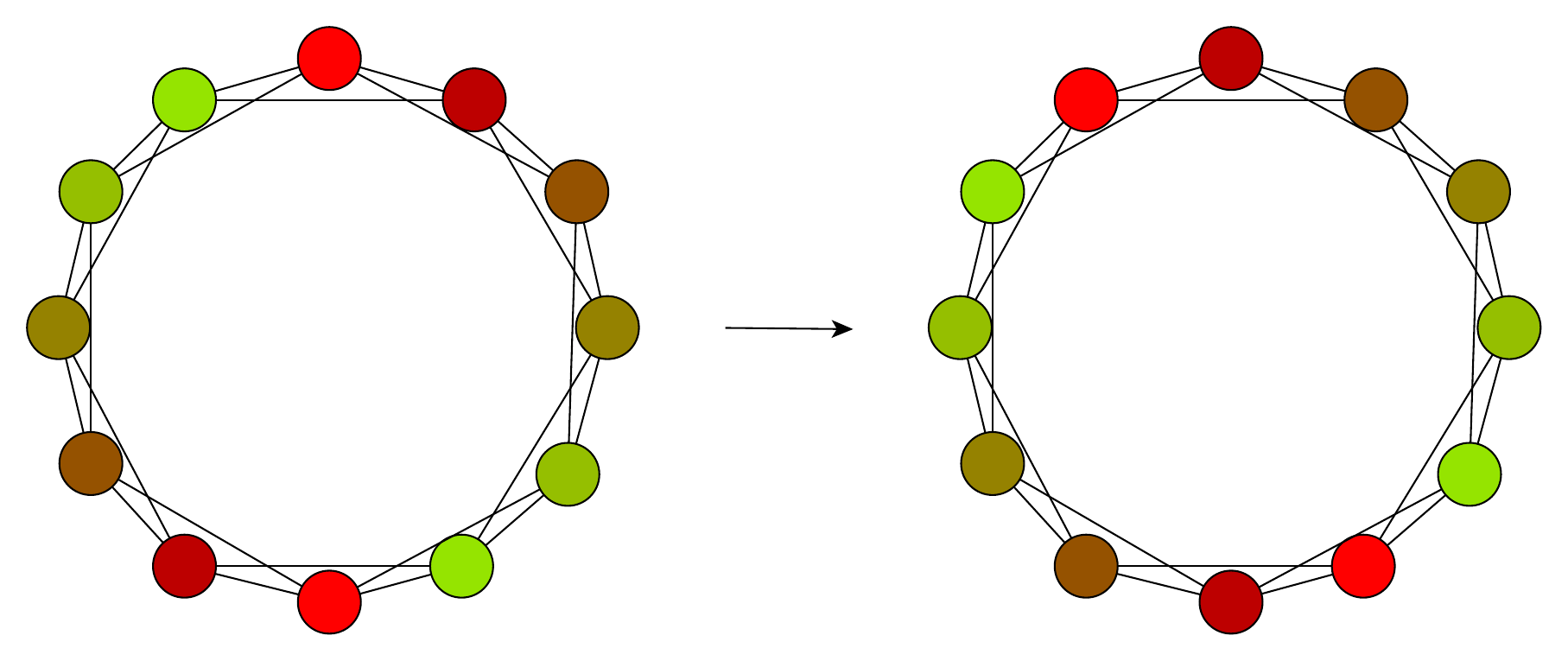}}  \\
&&&{\small $+\frac{T}{6}$}\\
&&&\\
&&\multirow{3}{*}{$x_1(t),x_1(t+\frac{T}{6}),x_1(t+\frac{2T}{6}),\dots, x_1(t+\frac{5T}{6}))$}&\\
&&&\\
&&&\\
\cline{2-4}
& \multirow{6}{*}{$\varsigma \bz_{12}^{t_2} \varsigma^{-1}$} & \multirow{3}{*}{$(x_1(t),x_1(t+\frac{5 T}{6}),x_1(t+\frac{4T}{6})\dots, x_1(t+\frac T6),$}&   \multirow{7}{*}{\includegraphics[width=4.7cm]{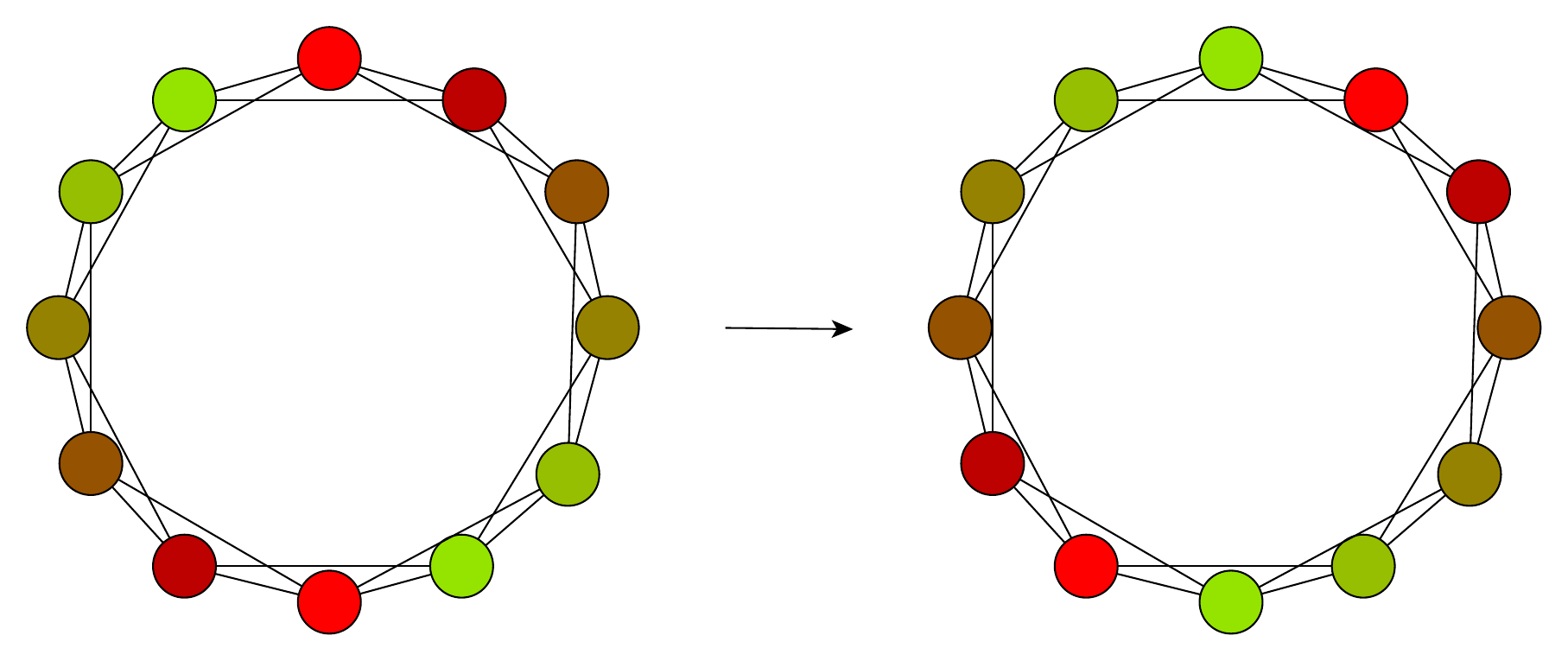}}  \\
&&&{\small $+\frac{T}{6}$}\\
&&&\\
&&\multirow{3}{*}{$x_1(t),x_1(t+\frac{5 T}{6}),x_1(t+\frac{4T}{6})\dots, x_1(t+\frac T6))$}&\\
&&&\\
&&&\\
\cline{2-4}
& \multirow{4}{*}{$D_{4}^d$} & \multirow{2}{*}{$(x_1(t),x_2(t),x_1(t+\frac{T}{2}),x_1(t+\frac{T}{2}), x_2(t),x_1(t), $}&  \multirow{12}{*}{\includegraphics[width=4.8cm]{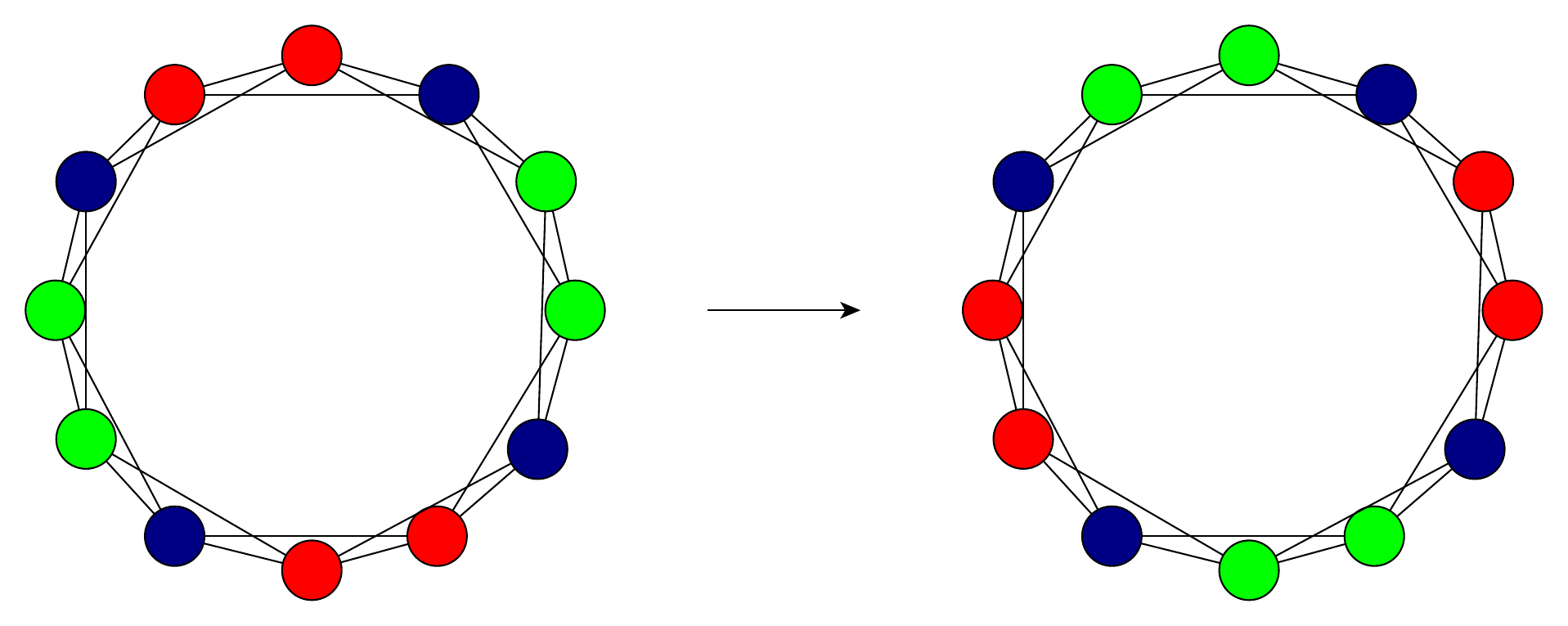}} \\ 
&&&\\
&&\multirow{2}{*}{$x_1(t),x_2(t), x_1(t+\frac{T}{2}), x_1(t+\frac{T}{2}), x_2(t), x_1(t)) $}&\\
&&&\\
\cline{2-3}
& \multirow{4}{*}{$\eta D_4^d \eta^{-1}$} & \multirow{2}{*}{$(x_1(t),x_1(t),x_2(t),x_1(t+\frac{T}{2}),x_1(t+\frac{T}{2}), x_2(t),$}&{\small $+\frac{T}{2}$} \\
&&&\\
&&\multirow{2}{*}{$x_1(t),x_1(t),x_2(t), x_1(t+\frac{T}{2}), x_1(t+\frac{T}{2}), x_2(t))$}&\\
&&&\\
\cline{2-3} 
& \multirow{4}{*}{$\eta^2 D_{4}^d \eta^{-2}$} & \multirow{2}{*}{$(x_2(t),x_1(t),x_1(t),x_2(t),x_1(t+\frac{T}{2}),x_1(t+\frac{T}{2}), $}& \\
&&&\\
&&\multirow{2}{*}{$x_2(t),x_1(t),x_1(t),x_2(t), x_1(t+\frac{T}{2}), x_1(t+\frac{T}{2}))$}& \\
&& &\\
\cline{2-4}
& \multirow{4}{*}{$D_{4}^{\hat d}$} & \multirow{2}{*}{$(x_1(t),x_2(t),x_1(t),x_1(t+\frac T2),x_2(t+\frac T2),x_1(t+\frac T2), $}&  \multirow{12}{*}{\includegraphics[width=4.8cm]{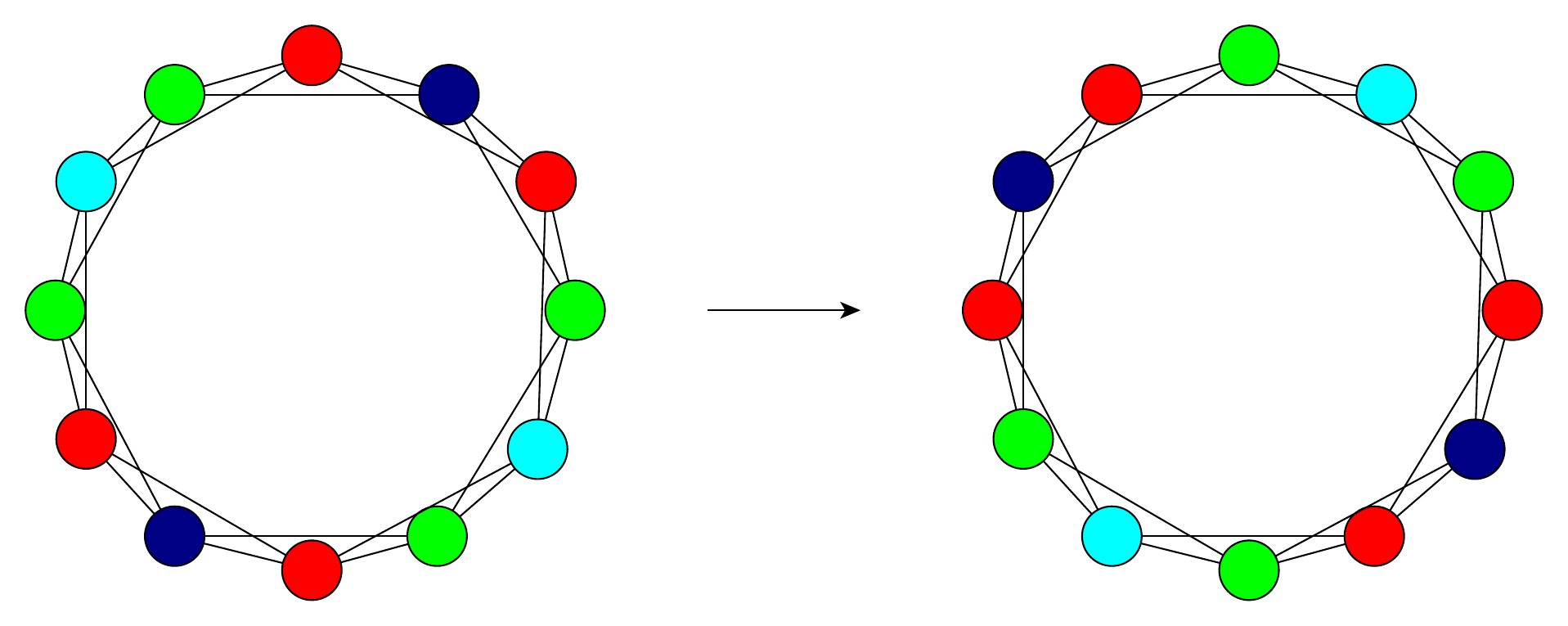}} \\
&&&\\
&&\multirow{2}{*}{$x_1(t),x_2(t),x_1(t),x_1(t+\frac T2),x_2(t+\frac T2),x_1(t+\frac T2)) $}&\\
&&&\\
\cline{2-3}

& \multirow{4}{*}{$\eta D_{4}^{\hat d}\eta^{-1}$} & \multirow{2}{*}{$(x_1(t+\frac T2), x_1(t),x_2(t),x_1(t),x_1(t+\frac T2),x_2(t+\frac T2), $}& {\small $+\frac{T}{2}$}\\
&&&\\
&&\multirow{2}{*}{$x_1(t+\frac T2), x_1(t),x_2(t),x_1(t),x_1(t+\frac T2),x_2(t+\frac T2)) $}&\\
&&&\\
\cline{2-3}

& \multirow{4}{*}{$\eta^2 D_{4}^{\hat d}\eta^{-2}$} & \multirow{2}{*}{$(x_2(t+\frac T2),x_1(t+\frac T2), x_1(t),x_2(t),x_1(t),x_1(t+\frac T2), $}& \\
&&&\\
&&\multirow{2}{*}{$x_2(t+\frac T2),x_1(t+\frac T2), x_1(t),x_2(t),x_1(t),x_1(t+\frac T2)) $}&\\
&&&\\
\hline

\multirow{24}{*}{$\xi_3$}
& \multirow{6}{*}{$\bz_{12}^{t_3}$} & \multirow{3}{*}{$(x_1(t),x_1(t+\frac T4),x_1(t+\frac T2),x_1(t+\frac {3T}{4}), $}&  \multirow{6}{*}{\includegraphics[width=4.6cm]{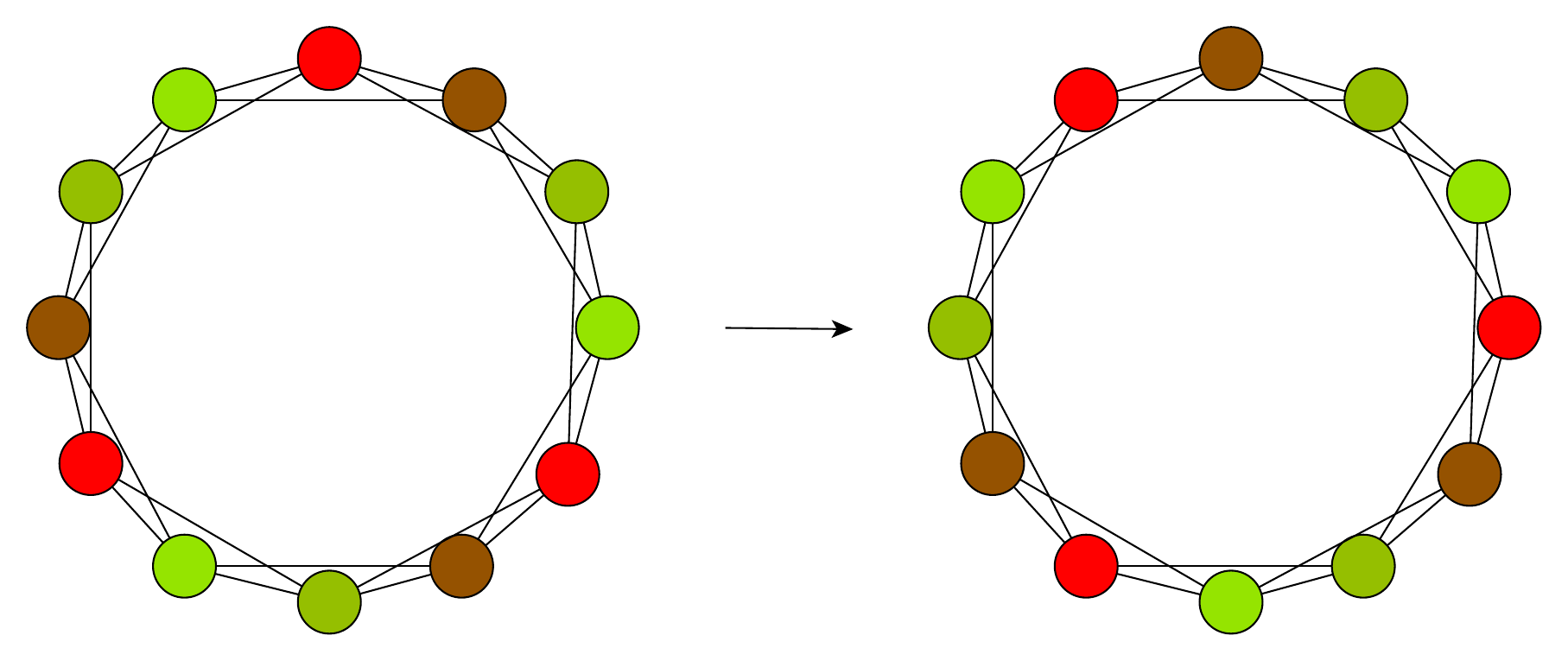}}  \\
&&&{\small $+\frac{T}{4}$}\\
&&&\\
&&\multirow{3}{*}{$x_1(t),x_1(t+\frac T4), \dots, x_1(t+\frac {3T}{4}))$}&\\
&&&\\
&&&\\
\cline{2-4}
& \multirow{6}{*}{$\varsigma \bz_{12}^{t_3} \varsigma^{-1}$} & \multirow{3}{*}{$(x_1(t),x_1(t+\frac {3T}{4}),x_1(t+\frac T2),x_1(t+\frac {T}{4}),$}&   \multirow{6}{*}{\includegraphics[width=4.6cm]{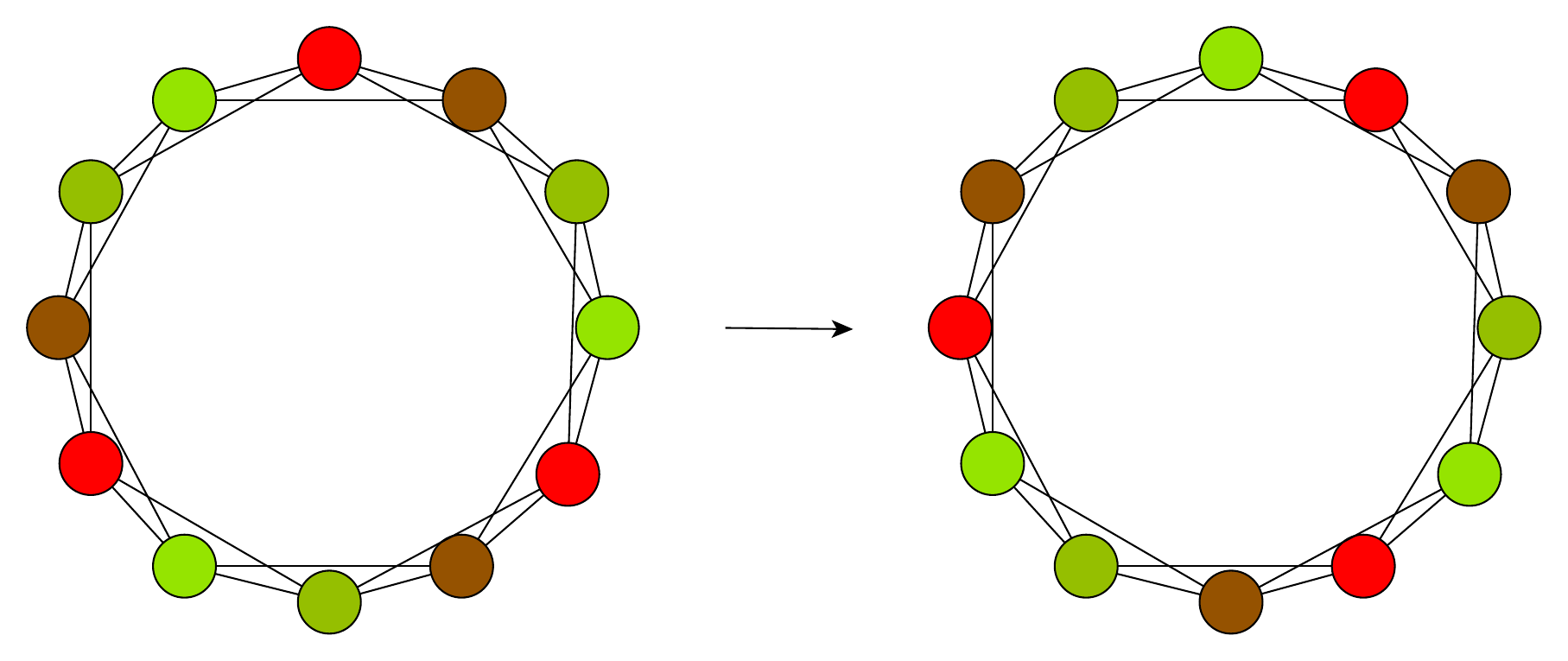}}  \\
&&&{\small $+\frac{T}{4}$}\\
&&&\\
&&\multirow{3}{*}{$x_1(t),x_1(t+\frac {3T}{4}), \dots, x_1(t+\frac {T}{4}))$}&\\
&&&\\
&&&\\
\cline{2-4}
& \multirow{4}{*}{$D_6^d$} & \multirow{2}{*}{$(x_1(t),x_1(t+\frac T2),x_1(t+\frac T2),x_1(t), x_1(t),x_1(t+\frac T2),$}& \multirow{8}{*}{\includegraphics[width=4.6cm]{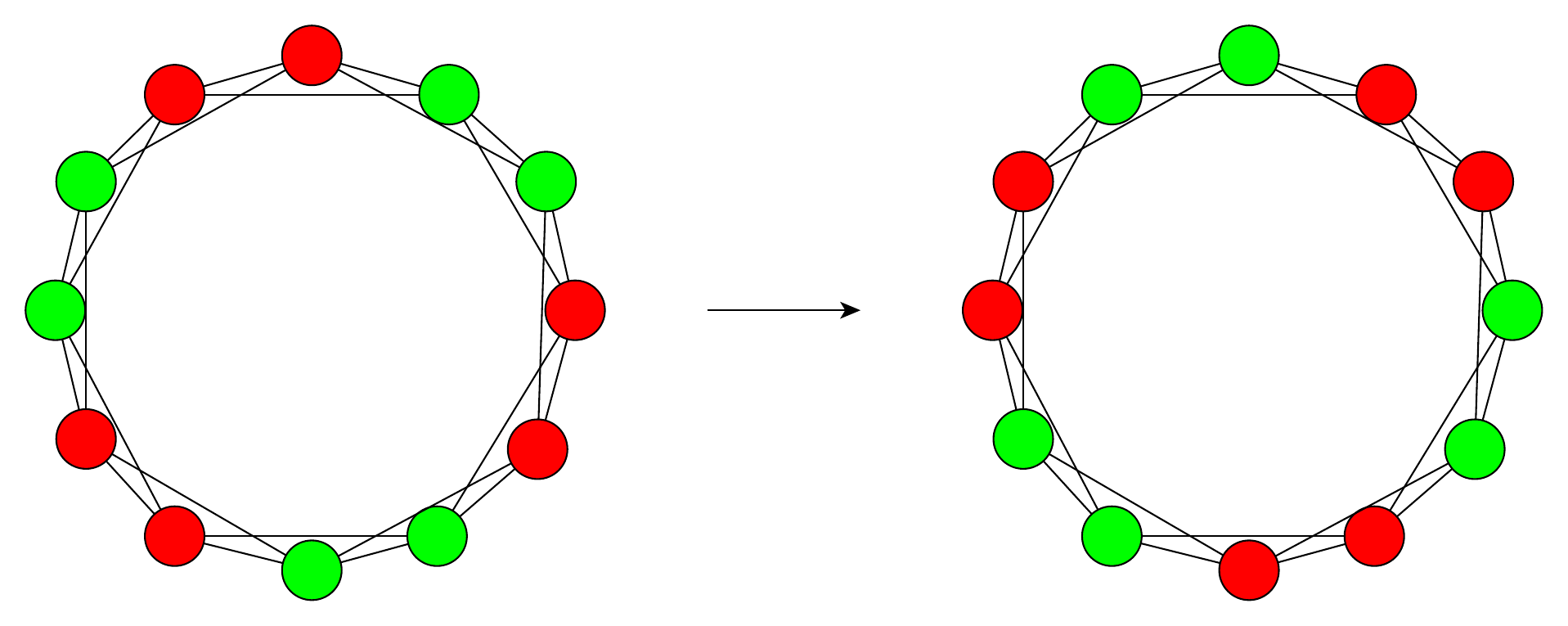}}\\
&&&\\
&&\multirow{2}{*}{$x_1(t+\frac T2),x_1(t),x_1(t),x_1(t+\frac T2),x_1(t+\frac T2),x_1(t)) $}&{\small $+\frac{T}{2}$}\\
&&&\\
\cline{2-3}
& \multirow{4}{*}{$\eta D_6^d \eta^{-1}$} & \multirow{2}{*}{$(x_1(t),x_1(t),x_1(t+\frac T2),x_1(t+\frac T2),x_1(t), x_1(t),$}& \\
&&&\\
&&\multirow{2}{*}{$x_1(t+\frac T2),x_1(t+\frac T2),x_1(t),x_1(t),x_1(t+\frac T2),x_1(t+\frac T2)) $}&\\
&&&\\
\cline{2-4}

& \multirow{4}{*}{$\tilde D_6^d$} & \multirow{2}{*}{$(x_1(t),x_2(t),x_1(t),x_2(t+\frac T2), x_1(t),x_2(t),$}& \multirow{8}{*}{\includegraphics[width=4.6cm]{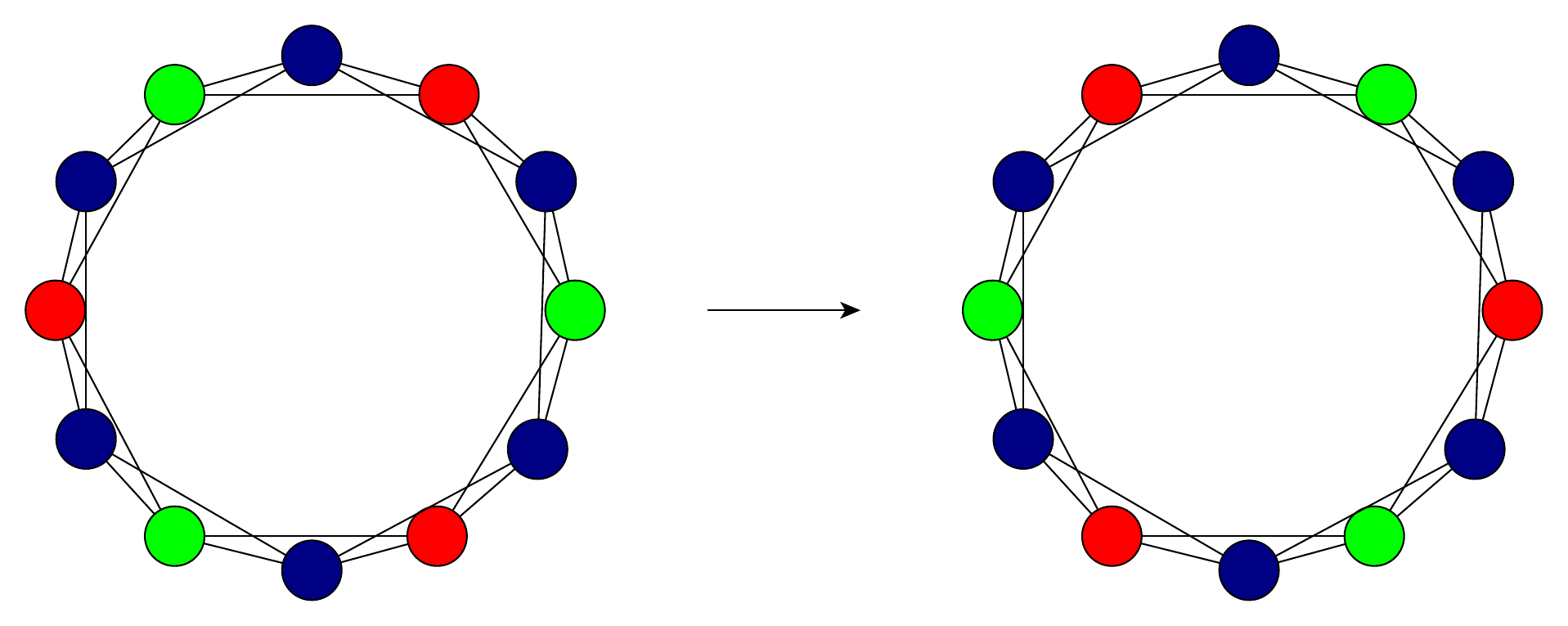}}\\
&&&\\
&&\multirow{2}{*}{$x_1(t),x_2(t+\frac T2),x_1(t),x_2(t),x_1(t),x_2(t+\frac T2)) $}&{\small $+\frac{T}{2}$}\\
&&&\\
\cline{2-3}
& \multirow{4}{*}{$\eta \tilde D_6^d \eta^{-1}$} & \multirow{2}{*}{$(x_2(t+\frac T2),x_1(t),x_2(t),x_1(t),x_2(t+\frac T2), $}& \\
&&&\\

&&\multirow{2}{*}{$x_1(t),x_2(t),x_1(t),x_2(t+\frac T2),x_1(t),x_2(t),x_1(t))  $}&\\
&&&\\
\hline 
\end{tabular}
\caption{Summary of distinct forms of oscillating states bifurcating from the equilibrium $x=0$ of the system (\ref{eq:1}), where cells are coupled to their nearest and next nearest neighbors (Part II).} \label{t:d12_hb_2}
\end{table}

\begin{table}
\hspace*{-1.2cm}
\hskip-.5cm
\begin{tabular}{|c|c|l|c|}
\hline 
\multirow{2}{*}{\bf Critical }& \multirow{4}{*}{\bf Symmetry }& \multirow{2}{*}{{\bf Form of Oscillating-States} } & \multirow{4}{*}{{\bf Figure} }  \\
&&&\\
\multirow{2}{*}{\bf Eigenvalue} &&\multirow{2}{*}{(for some period $T$)}&\\
& & & \\
\hline \hline
\multirow{24}{*}{$\xi_4$}
& \multirow{6}{*}{$\bz_{12}^{t_4}$} & \multirow{3}{*}{$(x_1(t),x_1(t+\frac T3), x_1(t+\frac {2T}{3}),$}&  \multirow{6}{*}{\includegraphics[width=4.7cm]{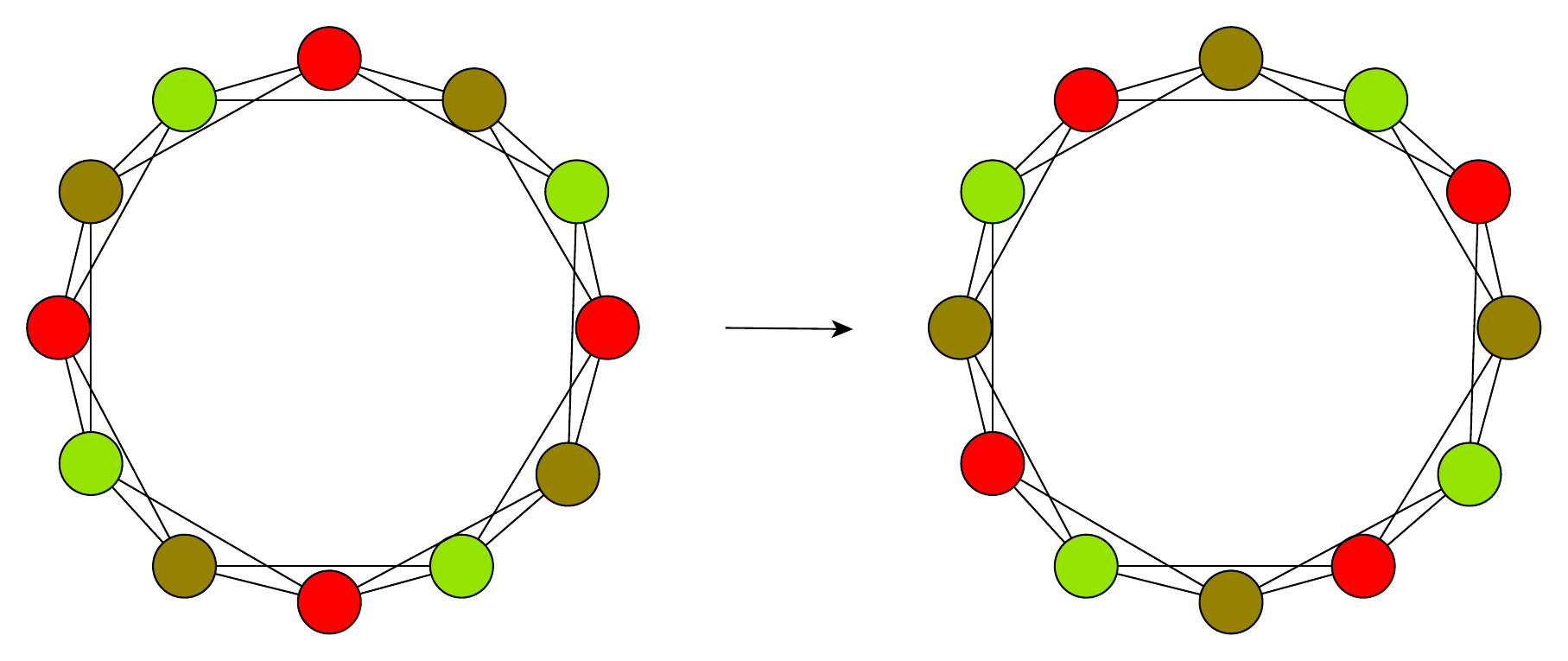}}  \\
&&&{\small $+\frac{T}{3}$}\\ 
&&&\\
&&\multirow{3}{*}{$ x_1(t), x_1(t+\frac T3),\dots, x_1(t+\frac {2T}{3}))$}&\\
&&&\\
&&&\\
\cline{2-4}
& \multirow{6}{*}{$\varsigma \bz_{12}^{t_4} \varsigma^{-1}$} & \multirow{3}{*}{$(x_1(t),x_1(t+\frac {2T}{3}), x_1(t+\frac {T}{3}),$}&   \multirow{7}{*}{\includegraphics[width=4.7cm]{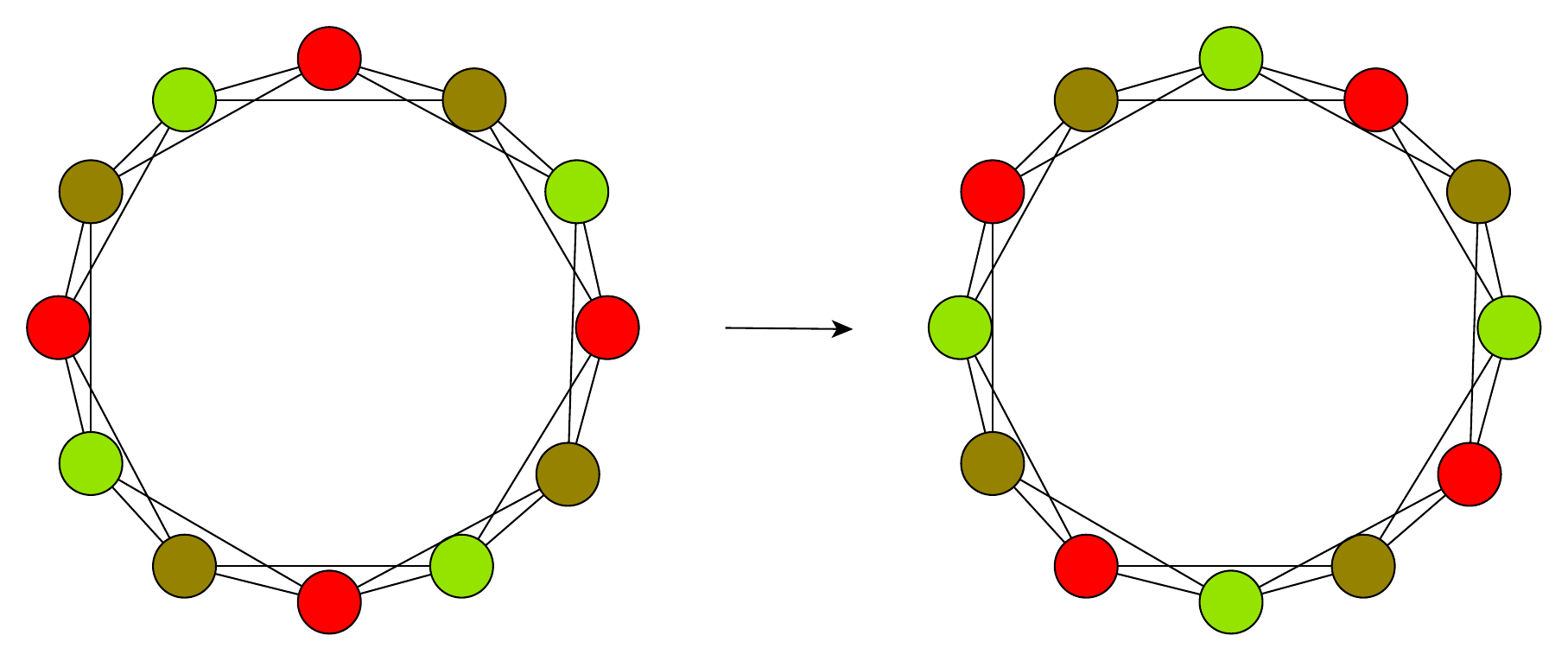}}  \\
&&&{\small $+\frac{T}{3}$}\\
&&&\\
&&\multirow{3}{*}{$x_1(t), x_1(t+\frac {2T}{3}),\dots, x_1(t+\frac {T}{3}))$}&\\
&&&\\
&&&\\
\cline{2-4}
& \multirow{4}{*}{$D_{4}^z$} & \multirow{2}{*}{$(x_1(t),x_2(t),x_1(t+\frac T2),x_1(t),x_2(t),x_1(t+\frac T2),$}&  \multirow{12}{*}{\includegraphics[width=4.8cm]{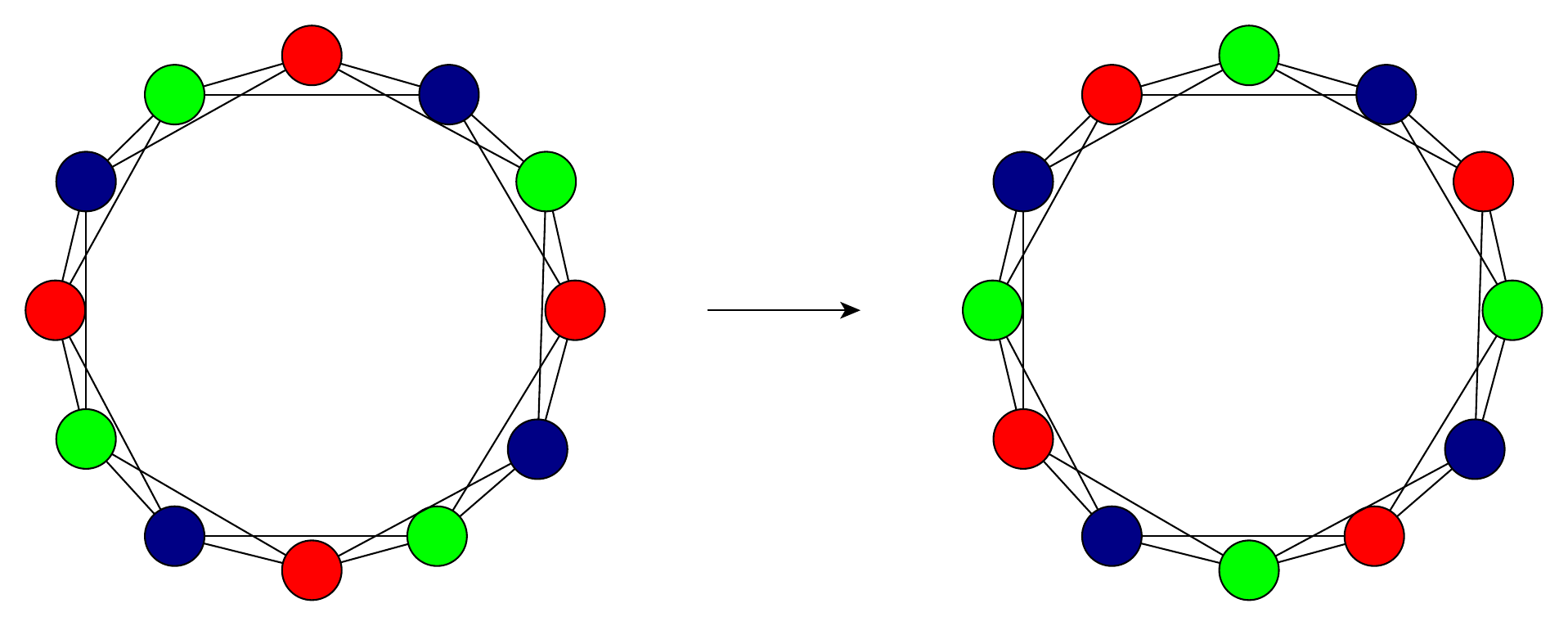}} \\
&&&\\
&&\multirow{2}{*}{$x_1(t),x_2(t),x_1(t+\frac T2),x_1(t),x_2(t),x_1(t+\frac T2)) $}&\\
&&&\\
\cline{2-3}
& \multirow{4}{*}{$\eta D_4^z \eta^{-1}$} & \multirow{2}{*}{$(x_1(t+\frac T2), x_1(t),x_2(t),x_1(t+\frac T2), x_1(t),x_2(t),$}&{\small $+\frac{T}{2}$} \\
&&&\\
&&\multirow{2}{*}{$x_1(t+\frac T2), x_1(t),x_2(t),x_1(t+\frac T2), x_1(t),x_2(t))$}&\\
&&&\\
\cline{2-3} 
& \multirow{4}{*}{$\eta^2 D_{4}^z \eta^{-2}$} & \multirow{2}{*}{$(x_2(t),x_1(t+\frac T2), x_1(t),x_2(t),x_1(t+\frac T2), x_1(t),$}& \\
&&&\\
&&\multirow{2}{*}{$x_2(t),x_1(t+\frac T2), x_1(t),x_2(t),x_1(t+\frac T2), x_1(t)) $}& \\
&& &\\
\cline{2-4}
& \multirow{4}{*}{$D_{4}$} & \multirow{2}{*}{$(x_1(t),x_2(t),x_1(t),x_1(t),x_2(t),x_1(t),$}&  \multirow{12}{*}{\includegraphics[width=2.3cm]{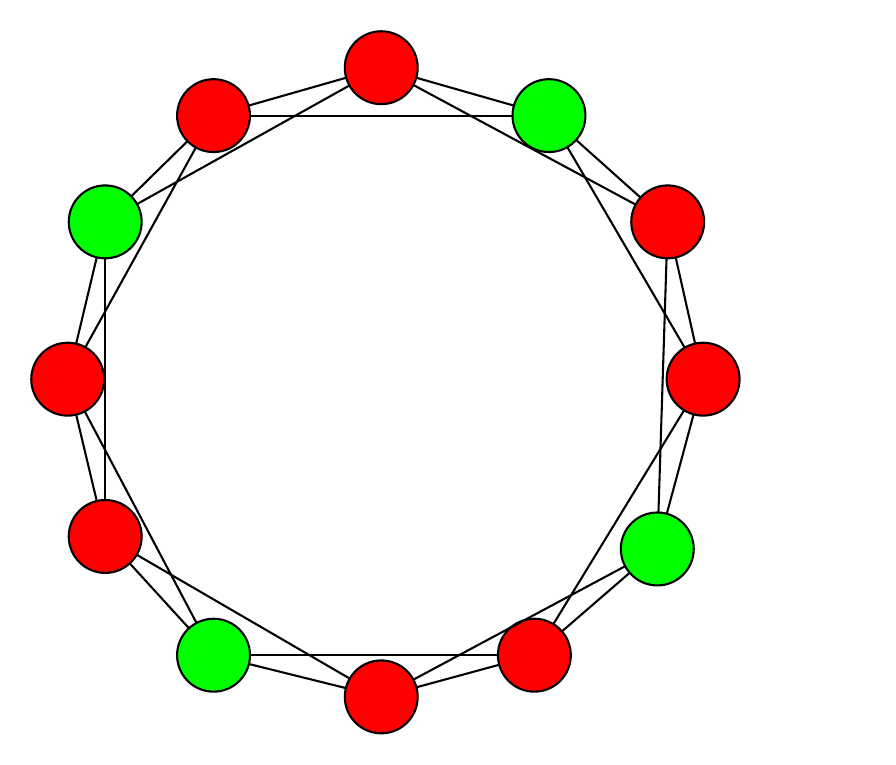}} \\
&&&\\
&&\multirow{2}{*}{$x_1(t),x_2(t),x_1(t),x_1(t),x_2(t),x_1(t)) $}&\\
&&&\\
\cline{2-3}

& \multirow{4}{*}{$\eta D_{4}\eta^{-1}$} & \multirow{2}{*}{$(x_1(t),x_1(t),x_2(t),x_1(t),x_1(t),x_2(t),$}& \\
&&&\\
&&\multirow{2}{*}{$x_1(t),x_1(t),x_2(t),x_1(t),x_1(t),x_2(t)) $}&\\
&&&\\
\cline{2-3}

& \multirow{4}{*}{$\eta^2 D_{4}\eta^{-2}$} & \multirow{2}{*}{$(x_2(t),x_1(t),x_1(t),x_2(t),x_1(t),x_1(t),$}& \\
&&&\\
&&\multirow{2}{*}{$x_2(t),x_1(t),x_1(t),x_2(t),x_1(t),x_1(t))$}&\\
&&&\\
\hline
\multirow{12}{*}{$\xi_5$}
& \multirow{6}{*}{$\bz_{12}^{t_5}$} & \multirow{2}{*}{$(x_1(t),x_1(t+\frac {5T}{12}),x_1(t+\frac {10T}{12}), x_1(t+\frac {3T}{12}),$}&  \multirow{7}{*}{\includegraphics[width=4.8cm]{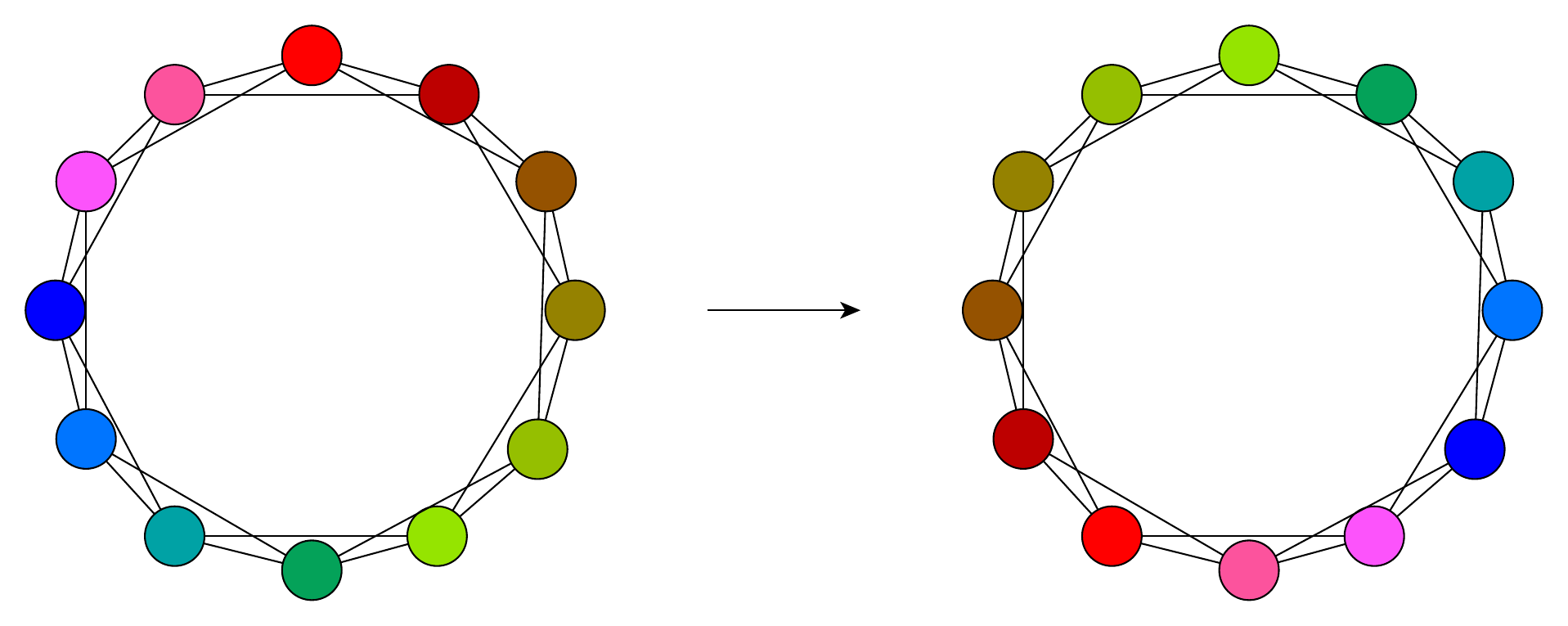}}  \\
&&&{\small $+\frac{T}{12}$}\\
&&\multirow{2}{*}{$x_1(t+\frac {8T}{12}), x_1(t+\frac {T}{12}), x_1(t+\frac {6T}{12}), x_1(t+\frac {11T}{12}),$}&\\
&&&\\
&&\multirow{2}{*}{$ x_1(t+\frac {4T}{12}), x_1(t+\frac {9T}{12}), x_1(t+\frac {2T}{12}), x_1(t+\frac {7T}{12})$}&\\
&&&\\
\cline{2-4}
& \multirow{6}{*}{$\varsigma \bz_{12}^{t_5} \varsigma^{-1}$} & \multirow{2}{*}{$(x_1(t),x_1(t+\frac {7T}{12}),x_1(t+\frac {2T}{12}), x_1(t+\frac {9T}{12}),$}&   \multirow{7}{*}{\includegraphics[width=4.8cm]{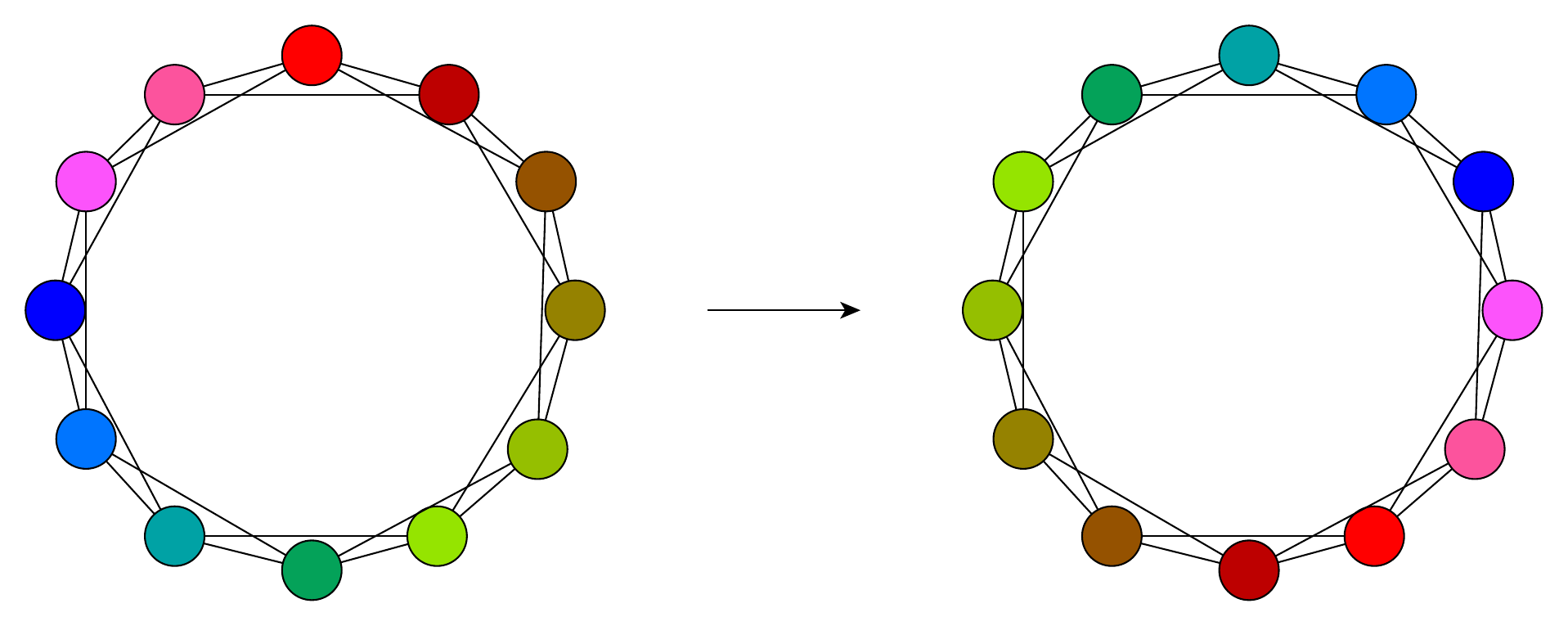}}  \\
&&&{\small $+\frac{T}{12}$}\\
&&\multirow{2}{*}{$x_1(t+\frac {4T}{12}), x_1(t+\frac {11T}{12}), x_1(t+\frac {6T}{12}), x_1(t+\frac {T}{12}),$}&\\
&&&\\
&&\multirow{2}{*}{$ x_1(t+\frac {8T}{12}), x_1(t+\frac {3T}{12}), x_1(t+\frac {10T}{12}), x_1(t+\frac {5T}{12})$}&\\
&&&\\
\hline

\multirow{6}{*}{$\xi_6$}
& \multirow{6}{*}{$D_{12}^{\hat d}$} & \multirow{6}{*}{$(x_1(t),x_1(t+\frac T2),x_1(t),x_1(t+\frac T2),\dots,x_1(t+\frac T2))$}&  \multirow{6}{*}{ \includegraphics[width=4.8cm]{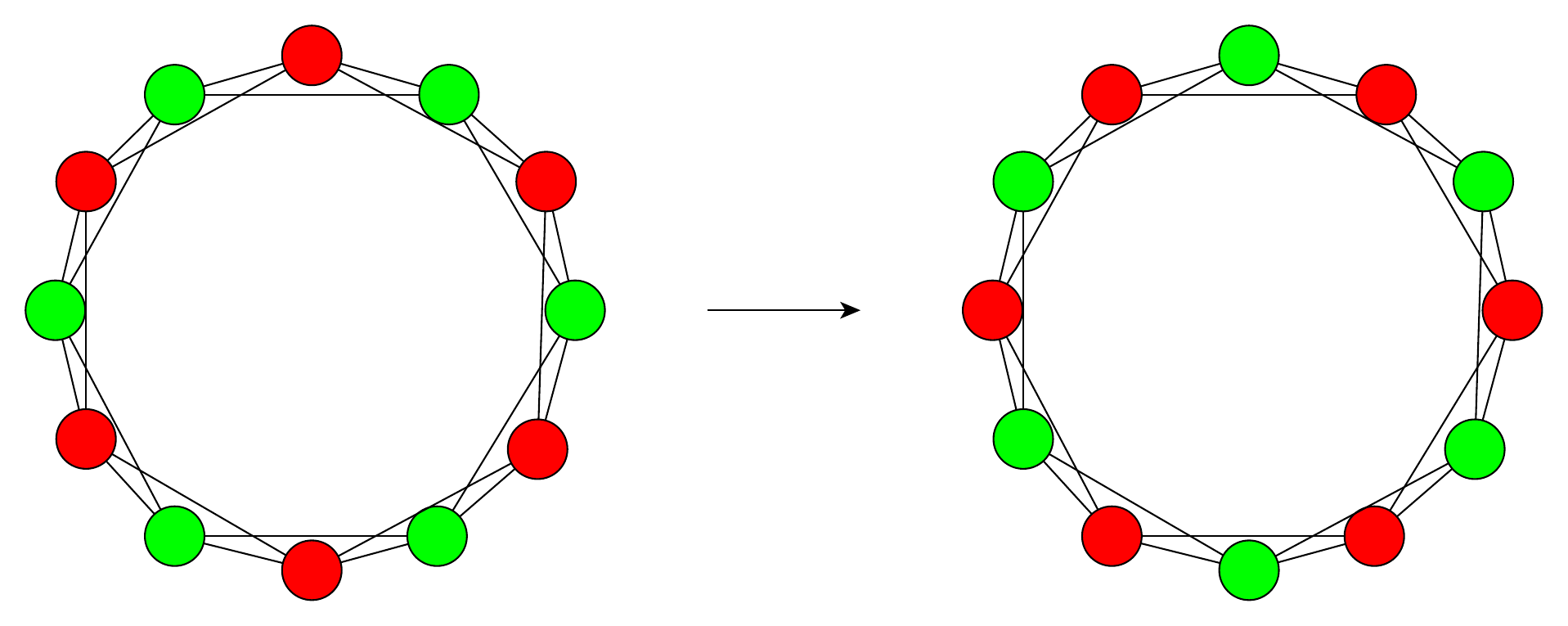}}\\
&&&{\small $+\frac{T}{2}$}\\
&&&\\
&&&\\
&&&\\
&&&\\

\hline 
\end{tabular}
\caption{Summary of distinct forms of oscillating states bifurcating from the equilibrium $x=0$ of the system (\ref{eq:1}), where cells are coupled to their nearest and next nearest neighbors (Part III).} \label{t:d12_hb_3}
\end{table}

\END
\end{ex}
\vs
\begin{rmk}\rm
Depending on the isotropy, different cells may possess different period. As an example, for $\tilde D_2^d$ (and its conjugacy copies) in Table~\ref{t:d12_hb_1}, consider the two diagrams in the bottom row. The cell at $2$ o'clock and the cell at $8$ o'clock do not change color (black) as the time $\frac T2$ elapses. This means they have half the period as other cells. More precisely, for isotropy $\eta^k \tilde D_2^d \eta^{-k}$ with $k=0,1,2,3,4,5$, the cells $3+k$ and $9+k$ have half the period as other cells. 
Similarly, in Table \ref{t:d12_hb_2}, for isotropy $\eta^k D_4^d \eta^{-k}$ with $k=0,1,2$, the cells $2+k$, $5+k$, $8+k$, $11+k$ (black in diagram) have  period $\frac T2$; for isotropy $\eta^k\tilde D_6^d\eta^{-k}$ with $k=0,1$, the  cells $1+k,3+k,5+k,7+k,9+k,11+k$ (black in diagram) have period $\frac T2$. In Table~\ref{t:d12_hb_3}, for isotropy $\eta^k D_4^z \eta^{-k}$ with $k=0,1,2$, the cells $2+k,5+k,8+k,11+k$ (black in diagram) have period $\frac{T}{2}$.
\END 
\end{rmk} 

\section{Near-Neighbor Coupling and Simulation Examples}
\label{sec:simul} 

In this section we consider the 12-cell ring with all possibilities of first- and second-closest-neighbor couplings. More precisely, each cell $i$ coupled to its two nearest neighbors on both sides with coupling strength $c_{i,i\pm1}=d_1$ and to its second-nearest neighbors on both sides with coupling strength $c_{i,i\pm2}=d_2$, as well as possibly having a self-feedback loop with strength $c_{ii}=c_0$. Each coupling strength is allowed to be positive, negative, or zero.
We illustrate some of the ensuing dynamics in the context of a concrete model mentioned in the Introduction, namely the nonlinear neural network model \eqref{neural}. We will also make reference to the diffusively-coupled system \eqref{eq-b} by choosing $c_0$ appropriately.

Since the coupling matrix $C$ is circulant,  we calculate its eigenvalues using \eqref{circulant-eigenvalues} as 
\begin{equation} \label{coupling-eigenvalues}
	\xi_j = c_0 + 2 d_1 \cos(\tfrac{\pi}{6} j) + 2d_2 \cos(\tfrac{\pi}{3} j), 
	\q j=0,1,2,\dots, n-1.
\end{equation}
We then determine the largest and smallest eigenvalues of $C$ in terms of the coupling strengths $d_1$ and $d_2$. Figures~\ref{F:max_eig_G12}
and \ref{F:min_eig_G12} graphically show which of the $\xi_i$ are the largest and smallest eigenvalues of $C$ for the complete range of coupling strengths $d_1$ and $d_2$. 
The figures are symmetric images of each other with respect to the origin since replacing $(d_1, d_2)$ by $(-d_1, -d_2)$ is equivalent to  multiplying $C$ by $-1$, which reverses the roles of smallest and largest eigenvalues. Note from \eqref{coupling-eigenvalues} that the self-coupling coefficient $c_0$ simply shifts the eigenvalues without altering their magnitude order. Hence, by changing $c_0$ one can make either the smallest or the largest eigenvalue the dominant one that determines the first bifurcation, in the context of the stability diagram of Figure~\ref{fig:bif}.
As depicted in Figures \ref{F:max_eig_G12}
and \ref{F:min_eig_G12}, every eigenvalue $\xi_i$ of $C$ can arise as the dominant one by appropriate choices of the coupling strengths $d_1$ and $d_2$. (Note that for $7\le j\le11$, $\xi_j$ is identical to $\xi_{12-j}$, by \eqref{coupling-eigenvalues}.)
Therefore, the present setting of coupling with two closest neighbor pairs permits a systematic investigation of the whole range of dynamics listed in Tables~\ref{t:d12_sb}--\ref{t:d12_hb_3}.

\begin{figure}[tbh] 
 \centerline{ 
 \includegraphics[width=.85\textwidth]{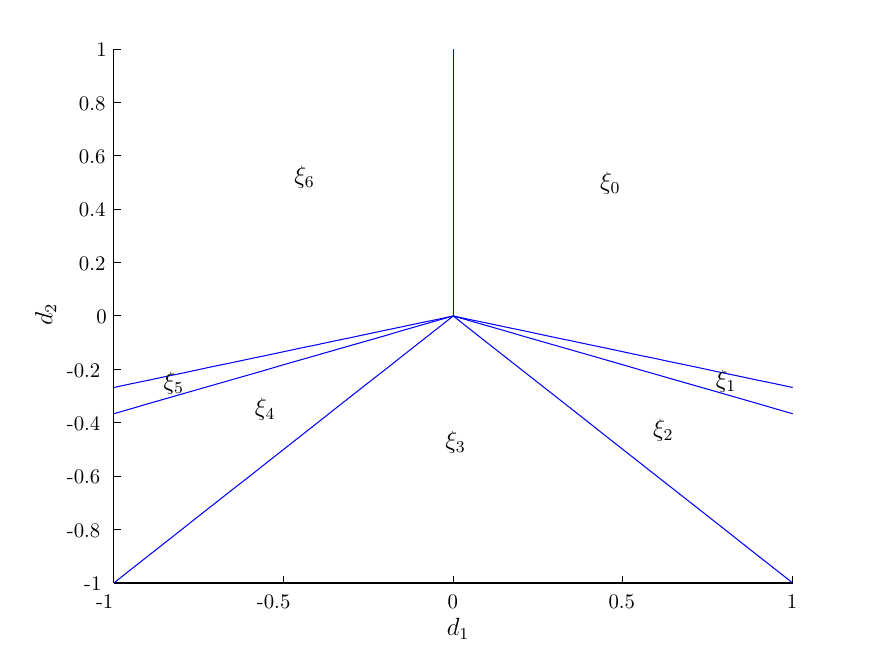}
}    
\caption{The largest eigenvalue of the coupling matrix $C$ for a circular arrangement of 12 cells where each cell is connected to four others, with coupling strength $d_1$ to its two immediate neighbors on each side and with coupling strength $d_2$ to the second-nearest neighbors.}
 \label{F:max_eig_G12}
 \end{figure}
 
\begin{figure}[tbh]  
\hskip1cm
 \includegraphics[width=0.9\textwidth]{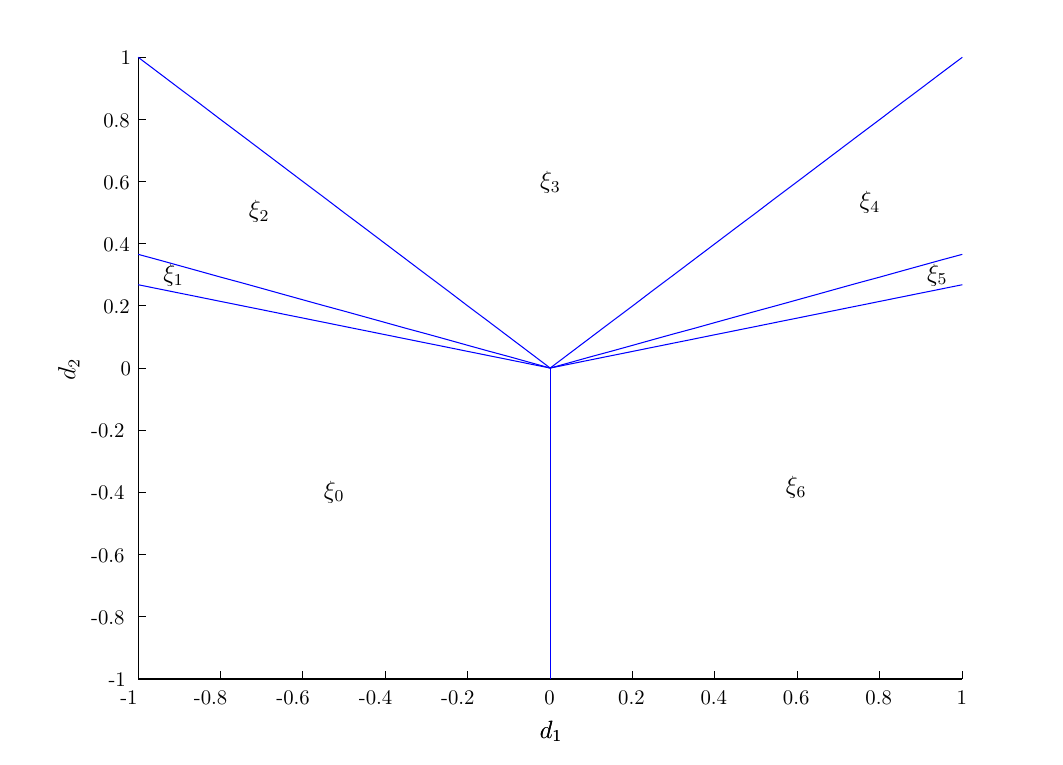}	
 \caption{The smallest eigenvalue of the coupling matrix $C$ for a circular ring of size 12, in terms of the coupling strengths. (See the caption of Figure~\ref{F:max_eig_G12} for explanation.)}
 \label{F:min_eig_G12}
\end{figure}
  
We consider the model \eqref{neural} with  $g(x)=\arctan x$ and $a_{ii}=0$ $\forall i$. Thus $f(x)=-x$, so that $f'(0)=-1$ and $g'(0)=1$.
In the coupling matrix $C$ we have $c_0=0$, and we fix the remaining coupling strengths as $d_1=0.25$, $d_2=0.5$. The eigenvalues $\xi_0$ to $\xi_6$ are
$\{1.5, 0.933013, -0.25, -1, -0.75, 0.066987, 0.5\}$. We take $\tau=1$ initially.

We first consider excitatory coupling by setting $\kappa=1$. 
The dominant eigenvalue is $\xi_0 = 1.5$; so the system settles to a nonzero steady state solution starting from random initial conditions, as shown in Figure~\ref{F:excitatory}. This is also the behavior of the undelayed system, which persists under the presence of delays. 
If we include a self-coupling term $c_0 = -1.5$, as in the diffusively coupled system \eqref{eq-b}, all eigenvalues of $C$ are shifted by $-1.5$. Now a negative eigenvalue, namely $\xi_3 = -2.5$, becomes the dominant one responsible for bifurcation. Consequently, the network splits into an oscillatory pattern (Figure~\ref{F:excitatory}). It is worth noting that, although diffusive coupling is expected to drive the system to a spatially uniform solution, in this example it breaks a uniform equilibrium and replaces it with a nontrivial spatial pattern of two clusters of synchronized oscillators.

\begin{figure}[tb]  
\centering	
 \begin{tabular}{cc}
  \includegraphics[scale=0.35]{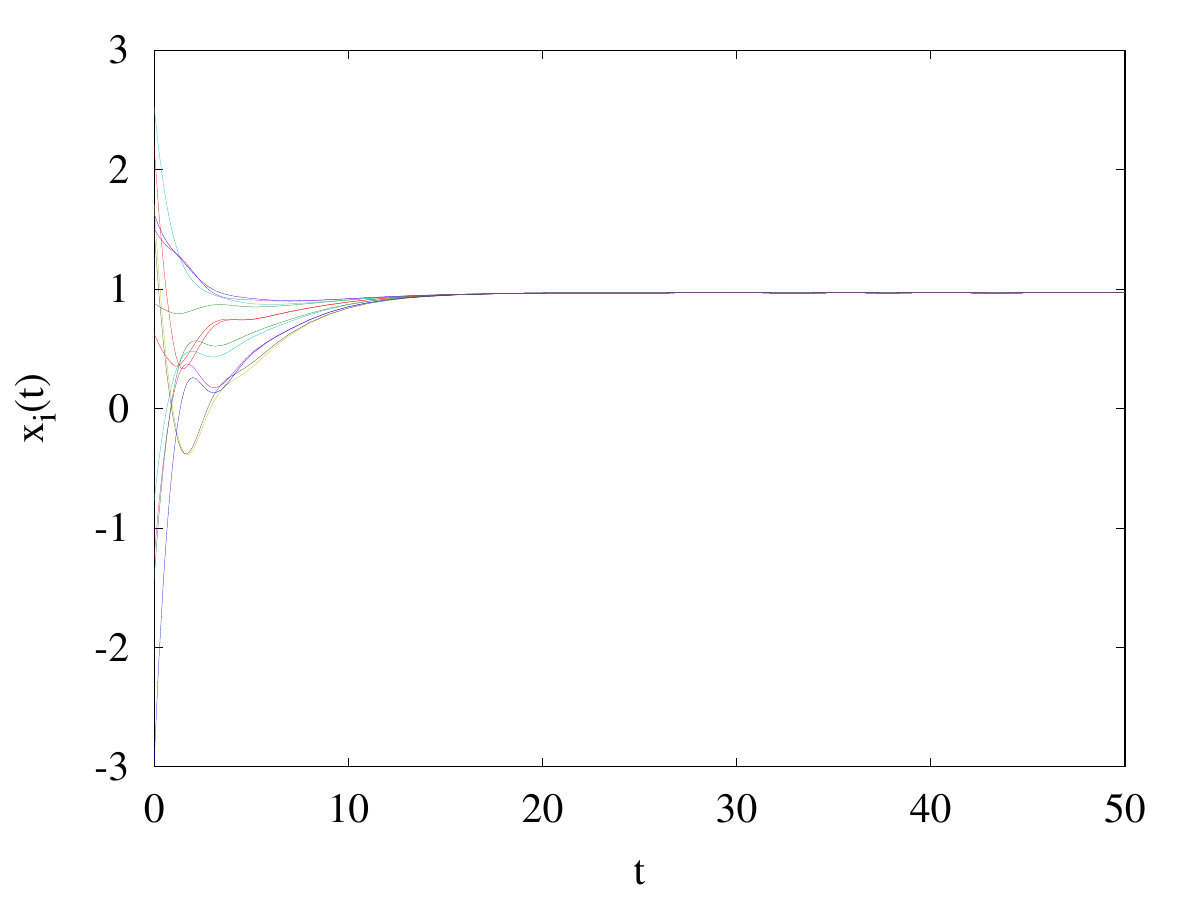} 
 & 
  \includegraphics[scale=0.35]{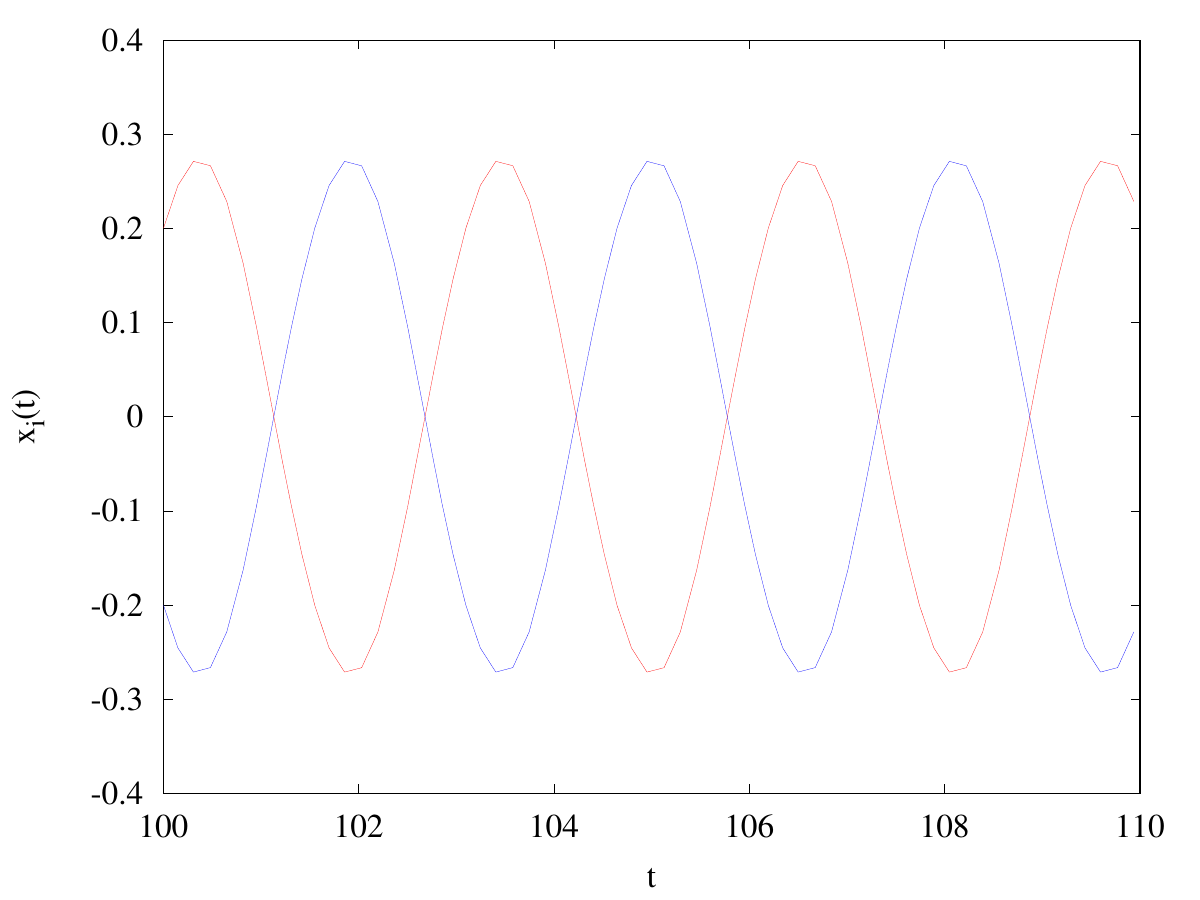}	\\
\end{tabular}
 \caption{System \eqref{neural} with excitatory coupling ($\kappa=1$) approaching a uniform steady-state solution from random initial conditions (left). The negative of the final state is also a stable equilibrium which can be observed for a different choice of initial conditions. Adding self-coupling yields the diffusively coupled system \eqref{eq-b}, which exhibits a stable oscillatory pattern of two clusters (right): Cells $\{2, 3, 6, 7, 10, 11\}$ form a synchronized cluster (blue curve) that oscillate in anti-phase with cells $\{1, 4, 5, 8, 9, 12 \}$ (red curve).}
 \label{F:excitatory}
\end{figure}

We now change the coupling from excitatory to inhibitory by setting $\kappa=-1.2$. We keep $d_1=0.25$, $d_2=0.5$, and $c_0 =0$ as before. The extreme eigenvalues of $\kappa C$ are $-1.2\xi_0 = -1.8$ and $-1.2\xi_3 = 1.2$. For the present value of $\tau=1$, the positive eigenvalue leaves the stability region first, so the systems settles into a two-cluster steady-state solution in accordance with $\xi_3$. When we take $\tau=3$, however, the negative eigenvalue becomes responsible for the bifurcation and the system exhibits spatially uniform periodic oscillations, in accordance with $\xi_0$. Here, in the absence of diffusive coupling, the delay apparently plays an important role in driving the system to a spatially uniform state, albeit an oscillatory one.

Although a rigorous stability analysis of the emerging spatio-temporal patterns is beyond the scope of the present study, repeated numerical simulations starting from random initial conditions manifest their stability. However, in many cases several stable patterns co-exist, so stability should only be inferred in a local sense.

\begin{figure}[tb]
\centering	
 \begin{tabular}{cc}
  \includegraphics[scale=0.35]{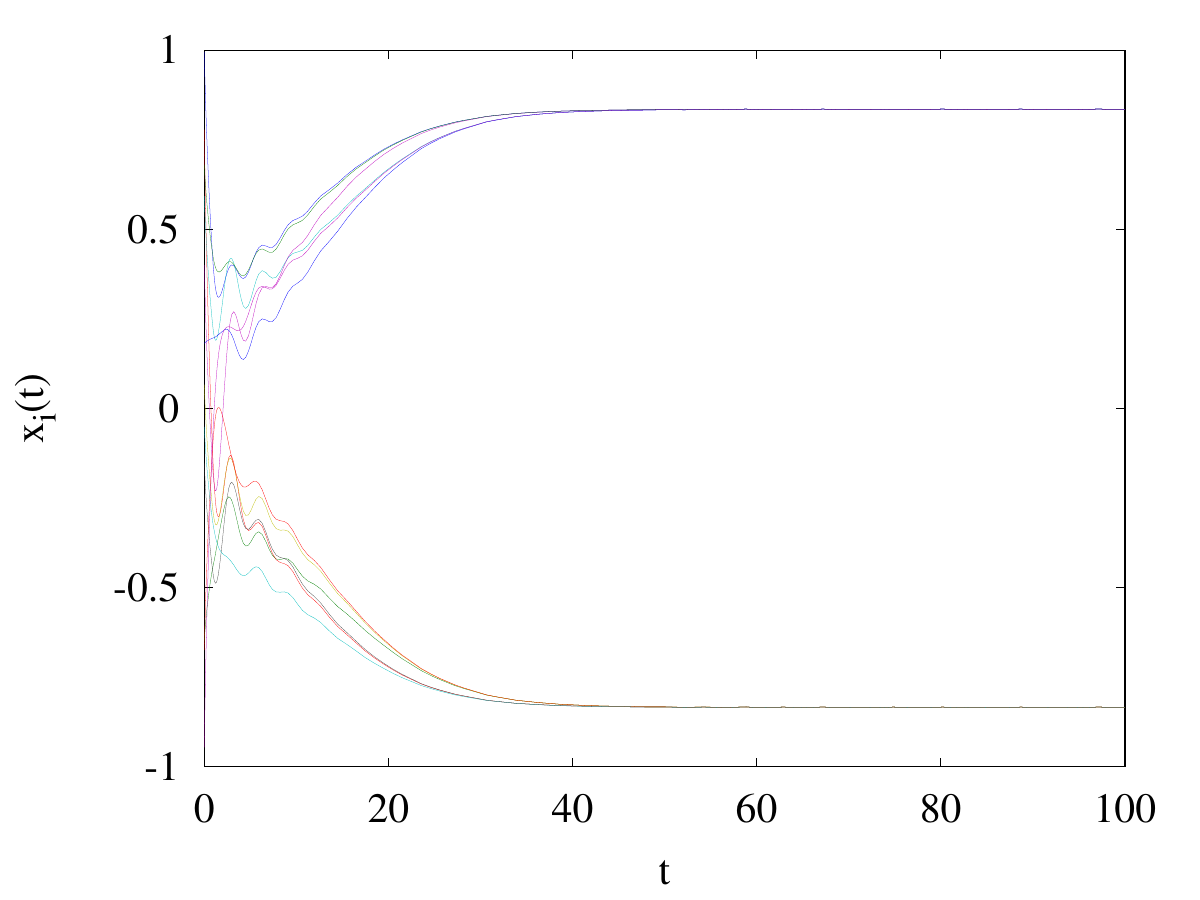} 
 &
  \includegraphics[scale=0.35]{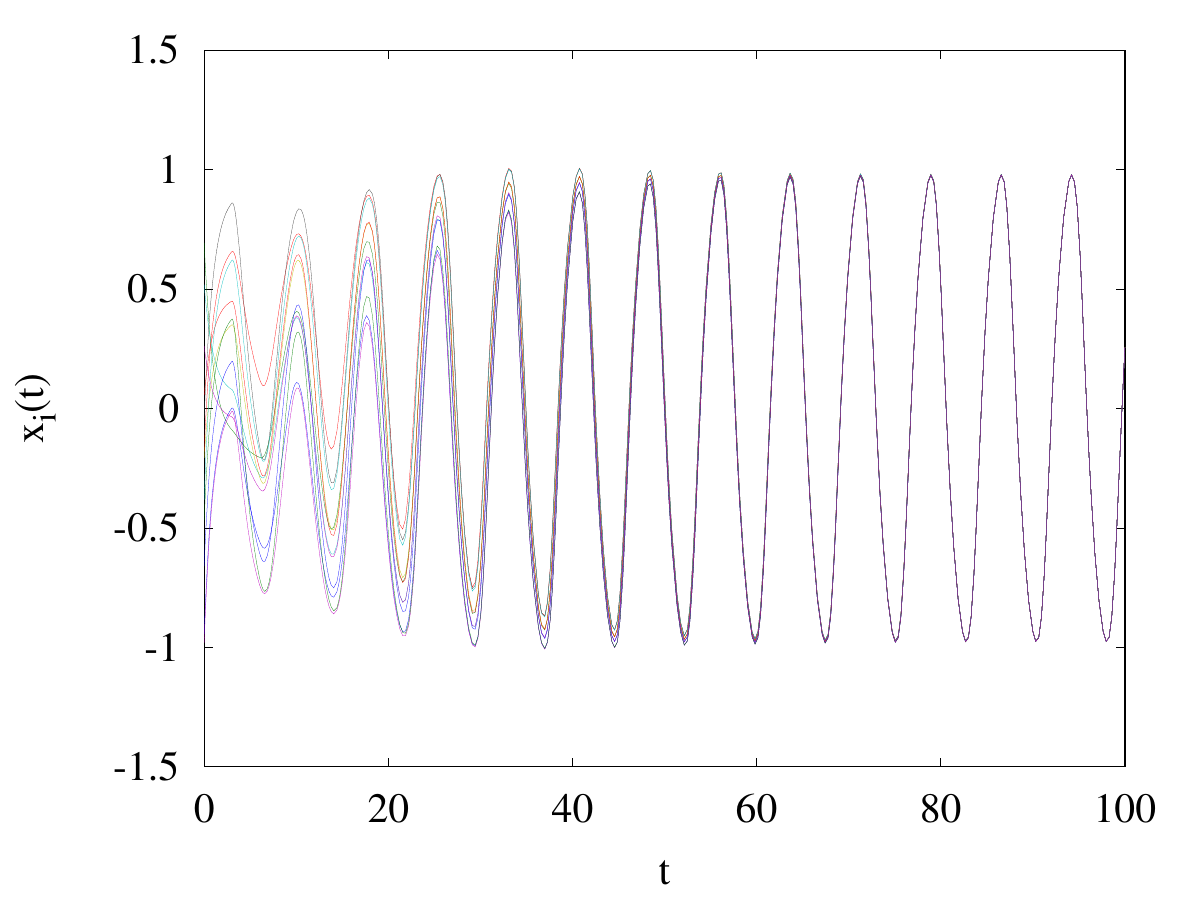}	\\
\end{tabular}
 \caption{System \eqref{neural} with inhibitory coupling ($\kappa=-1.2$). When $\tau=1$, the network approaches a two-cluster steady-state solution from random initial conditions (left). The clusters are the same as in the oscillatory pattern of Figure~\ref{F:excitatory}. Increasing the delay to $\tau=3$ yields spatially uniform synchronized oscillations shown on the right.}
 \label{F:inhibitory}
\end{figure}

\appendix
\section{Proof of Theorem  \ref{thm:steady}} \label{sec:proof}

\noi{\bf Theorem \ref{thm:steady}}
{\it  
Let $(\a_o,\beta_o)$ be such that $\a_o=-\beta_o$, and let $\xi_o\in \sig(C)$ be  given by (\ref{eq:xio}). If $\ome_0$ given by (\ref{eq:ome0}) is of form 
\[\ome_0=c_1(K_1)+c_2(K_2)+\cdots +c_p(K_p)
\]
for some $c_i\ne 0$, then there exists a bifurcating branch of steady states of symmetry at least $(K_i)$.
}

\begin{proof} The parameter pair $(\a, \beta)$ escapes the shaded region in Figure \ref{fig:bif} by crossing over $\LL$ through $(\a_o,\beta_o)$ (cf. Figure \ref{F:crossing}). 
\begin{figure}[!htb]
\centering
\includegraphics[width=.4\textwidth]{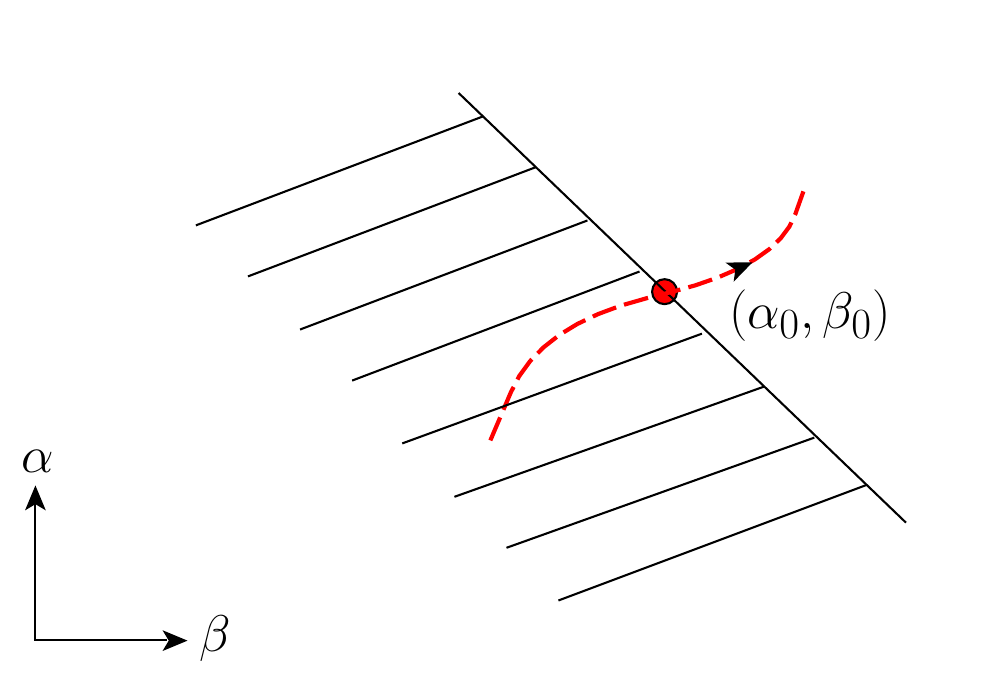}
\caption{The crossing of $(\a,\beta)$ through $(\a_o,\beta_o)\in  \LL$.}
 \label{F:crossing}  
 \end{figure}

\noindent Let $c:[\lam_-,\lam_+]\subset \br\to \br^2$ be a parametrization of the crossing curve such that $c(\lam_-)=(\a_-,\beta_-)$, $c(\lam_o)=(\a_o,\beta_o)$ and $c(\lam_+)=(\a_+,\beta_+)$. 
Then the initial bifurcation problem becomes a bifurcation problem around $ \lam_o$. More precisely, we have a $\Gamma$-equivariant map  $F:\br\times \br^n\to \br^n$ such that $F(\lam,0)=0$ for all $\lam\in [\lam_-,\lam_+]$. The spectrum of $D_xF(\lam,0)$ belongs to $\bc_-$ (the left half of the complex plane) for all $\lam\in [\lam_-,\lam_o]$, and as $\lam$ crosses $\lam_o$, the spectrum of $D_xF(\lam_o,0)$ intersects with $i \br$ at $0$.

Without loss of generality, let $\lam_\pm=\pm 4$ and $\lam_o=0$. Define a box around the bifurcation point $(0,0)\in \br\times \br^n$  (cf. Figure \ref{F:box}):
\[\Omega_1:=\{(\lam,x)\,:\, |\lam|<4,\q\|x\|<\rho\},\]
where $\rho>0$ is such that $F(\pm 4,\cdot)$ is a homeomorphism on $\{x\in \br^n\,:\, \|x\|<\rho\}$. 
\begin{figure}
 \centerline{ 
 \includegraphics[width=.32\textwidth]{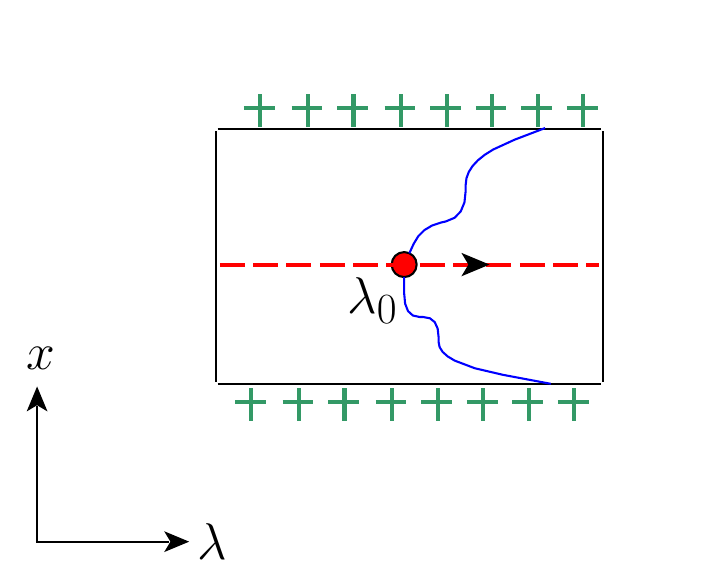}
}     
 \caption{An isolating box $\Omega_1$ around the bifurcating point $\lam=\lam_o$, where the red line is the equilibrium, the blue curves are potential bifurcating solutions and the plus signs ``$+$'' are the signs of auxiliary function $\zeta_1$.}
 \label{F:box}  
 \end{figure}
Without loss of generality, let $\rho=2$. Define $\FF_1:\overline\Omega_1\to \br \times \br^n$ by
\[\FF_1(\lam,x):=\big(|\lam|(\|x\|-2)+\|x\|-1, F(\lam,x)\big):=(\zeta_1(\lam,x),F(\lam,x)).\]
Note that $\zeta_1>0$ for $\|x\|=2$ and $\zeta_1<0$ for $\|x\|=0$. Functions with this property are called {\it auxiliary functions} on $\Omega_1$. Thus, by construction, zeros of $\FF_1$ in  $\Omega_1$ are contained properly in $\Omega_1$, i.e. $\FF_1^{-1}(0)\cap \overline\Omega_1 \subset\Omega_1$,  and if $\FF_1(\lam,x)=0$, then $x\ne 0$. In other words, zeros of $\FF_1$ correspond precisely to non-trivial zeros of $F$ in $\Omega_1$. The bifurcation invariant $\omega_0$ is defined by
\[\omega_0=\gmdeg(\FF_1,\Omega_1).\]
To compute $\omega_0$, we perform several homotopies on $\FF_1$. Define  $\FF_2:\overline\Omega_1\to \br\times \br^n$ by
$$\FF_2(\lam,x):=\big(|\lam|(\|x\|-1)+\|x\|+1, F(\lam,x)\big):=(\zeta_2(\lam,x), F(\lam,x)).$$
Since $\zeta_2>0$ for $\|x\|=2$, we have $\FF_1$ and $\FF_2$ are homotopic on $\Omega_1$ by a linear homotopy. {Thus, by homotopy invariance, we have $\gmdeg(\FF_1,\Omega_1)=\gmdeg(\FF_2,\Omega_1)$.} Also, $\zeta_2>0$ for $|\lam|\le\frac 12$. Thus, {all} zeros of $\FF_2$ in $\Omega_1$ are contained in the following subset of $\Omega_1$:
\[\Omega_2:=\{(\lam,x)\,:\, \frac 12<\lam <4,\, \|x\|<2\}{\subset \Omega_1}. \]
{In other words, $\FF_2$ does not have zeros in $\Omega_1\setminus \Omega_2$. By (the double negation of) the existence property, we have 
$$\gmdeg(\FF_2,\Omega_1\setminus \Omega_2)=0.$$
It then follows that $$\gmdeg(\FF_2,\Omega_1)=\gmdeg(\FF_2,\Omega_2)+\gmdeg(\FF_2,\Omega_1\setminus \Omega_2)=\gmdeg(\FF_2,\Omega_2).$$} 

Moreover, $\FF_2$ is homotopic to $\FF_3:\overline\Omega_2\to \br\times \br^n$ defined by
$$\FF_3(\lam,x):=(\zeta_2(\lam,x), D_xF(\lam,0)).$$
Decompose $\br^n$ into the sum of the critical eigenspace and the eigenspaces of the rest (all negative) eigenvalues of $D_xF(\lam_o,0)$, say $\br^n=R_0\times R_1$. Then, for $x=(x_1,x_2)\in R_0\times R_1$, the linear map $D_xF(\lam,0)(x_1,x_2)$ is homotopic to $(\lam x_1, -x_2)$. Thus, $\FF_3$ is homotopic to $\FF_4:\overline\Omega_2\to \br\times \br^n$ defined by
$$\FF_4(\lam,x):=(\zeta_2(\lam,x), (\lam x_1, -x_2)),\q\text{for } x=(x_1,x_2)\in R_0\times R_1 .$$
Note that $\FF_4(\lam,x)=0$ only if $x=0$. Substituting $x=0$ into $\zeta_2(\lam,x)$, we have
 $\zeta_2(\lam,0)=0$ if and only if $\lam=\pm 1$. That is,
 \[\FF_4^{-1}(0)\cap \Omega_2=\{(-1,0),(1,0)\}.\]
 It follows that 
$$\gmdeg(\FF_2,\Omega_2)=\gmdeg(\FF_4,\Omega_2)=\gmdeg(\FF_4,\NN_{-1})+\gmdeg(\FF_4,\NN_1),$$
where $\NN_{-1}$ (resp. $\NN_1$) is a small neighborhood of $(-1,0)$ (resp. $(1,0)$). 
On $\NN_{-1}$, we have that $\FF_4$ is homotopic to $(1+\lam, -x_1,-x_2)$. By suspension, we obtain
\[\gmdeg(\FF_4,\NN_{-1})=\gmdeg(-\id,B_1( \br^n)),\]
where $B_1(\cdot)$ denotes the unit ball. On the other hand, $\FF_4$ is homotopic to $(1-\lam, x_1,-x_2)$ on $\NN_1$, so by multiplication, we have
 \[\gmdeg(\FF_4,\NN_{1})=-\gmdeg(-\id,B_1( R_1)).\]
 Therefore,
 \[\omega_0=\gmdeg(-\id,B_1( \br^n))-\gmdeg(-\id,B_1( R_1)).\]
 Using $\text{\tt showdegree}$, it is expressed as
 \[\omega_0=\text{\tt showdegree}[\Gamma](n_0,n_1, \dots,n_r,1,0,\dots,0)-\text{\tt showdegree}[\Gamma](u_0,u_1, \dots,u_r,1,0,\dots,0),\]
 where $n_i$'s and $u_i$'s are defined by (\ref{eq:Vi_ni})--(\ref{eq:ui}).

The statement then follows from the existence property of degree.

\end{proof}
\vs
\ack{FMA acknowledges the support of the European Union within its 7th Framework Programme (FP7/2007-2013) under grant agreement no.~318723 (MatheMACS). The work of HR was supported by the Deutsche Forschungsgemeinschaft under grant DFG RU 1748/2-1.}
\vs  


\end{document}